\documentclass[11pt,a4paper,reqno,french,tikz]{amsart}
\usepackage[utf8]{inputenc} \usepackage[T1]{fontenc}
\usepackage{amsthm,amsmath,amsfonts,amssymb,amsxtra,appendix,bookmark,dsfont,bm,mathrsfs,amstext,amsopn,mathrsfs,mathtools,comment,cite,hyperref,color,xcolor,cite}
\usepackage[margin=4cm]{geometry}
\usepackage{tikz}
\usetikzlibrary{arrows.meta,calc}

\newtheorem{theorem}{Theorem}[section]\newtheorem{lemma}[theorem]{Lemma}\newtheorem{proposition}[theorem]{Proposition}\newtheorem{remark}[theorem]{Remark}

\let\C\relax\newcommand{\C}{\mathbb{C}}\newcommand{\Z}{\mathbb{Z}}\newcommand{\R}{\mathbb{R}}\newcommand{\N}{\mathbb{N}}\newcommand{\bbL}{\mathbb{L}}

\newcommand\cA{\mathcal{A}}\newcommand\cB{\mathcal{B}}\newcommand\cC{\mathcal{C}}\newcommand\cD{\mathcal{D}}\newcommand\cE{\mathcal{E}}\newcommand\cF{\mathcal{F}}\newcommand\cG{\mathcal{G}}\newcommand\cH{\mathcal{H}}\newcommand\cJ{\mathcal{J}}\newcommand\cL{\mathcal{L}}\newcommand\cM{\mathcal{M}}\newcommand\cP{\mathcal{P}}\newcommand\cQ{\mathcal{Q}}\newcommand\cR{\mathcal{R}}\newcommand\cS{\mathcal{S}}\newcommand\cT{\mathcal{T}}\newcommand\cU{\mathcal{U}}\newcommand\cV{\mathcal{V}}\newcommand\cW{\mathcal{W}}

\DeclareMathOperator{\ran}{Ran}\DeclareMathOperator{\Span}{Span}\DeclareMathOperator{\Ker}{Ker}\DeclareMathOperator{\dist}{dist}
\def\d{{\rm d}}

\renewcommand{\ge}{\geqslant}\renewcommand{\le}{\leqslant}

\newcommand{\pa}[1]{\left( #1 \right)}                              
\newcommand{\bpa}[1]{\big( #1 \big)}                                
\newcommand{\acs}[1]{\left\{ #1 \right\}}                       
\newcommand{\ab}[1]{\left|#1\right|}                                
\newcommand{\ps}[1]{\left< #1 \right>}                              
\newcommand{\nor}[2]{ \left| \! \left| #1 \right| \! \right|_{#2} } 
\newcommand{\norm}[1]{ \left| \! \left| #1 \right| \! \right| }     

\newcommand\vp{\varphi}                            
\newcommand{\ep}{\varepsilon} 
\let\p\relax\newcommand{\p}{\psi}                  
\newcommand{\na}{\nabla}                           

\newcommand{\f}[2]{\frac{#1}{#2}}                                              
\newcommand{\mymin}[1]{\underset{\substack{#1}}{\text{\normalfont{min}}}\;} 
\newcommand{\mysup}[1]{\underset{\substack{#1}}{\text{\normalfont{sup}}}\;} 
\newcommand{\mylim}[1]{\underset{\substack{#1}}{\text{\normalfont{lim}}}\;} 
\newcommand{\mat}[1]{\begin{pmatrix} #1 \end{pmatrix}}                         

\newcommand{\uz}[1]{U^0_{#1,k}}
\newcommand{\uo}[1]{U^1_{#1,k}}
\newcommand{\restr}[2]{#1_{\mkern 1mu \vrule height 2ex\mkern2mu #2}}

\newcommand{\fh}{\mathfrak{h}}
\def\fR{{\mathfrak R}}

\newcommand{\bbH}{\mathbb{H}}
\newcommand{\bbU}{\mathbb{U}}
\newcommand{\bbE}{\mathbb{E}}
\newcommand{\bbF}{\mathbb{F}}
\newcommand{\bbP}{\mathbb{P}}
\newcommand{\der}{D}
\newcommand{\fer}{_{\rm F}}
\newcommand{\per}{_{\rm per}}
\newcommand{\loc}{_{\rm loc}}

\newcommand{\xep}{\pa{\tfrac{y}{\ep}}}
\newcommand{\elk}{^{\ep}_k}

\newcommand{\sao}{\sqrt{\ab{\Omega}}}
\newcommand{\w}{w}

\newcommand{\tdr}{t}
\newcommand{\rdr}{r}
\newcommand{\sdr}{s}
\newcommand{\tur}{t'}
\newcommand{\rur}{r'}
\newcommand{\sur}{s'}
\newcommand{\chir}{\chi}
\newcommand{\gaur}{\gamma_1}
\newcommand{\gadr}{\gamma_2}

\newcommand{\vpp}{\widetilde{v}\fer}
\newcommand{\PP}{\mathbb{P}}

\usepackage{slashed}

\newcommand{\icol}[1]{
  \left(\begin{smallmatrix}#1\end{smallmatrix}\right)%
}

\usepackage{dsfont} 
\def\1{{\mathds{1}}}

\usepackage{subcaption}

\usepackage{graphicx}
\graphicspath{
  {home/louis/heavy_stuff/graphs_code/improved_graphene/plot_bands/}
  {pics/}
}

\title[Graphene superlattice effective operators from VPT]{Some effective operators\\for graphene monolayer superlattices,\\from variational perturbation theory}
\author[L. Garrigue]{Louis Garrigue}
\address{Laboratoire ``analyse géométrie modélisation'', CY Cergy Paris Université, 95302 Cergy-Pontoise, France}
\email{louis.garrigue@cyu.fr}
\date{\today}
\begin{document}
\maketitle

\begin{abstract}
Our goal is to provide precise effective operators for monolayer graphene at Fermi energy. We consider the microscopic potential created by a lattice, and add a macroscopic potential with the same periodicity but varying at a scale $\varepsilon^{-1} \in \mathbb{N}$, creating a superlattice. Our approach combines a variational approximation with perturbation theory and a multiscale method. At the effective level the usual massless Dirac operator is enriched by other operators. We provide simulations for graphene to show the increased accuracy, and a theoretical result proving the improved validity time scale in the dynamical case.
\end{abstract}

\section{Introduction}%
\label{sec:Introduction}

In this document, we consider semiclassical Schrödinger operators of the form
\begin{align}\label{eq:exact_ep_inv_omega} 
	h^\ep := \tfrac{1}{2} \pa{-i\na + \ep A(\ep x)}^2 + v(x) + \ep V(\ep x),
\end{align}
where $v$, $V$ and $A$ are $\Omega$-periodic real potentials having the same periodicity. Throughout the document we assume
\begin{align*}
\ep^{-1}=n\in\N\setminus\{0\},
\end{align*}
so that the full electromagnetic potential is periodic on the supercell $n\Omega=\ep^{-1}\Omega$. When $\ep$ becomes small, the spectrum of~\eqref{eq:exact_ep_inv_omega} is numerically hard to reach because this periodicity cell grows like $\ep^{-1}$, hence one needs to develop effective models. We will be interested in the states at the interface between filled and unfilled bands, because they are the main ones determining the physical behavior of quantum systems. In particular, this applies to graphene considered as a one-body model such as the one obtained from Density Functional Theory (DFT).


The physics literature investigated those doubly periodic systems, called monolayer superlattices, enabling a better control on the propagation of electronic waves for graphene, with periodic electric potentials~\cite{BarPeeVas08,BreFer09,WanZhu10,LevKimBro12,DeaYouWan12,SilEng12,FerRodFal16,Fernandes23}, magnetic ones~\cite{RamVasMat08,BerKezLoz08,BerGumLoz08,DelMar09,DelMar11,YuaZhaZen11,LePhaNgu12,WuLiuChe17,BezLim20} for instance to compute the mini-band structures~\cite{ParYanSon08b,ParYanSon08,IsaJonKin08,BreFer09,BalJorNil10,LiDieFor21,WuKilPar12}, or transport properties~\cite{BliFreSav09,BurYeyBre11,CheKraDan20}. The general construction in Sections~\ref{sec:intro}--\ref{sec:The effective operators} and the proof of Proposition~\ref{prop:weak_conv} retain the periodic magnetic vector potential $A$. Starting in Section~\ref{sec:Particular effective models}, and in all simulations, we set $A=0$ in order to focus on the electric superlattice. Notice also that a periodic vector potential has zero magnetic flux through one cell, since the integral of $\operatorname{curl}A$ over $\Omega$ vanishes by periodicity. It is well-known that the effective operator modeling particles around the Fermi level, for only one valley, is then a massless Dirac operator in the potential $V$, i.e. it is
\begin{align*}
v\fer \sigma \cdot (-i\na) + V,
\end{align*}
where $\sigma := (\sigma_1,\sigma_2)$ are the Pauli matrices, and $v\fer \in \R_+$ is the graphene Fermi velocity, see for instance~\cite[(2)]{IsaJonKin08}. Notice that in some articles, authors start from $\tfrac{1}{2} (-\Delta) + v(x) + V(\ep x)$ and derive $v\fer \sigma \cdot (-i\na) + \f{1}{\ep} V(x)$ but this involves a singular potential $\f{1}{\ep}V$ so here we rather start from~\eqref{eq:exact_ep_inv_omega} and at the end one can formally replace $V \rightarrow \f{1}{\ep}V$. This approach of modulating  monolayer graphene using a parametrized larger scale potential provides an alternative to moiré systems to obtain precise control over the electronic behavior. Moreover, a macroscopic periodic modulation of monolayer graphene is a first step before a better understanding of bilayer moirés~\cite{CarMasFan17}, relaxed graphene (see \cite[p58]{LanLif86}, \cite[eq. (45)]{Naumis20}, \cite[Chapter 9, p213]{Katsnelson20}) or two untwisted microscopic lattices with different lattice constants~\cite{KinUchMil12,OrtYanBri12}.

The mathematics literature also contains works on the study of~\eqref{eq:exact_ep_inv_omega}, or in the scaling 
\begin{align}\label{eq:GMS_ham} 
\tfrac{1}{2} \pa{-i\na + A(\ep x)}^2 + v(x) + V(\ep x).
\end{align}
In particular, Gérard--Martinez--Sjöstrand~\cite{GerMarSjo91} give, for~\eqref{eq:GMS_ham}, a general, rigorous reduction of the spectral problem for periodically perturbed operators to a finite system of semiclassical pseudodifferential operators with operator-valued symbols. Their framework yields operator-level asymptotic expansions in powers of $\ep$. For the scaling used here, in which the macroscopic perturbations carry an additional factor $\ep$, their construction may be combined with a Taylor expansion of the operator-valued symbols in $\ep$. This general method is less explicit than the finite reduced matrices computed below, but it supplies the appropriate mathematical framework and may also help analyze the spectral-pollution phenomenon discussed in Section~\ref{sub:The cutoff}. Related semiclassical constructions appear in~\cite{Buslaev87,GuiRalTru88,Dimassi93}. Dynamical effective models were studied in~\cite{GerMarMau97,PanSpoTeu03,FefWei14,BalCazMas26}, and a spectral approach in the strong-binding regime was developed in~\cite{FefLeeWei18}. The case of twisted bilayer graphene received a lot of attention, derivations of effective operators for twisted bilayer graphene began with the Bistritzer-MacDonald model~\cite{BisDon11}, which was mathematically studied in~\cite{BecEmbWit22,BecHumZwo23,BecHumZwo24,BecZwo24}. Other derivations were done in~\cite{WatKonMac23,CanMen23,CanGarGon23b,QuaWatMas25,KonLiuLus24}, and in~\cite{QuiKonLus25} for higher-order models. To develop effective models in multiscale situations, reduced basis methods were applied in homogenization~\cite{Nguyen08,Boyaval08,HerOliHue14,OhlSch15,ChuEfeHou23}, and ideas from multiscale finite element methods~\cite{HouWuCai99,EfeHou09,BabLip11} were applied to semiclassical Schrödinger operators~\cite{CheMaZha20,LiZha25}. See Appendix~\ref{sec:Literature} for more detail about the existing results.


In~\cite{CanGarGon23b}, the authors used a variational approximation which produced an effective operator $v\fer \sigma \cdot (-i\na) + \tfrac{1}{2} \ep (-\Delta) + V$. The variational space consisted in functions of the form $\sum_{a=1}^{2} \alpha_a(y) \phi_a \xep$, where $\alpha$ is let free and $\phi_1$ and $\phi_2$ are the Bloch eigenfunctions at the Dirac point corresponding to the Fermi energy. The variational space is hence free at the macroscopic level but imposed at the microscopic level. 
In this document, we use the same approach but also couple it to perturbation theory. To the best of our knowledge, the coupling between the variational approximation and perturbation theory was first used in~\cite{NooLow74}. It is now used in the physics literature~\cite{FraHeIps18,DugEksFur24}, and a mathematical analysis of this method is provided in~\cite{GarSta24}. To develop an improved variational space for~\eqref{eq:exact_ep_inv_omega}, we enrich the microscale functions by not only inserting the Bloch eigenfunctions at the Dirac point, but also their derivatives with respect to the momentum parameter $k$. We hence do not obtain a $2\times 2$-valued matrix operator but an $N_{\cF} \times N_{\cF}$-valued one, where $N_{\cF}$ is the number of functions in the microscale basis. For instance, corresponding to order one in perturbation theory, one can add $4$ derivatives (two perturbation directions for two states), resulting in a $6 \times 6$-matrix valued effective operator presented in Section~\ref{sub:Order 1, independent on k}.


We introduce the problem in Section~\ref{sec:intro}, then in Section~\ref{sec:The effective operators} we present the general effective operators. In Section~\ref{sec:Particular effective models} we take particular cases linked to momentum Taylor expansions of the Dirac eigenmodes. Then in Section~\ref{sec:Improved validity time} we show that our new effective operators provide improved validity time for the dynamical approximation of wavepackets propagating at Dirac momentum. The time scale of validity is $\ep^{-2(\ell +1) + \eta}$ with $\eta >0$ when we include perturbation theory up to order $\ell$, extending~\cite{FefWei14} to $\ell \ge 1$. Finally in Section~\ref{sec:Simulations} we display some simulations for the band diagrams, showing that the effective operators derived in this document provide more accurate band diagrams and eigenvectors than the traditional massless Dirac operator. We regroup all our conclusions in Section~\ref{sub:Conclusion}. We recall in Appendix~\ref{sec:Literature} the mathematical results obtained in the literature which can be applied to the problem of deriving approximations for multiscale graphene.







\section{Introduction of the setting}%
\label{sec:intro}



In this section we introduce the objects that will be manipulated, most importantly we present the scaled exact operator.

\subsection{The lattice}%
\label{sub:The lattice}

We consider the two-dimensional space $\R^2$ instead of $\R^3$ in which graphene is embedded, but everything can be extended to the $\R^3$ situation. Let us consider $a_0 > 0$, $k_{\rm D} := \frac{4\pi}{3a_0}$, and the graphene lattice vectors
\begin{align}\label{eq:def_lattice} 
	a_1 &:= a_0\mat{\frac{1}{2} \\ -\frac{\sqrt{3}}{2}}, \qquad  & a_2 := a_0\mat{\frac{1}{2} \\ \frac{\sqrt{3}}{2}}, \\
	a_1^* &:= \sqrt{3} k_{\rm D}\mat{\frac{\sqrt{3}}{2} \\ -\frac{1}{2}}, \qquad  & a_2^* :=  \sqrt{3} k_{\rm D}\mat{\frac{\sqrt{3}}{2} \\ \frac{1}{2}}.
\end{align}
We define the action of $\R^2$ on $a_1$ and $a_2$ as $n a := n_1 a_1 + n_2 a_2$ and similarly $m a^* := m_1 a^*_1 + m_2 a^*_2$. 
The direct lattice is $\bbL := \{n a \;|\; n \in \Z^2\}$, the dual one is thus $\bbL^* := \{m a^* \;|\; m \in \Z^2\}$. The fundamental cell is $\Omega := \{x a \;|\; x \in [0,1]^2\}$ and we denote by $\cB$ the first Brillouin zone. With the lattice convention in~\eqref{eq:def_lattice}, we consider the Dirac point
\begin{align*}
K := -\tfrac{1}{3} (a_1^* + a_2^*) = (-k_{\rm D},0) \in \cB,
\end{align*}
which is a vertex of the first Brillouin zone. The opposite vertex $-K$ is the second Dirac point, the two are not equivalent modulo $\bbL^*$, and the choice matters for one of the coefficients computed in Section~\ref{sub:Order 1, independent on k}.

\subsection{Scaling}%
\label{sub:Scaling}

For any $n \in \N\setminus\{0\}$ and any $q \in [1,+\infty[$, we define the set of periodic functions
\begin{align*}
	L^q\per(n\Omega) := \acs{f \in L^q\loc(\R^2) \;|\; \forall x \in \R^2, \forall y \in \bbL, f(x + ny) = f(x)}.
\end{align*}
For any $n \in \N\setminus\{0\}$ and any $f,g \in L^2\per(n\Omega)$, we define $\ps{f,g}_{L\per^2(n\Omega)} := \int_{n\Omega} \overline{f}g$. In the supercell problem, $n=\ep^{-1}\in\N\setminus\{0\}$. Moreover, we will use mainly
\begin{align*}
\ps{\cdot,\cdot} := \ps{\cdot,\cdot}_{L^2\per(\Omega)}.
\end{align*}
We define the scaling operator $Q$ on functions $f : \R^2 \rightarrow \R$ by 
\begin{align*}
(Qf)(x) := f\pa{\tfrac{x}{\ep} },
\end{align*}
its inverse is $\pa{Q^{-1} f}(x) = f(\ep x)$. We have $Q : L^2\per(\ep^{-1}\Omega) \rightarrow L^2\per(\Omega)$ and for $\ep \in 1/\N := \{1/n, n \in \N \backslash \{0\}\}$, $\ps{Qf,Qg}_{L\per^2(\Omega)} = \ep^2 \ps{f,g}_{L\per^2(\ep^{-1} \Omega)}$ so in those spaces, $Q^* = \ep^2 Q^{-1}$.

\subsection{Symmetries of the potentials}%
\label{sub:Symmetries of the potentials}

Let us define the rotation matrix 
\begin{align*}
R_\theta := \mat{\cos \theta & - \sin \theta \\ \sin \theta & \cos \theta},
\end{align*}
and its corresponding action $(\cR_\theta f)(x) := f(R_{-\theta} x)$ on functions $f : \R^2 \rightarrow \R$. We consider potentials with honeycomb symmetry, which is defined in~\cite[Definition 2.1]{FefWei12}. They are potentials $v \in L^p(\R^2)$, with $p>1$, such that 
\begin{align}\label{eq:honeycomb} 
v\in L^p\per(\Omega), \quad  v(-x) = v(x), \quad  \cR_{\f{2\pi}{3} }v = v, \quad v(x_1,-x_2) = v(x_1,x_2).
\end{align}
As for the potentials $V$ and $A$, we only assume that $V,A \in L^p\per(\Omega)$ but for simulations of Section~\ref{sec:Simulations} we will choose $A=0$ and $V$ respecting the symmetries~\eqref{eq:honeycomb}. 

\subsection{The scaled exact operator}%
\label{sub:The scaled exact operator}

\subsubsection{The microscopic exact problem}%
\label{ssub:The microscopic exact problem}

For any $q \in \R^2$, we define the covariant derivative $\na_q := e^{-iqx}\na e^{iqx}=\na+iq$, where $e^{iqx}$ is considered as a multiplication operator, equivalently, $-i\na_q=-i\na+q$.

Let us denote by $(E^q_m, u^q_m) \in \R \times H^1\per(\Omega)$ the Bloch eigenpairs of 
\begin{align*}
	h_q :=\tfrac{1}{2} (-i\na_q)^2 + v,  \qquad \bbL-\text{periodic},
\end{align*}
which is bounded from below, for $m \in \N$ and $q \in \cB$, sorted such that $E^q_m$ is increasing with $m$. Hence 
\begin{align}\label{eq:schro_V0} 
h_q u^q_m =  E^q_m u^q_m .
\end{align}
We assume that there is $m\fer \in \N$ such that the dispersion relation of $k \mapsto h_k$ has a conical intersection at the Dirac point $k = K$ at the level $m\fer$. This situation is generic, as proved in~\cite[Theorem 5.1]{FefWei12}. In particular, this implies that $E^K_{m\fer} = E^K_{m\fer +1} =: E\fer$, which we call the Fermi level. We also assume that this level is exactly 2-fold degenerate, so $E^K_{m\fer -1} < E^K_{m\fer} = E^K_{m\fer +1} < E^K_{m\fer+2}$. 

Let us define 
\begin{align}\label{eq:w1w2} 
\w_1 := u^K_{m\fer}, \qquad \qquad \w_2 := u^K_{m\fer+1}
\end{align}
such that $\nor{\w_1}{L^2\per(\Omega)} = \nor{\w_2}{L^2\per(\Omega)} = 1$ and such that with 
\begin{align}\label{eq:def_phi} 
\phi_a(x) := e^{iKx}\w_a(x),
\end{align}
 we have 
\begin{align}\label{eq:syms_phi} 
\cR_{\f{2\pi}{3} } \phi_a = \omega^a \phi_a, \qquad \qquad  \text{for all } a \in \{1,2\},
\end{align}
where $\omega := e^{i\f{2\pi}{3}}$, see~\cite{FefWei12}. The functions $\phi_1$ and $\phi_2$ are the Bloch states at the Dirac point and Fermi level, see~\cite{FefWei12} for more details about the symmetries of those Bloch states. 

\subsubsection{The exact operator}%
\label{ssub:The exact operator}

We want to study the operator $h^\ep$ defined in~\eqref{eq:exact_ep_inv_omega}, which is $\ep^{-1}\bbL$-periodic when $\ep^{-1}\in\N$. Instead of $h^\ep$, we work on its Bloch transform, acting on $L^2\per(\ep^{-1} \Omega)$ and given by
\begin{align}\label{eq:def_h_ep_q} 
h^\ep_q = e^{-iqx} h^\ep e^{iqx} = \tfrac 12 (-i\na_q + \ep Q^{-1} A)^2 + v + \ep Q^{-1} V.
\end{align}
We have $\na_q Q = \ep^{-1} Q \na_{\ep q}$ so $\ep Q^{-1}\na_{\ep ^{-1} q}  = \na_{ q} Q^{-1}$, and
\begin{align*}
(-i\na_q + \ep Q^{-1} A)^2 Q^{-1} = \ep^2 Q^{-1} \pa{-i\na_{\ep^{-1} q} + A}^2.
\end{align*}
Instead of $h^\ep_q$ we work on its rescaled (and hence semiclassical) version $Q h^\ep_q Q^{-1}$. We center the study of $h^\ep$ around the Fermi level, which we assume to be at the Dirac point $K$, presenting a conical intersection in the dispersion relation, see~\cite{FefWei12} for a precise definition. Hence we work on
\begin{align}\label{eq:exact_op} 
	H\elk := Q \pa{h^\ep_{K + \ep k} - E\fer} Q^{-1} = \tfrac{1}{2} \ep^2 \pa{-i\na_{\ep^{-1} K + k} + A}^2 + Qv + \ep V - E\fer,
\end{align}
which acts on $L^2\per(\Omega)$. The two operations ``scaling'' and ``Bloch transform'' commute, hence the order in which we apply them does not matter.


\subsubsection{Exact eigenvectors when $V=0$}%
\label{ssub:Exact eigenvectors}

In this subsection only, we assume $A=0$ in addition to $V=0$. Let us denote by $H^{\ep,V=0}_k$ the resulting operator $H\elk$. Its band diagram is simply a folding of the band diagram of $-\f{1}{2} \Delta + v$ (see~\cite{KuBerLee10,Farjam14,MayYndSol20,QuaRybSch25} for the folding/unfolding operations), and the eigenmodes of $H^{\ep,V=0}_k$ can be deduced from the ones of $\tfrac{1}{2}(-i\na_p)^2+v$. Let us analyze the branches starting at the Dirac point. We compute that $H^{\ep,V=0}_k Q u^{K+\ep k}_m = (E^{K+\ep k}_m -E\fer) Q u^{K+\ep k}_m$, so $\pa{E^{K+\ep k}_m -E\fer, Q u^{K+ \ep k}_m}_{m \in \N}$ are eigenmodes for $H^{\ep,V=0}_k$. Moreover, if $V$ is considered as a perturbation, the eigenmodes of $H^{\ep,V=0}_k$ are smoothly deformed to the ones of $H^\ep_k$.

\section{The general effective operators}%
\label{sec:The effective operators}

In this section we present effective operators for two-scales graphene, in a general way. More specific choices will be made in Section~\ref{sec:Particular effective models}.

\subsection{Two-scales reduced basis}%

Let us consider a family $\cF$ of $N_{\cF}$ linearly independent and periodic functions $\psi_1,\dots,\psi_{N_{\cF}} \in D(h_K)\subset H^1\per(\Omega)$,
\begin{align}\label{eq:families} 
\cF := \pa{\psi_j}_{1 \le j \le N_{\cF}}.
\end{align}
This family will impose the behaviors of eigenfunctions of the reduced space at the microscopic scale. They potentially depend on the momentum $k$, but not on $\ep$. As in~\cite{CanGarGon23b}, we then build the reduced vector space
\begin{align}\label{eq:proj_space} 
 \acs{ \begingroup\textstyle\sum\endgroup_{j=1}^{N_{\cF}}  \alpha_j(y) \psi_j\xep \;|\; \alpha_j \in \cC^{\infty}\per(\Omega), j\in \{1, \dots, N_{\cF}\}},
\end{align}
which is the two-scale space ``spanned'' by the space of microscopic functions of $\cF$ when the macroscopic behavior is let free. Efficient families will 
\begin{itemize}
\item have $N_{\cF}$ small so the resulting effective model will be ``computable'',
\item choose $\p_j$ in a way so that the resulting effective operator will be accurate.
\end{itemize}
We then define the operator
	\begin{align*}
		\cJ : \hspace{0.5cm} 
	\begin{array}{rcl}
		\C^{N_{\cF}} \otimes \cC^{\infty}\per(\Omega,\C) & \longrightarrow &  H^1\per(\Omega,\C) \\
		\alpha & \longmapsto & \sum_{j=1}^{N_{\cF}} \alpha_j(x) \psi_j \pa{\tfrac{x}{\ep} },
	\end{array}
	\end{align*}
	which enables to go from envelop functions to real-space functions. For any $Z \in \C^{N_{\cF}}$, we define
\begin{align*}
JZ :=\sum_{j=1}^{N_{\cF}}Z_j\psi_j.
\end{align*}

\subsection{Effective operator}%
\label{sub:Effective operator}

We study the action of the exact operator $H\elk$ in the projected space~\eqref{eq:proj_space}. We want to associate it a low-dimensional model which will reproduce the Fermi eigenmodes with precision. Let us define 
\begin{align*}
	\cM &:= J^* \pa{h_K - E\fer} J = \pa{\ps{ \psi_a, \pa{h_K - E\fer} \psi_b}}_{1 \le a,b \le N_{\cF}}, \\
	\cS &:= J^* J = \pa{\ps{\psi_a, \psi_b}}_{1 \le a,b \le N_{\cF}}, \\
	\cL &:= J^* (-i\na_K) J = \pa{\ps{\psi_a, (-i\na_K) \psi_b}}_{1 \le a,b \le N_{\cF}},
\end{align*}
where the coefficients of $\cL$ belong to $\C^2$. Those $N_{\cF} \times N_{\cF}$ matrices are all hermitian, $\cS$ is a Gram matrix, so since $\cF$ is composed of linearly independent functions, $\cS$ is strictly positive. Then an effective model for $H\elk$ is
\begin{align}\label{eq:effective_op_general} 
	\bbH\elk := \ep^{-1} \cM \otimes \1 + \cL \cdot\otimes (-i\na_k + A)  +   \cS \otimes \pa{\tfrac{1}{2}\ep(-i\na_k + A)^2 + V}
\end{align}
as an operator of $\C^{N_{\cF}} \otimes \cC^{\infty}\per(\Omega)$, where the only spatial dependence exists via $V$ and $A$. The coefficients of the first-order term and of the Laplacian come from the exact identity
\begin{align}\label{eq:exact_ritz_identity}
    J^*(h_{K+k}-E\fer)J
    =\cM+\cL\cdot k+\tfrac12|k|^2\cS.
\end{align}
Its weak-form consistency with $H\elk$ is given by the following proposition.
\begin{proposition}[Weak convergence]\label{prop:weak_conv} 
	Take $\alpha, \beta \in \C^{N_{\cF}} \otimes \cC^{\infty}\per(\Omega)$, and $V,A \in \cC^{\infty}\per(\Omega)$. For any $N \in \N$, there exists $C_{N,\alpha,\beta} > 0$ such that for any $\ep \in ]0,1[ \cap (1/\N)$,
\begin{align*}
	\ab{\ep^{-1} \ab{\Omega} \ps{\cJ \alpha , H\elk   \cJ \beta }_{L^2\per(\Omega,\C)}  - \ps{\alpha , \bbH\elk  \beta}_{\C^{N_{\cF}} \otimes L^2\per(\Omega,\C)}} \le C_{N,\alpha,\beta} \ep^N.
\end{align*}
\end{proposition}
A proof is given in Section~\ref{sec:proof of prop}. The estimate is only a statement about the displayed sesquilinear forms for fixed smooth $\alpha$ and $\beta$. In particular, it does not assert $O(\ep^\infty)$ convergence of spectra or eigenvectors. Moreover, one can see in the proof that $C_{N,\alpha,\beta} \le C_N \nor{\alpha}{H^{q_N}} \nor{\beta}{H^{q_N}}$ where $C_N$ is independent on $\alpha$ and $\beta$, and $q_N \in \N \cup \{0\}$ is an exponent. 

The eigenvalue equation for $\bbH\elk$ is then
\begin{align*}
\bbH\elk \alpha\elk =\bbE\elk \cS \alpha\elk,
\end{align*}
where the generalized eigenpairs of $\bbH\elk$ will be denoted by $\pa{\bbE\elk, \alpha\elk}$. For the branches identified and tracked numerically in Section~\ref{sec:Simulations}, we compare the effective and exact eigenpairs through
\begin{align*}
\ep^{-1} \pa{\cE\elk - \cE^{\ep}_0} \simeq \bbE\elk - \bbE^\ep_0, \qquad \qquad \Phi\elk \simeq \cJ \alpha\elk \textup{ up to a phase factor}.
\end{align*}
Here $\simeq$ denotes the observed numerical agreement between corresponding tracked branches, it is not an error estimate. The eigenvalue discrepancies and phase-optimized eigenvector distances are measured directly in Section~\ref{sec:Simulations}.  The maps $(k,V) \mapsto \cE\elk$ and $(k,V) \mapsto \bbE\elk$ are locally smooth away from eigenvalue crossings. So to detect which level of the effective model corresponds to which one of the exact model, we start by identifying the eigenmodes of $H^{\ep,V=0}_0$ and $\bbH^{\ep,V=0}_0$ (see Section~\ref{ssub:Exact eigenvectors}), which eigenvalues then vanish, and then we let $k$ and $V$ change smoothly to see the deformations of eigenvalues and keep track of them even at crossings.

\begin{remark}[Undo the Bloch transform]
When the coefficient matrices $\cM, \cL$ and $\cS$ do not depend on $k$, the operator in real space can be written as the inverse of the Bloch transform $\bbH^{\ep}  = e^{ikx} \bbH\elk e^{-ikx}$ and does not depend on $k$.
\end{remark}

\begin{remark}[Effective operators are defined for any $\ep >0$]
While the exact operator $H\elk$ is only defined for $\ep \in 1/\N$, the effective operators $\bbH\elk$ are defined for any $\ep > 0$. Moreover, the computation cost of the effective operators is essentially constant in $\ep$ while it was exponentially increasing with $1/\ep$ in the exact model.
\end{remark}

\subsection{First simplification}%
\label{sub:First simplification}

We recall the Pauli matrices
\begin{align*}
\sigma_1 := \mat{0 & 1 \\ 1 & 0}, \qquad \sigma_2 := \mat{0 & -i \\ i & 0}, \qquad \sigma_3 := \mat{1 & 0 \\ 0 & -1}.
\end{align*}
In our choices of $\cF$, for any $k$ we will always have $\psi_a = \w_a$ for $a \in \{1,2\}$, where $\w_a$ are the Bloch states at the point $K$, defined in Section~\ref{ssub:The microscopic exact problem}. We will also always have $\w_a \perp \psi_b$ for $a \in \{1,2\}$ and $b \in \{3,\dots,N_{\cF}\}$. Hence the operators $\cM$, $\cL$ and $\cS$ can be simplified as
\begin{align}\label{eq:M_L_S} 
	\cM = \mat{0 & 0 \\ 0 & M}, \qquad \cL = \mat{v\fer \sigma  & T  \\ T^*  & L}, \qquad \cS = \mat{\1 & 0 \\ 0 & S}.
\end{align}
where
\begin{align}\label{eq:def_matrices} 
	S &:= \pa{\ps{\psi_{a+2}, \psi_{b+2}}}_{1 \le a,b \le N_{\cF}-2}, \quad M := \pa{\ps{ \psi_{a+2}, \pa{h_K - E\fer} \psi_{b+2}}}_{1 \le a,b \le N_{\cF}-2} \nonumber \\
	T& := \pa{\ps{\psi_{a}, (-i\na_K) \psi_{b+2}}}_{\substack{1 \le a \le 2 \\  1 \le b \le N_{\cF}-2}}, L := \pa{\ps{\psi_{a+2}, (-i\na_K) \psi_{b+2}}}_{1 \le a,b \le N_{\cF}-2}.
\end{align}
Moreover, $v\fer \in \R_+$ is defined via $\ps{w_1, (-i\na_K) w_2} = v\fer \pa{1,  -i}^T$, see~\cite[Supplementary Material, Appendix I]{CanGarGon23b}.

\section{Particular effective models}%
\label{sec:Particular effective models}

From now on we take $A = 0$.

\subsection{Preliminary information about the unscaled system}%
\label{sub:Preliminary information about the unscaled system}

In this section, we present particular choices for the families $\cF$ defined in~\eqref{eq:families}, which will implement a multiscale variational perturbation theory.

\subsubsection{A family of Bloch eigenmodes with special symmetry}%
\label{ssub:A family of Bloch eigenmodes with special symmetry}

From~\cite{FefWei12,BerCom18}, the symmetries of $\w_a$ presented in Section~\ref{ssub:The microscopic exact problem} imply that
\begin{align}\label{eq:vF} 
	\ps{\w_a, (-i\na_K) \w_a} = 0, \; a \in \{1,2\}, \quad \qquad \ps{\w_1, (-i\na_K) \w_2} = v\fer \icol{1 \\ -i}.
\end{align}
Those computations are enough to derive the effective massless Dirac operator $v\fer \sigma \cdot (-i\na)$ from $-\f 12 \Delta + v$ where $v$ has honeycomb symmetry.

\subsubsection{Pseudo-inverse}%
\label{ssub:Pseudo-inverse}

Let us denote by $P$ the orthogonal projector onto the eigenspace $\Ker \pa{h_K-E\fer} = \Span\pa{w_1,w_2}$, and $P_\perp := 1-P$. An important quantity which we will need is the pseudo-inverse
\begin{align*}
	R := \left\{
	\begin{array}{ll}
	\pa{(E\fer - h_K)_{\mkern 1mu \vrule height 2ex\mkern2mu P_\perp L^2\per(\Omega) \rightarrow P_\perp L^2\per(\Omega)}}^{-1} & \mbox{on } P_\perp L^2\per(\Omega) \\
	0 & \mbox{on } P L^2\per(\Omega)
	\end{array}
	\right.
\end{align*}
extended on $L^2\per(\Omega)$ by linearity, so $(E\fer - h_K) R = R (E\fer - h_K) = P_\perp$.

\subsection{Variational perturbation theory : Dirac modes and their derivatives up to order $\ell$}%
\label{sub:Dirac modes and their derivatives}

The first way of building families $\cF$ comes from variational perturbation theory~\cite{NooLow74,FraHeIps18}, which was analyzed mathematically in~\cite{GarSta24}. It consists in building a variational space using the derivatives of the eigenfunctions. Here we build $\cF$ using $u^q_{m\fer}$ and their derivatives with respect to $q$. In Appendix~\ref{sec:Degenerate perturbation theory our problem} we present how one can differentiate those eigenstates with respect to $q$ at $q = K$.

For any $k \in \R^2 \backslash \{0\}$, we will use the differentiation operator in direction $\f{k}{\ab{k} } $ denoted by
\begin{align*}
\der_k := \tfrac{k}{\ab{k} } \cdot (-i\na_K ),
\end{align*}
where we recall that $\na_K = \na + iK$. We use the notation $\partial_{K,j} = \pa{\na_K}_j = \partial_j + iK_j$. Labeling the first index of $\cF$ by the perturbative order, we define, for any direction $k \in \R^2 \backslash \{0\}$,
\begin{align}\label{eq:cFs} 
	\cF_0 & := \pa{\w_1, \w_2 }, \qquad \cF_{1,k} := \pa{\w_1, \w_2, R \der_k \w_1, R \der_k \w_2}, \\
	\cF_1 &:= \pa{\w_1, \w_2, R (-i\partial_{K,1}) \w_1, R (-i\partial_{K,2}) \w_1, R (-i\partial_{K,1}) \w_2, R (-i\partial_{K,2}) \w_2}, \nonumber\\
	\cF_{2,k} &:= \pa{\w_a, R \der_k \w_a,  (R \der_k)^2 \w_a - R^2 \der_k P \der_k \w_a}_{1 \le a \le 2} \nonumber\\
	\cF_2 &:= \pa{\w_a, R (-i\partial_{K,b}) \w_a ,  R^2 (-i\partial_{K,b}) \w_a,  R (-i\partial_{K,b}) R (-i\partial_{K,c}) \w_a}_{1 \le a,b,c \le 2}.\nonumber
\end{align}
Remark that the vectors of the families $\cF_{\ell,k}$ depend on $k$ only via the direction $\f{k}{\ab{k}} $. We have $\ab{\cF_0} = 2$, $\ab{\cF_{1,k}} = 4$, $\ab{\cF_1} = \ab{\cF_{2,k}} = 6$, $\ab{\cF_{2}} = 18$. As explained in Appendix~\ref{sec:Degenerate perturbation theory our problem}, each $\cF_{\ell}$ and $\cF_{\ell,k}$ contain the directional derivatives of $q \mapsto u^{q}_{m\fer}$ up to order $\ell$ at $q = K$, and Section~\ref{sec:Formal derivation of the perturbation series for the macroscopic functions} gives a justification of this choice. More precisely, those spaces are built such that for $\ell \in \{0,1,2\}$,
\begin{align*}
	u^{K + k}_{m\fer} = U + O_{\ab{k}\to 0 }\bpa{\ab{k}^{\ell +1} },\qquad  \text{ where } U \in \Span \cF_{\ell,k} \subset \Span \cF_\ell.
\end{align*}

Following the method presented in Section~\ref{sub:Building reduced spaces from perturbation theory}, one can similarly build families $\cF_{\ell,k}$ at all perturbative orders $\ell$.

\subsection{Order 0}

The family $\cF_0$ is the choice made in~\cite[Remark 1]{CanGarGon23b}. The effective model consists essentially of the massless Dirac operator, with the lower-order correction $\tfrac{1}{2}\ep(-i\na_k)^2$, and those two operators commute,
\begin{align*}
\bbH \elk = v\fer \sigma \cdot (-i\na_k) + V + \tfrac{1}{2} \ep (-i\na_k)^2.
\end{align*}



\subsection{Order 1, $k$-dependent}%
\label{sub:Order 1}

We want to compute the matrices $S$, $M$, $L$ and $T$ defined in~\eqref{eq:def_matrices} and which are sufficient to describe the effective operator.

For any $q = (q_1,q_2) \in \R^2$, we define $q_\C := q_1 + i q_2$. For any $k\in \cB \backslash \{0\}$, let us define $\theta_k \in [0,2\pi[$ such that $k_\C = \ab{k} e^{i\theta_k}$. Let us define
\begin{align}\label{eq:def_trs} 
\begin{array}{ll}
\tdr := \nor{R (-i\partial_{K,1}) w_1}{}^2 , & \tur := \ps{(-i\partial_{K,1}) w_1, R (-i\partial_{K,1}) w_1}, \\
 \rdr := \ps{R (-i\partial_{K,1}) w_1, R (-i\partial_{K,1}) w_2}, & \rur := \ps{(-i\partial_{K,1}) w_1, R (-i\partial_{K,1}) w_2}, \\
 \sdr := -i\ps{R (-i\partial_{K,2}) w_1, R (-i\partial_{K,1}) w_1}, & \sur := -i\ps{(-i\partial_{K,2}) w_1, R (-i\partial_{K,1}) w_1}, \\
\end{array}
\end{align}
and
\begin{align}\label{eq:def_cgg} 
	\begin{array}{l}
 \chir := \ps{R (-i\partial_{K,1}) w_1, (-i\partial_{K,1}) R (-i\partial_{K,1}) w_1},  \\
 \gaur := \ps{R (-i\partial_{K,1}) w_1, (-i\partial_{K,2}) R (-i\partial_{K,2}) w_2},  \\
 \gadr := \ps{R (-i\partial_{K,2}) w_1, (-i\partial_{K,1}) R (-i\partial_{K,2}) w_2}.
\end{array}
\end{align}

\begin{proposition}[Effective matrices for $\cF_{1,k}$]\label{prop:order1_kdep} 
	In the case $\cF_{1,k}$, we have that $\rdr, \tur ,\rur, \sdr, \sur , \chir, \gaur, \gadr\in \R$ and the matrices defined in~\eqref{eq:def_matrices} are
\begin{align}\label{eq:expr_SM_k} 
	S= \mat{\tdr & \rdr e^{i 2\theta_k} \\ \rdr e^{-i 2\theta_k} & \tdr}, \qquad M = -\mat{\tur & \rur e^{i 2\theta_k} \\ \rur e^{-i 2\theta_k} & \tur}
\end{align}

\begin{align}\label{eq:expr_T_k} 
T = \mat{\pa{\tur \1 + \sur \sigma_2} \tfrac{k}{\ab{k} } & \rur e^{i\theta_k} {\tiny{\mat{1 \\ i}}} \\ \rur e^{-i\theta_k} {\tiny{\mat{1 \\ -i}}} & \pa{\tur \1 - \sur \sigma_2} \tfrac{k}{\ab{k}} },
\end{align}
and using $\sigma = (\sigma_1,\sigma_2)$,
\begin{align}\label{eq:expr_L_k} 
	L = \mat{      
\chir \mat{\cos(2 \theta_k) \\ -\sin (2\theta_k)}
&
2 \gaur e^{-i\theta_k} \mat{ \cos \theta_k \\  \sin \theta_k}
\\
2 \gaur e^{i\theta_k}\mat{\cos \theta_k \\  \sin \theta_k}
&
\chir \mat{\cos(2 \theta_k) \\ -\sin (2\theta_k)}
	} + \gadr \sigma.
\end{align}
\end{proposition}
The computations are provided in Section~\ref{sub:Matrices in the case of k}. A drawback is that the effective matrices are not defined for $k = 0$.

\subsection{Order 1, $k$-independent}%
\label{sub:Order 1, independent on k}

We choose $\cF_1$, its vectors in the order presented in~\eqref{eq:cFs}, that is $\psi_1 = \w_1$, $\psi_2 = \w_2$, $\psi_3 = R (-i\partial_{K,1}) \w_1$, $\psi_4 = R (-i\partial_{K,2}) \w_1$, $\psi_5 = R (-i\partial_{K,1}) \w_2$, $\psi_6 = R (-i\partial_{K,2}) \w_2$.

We use the matrix
\begin{align*}
D := \mat{1 & i \\ i & -1}=\sigma_3+i\sigma_1.
\end{align*}

\begin{proposition}[Effective matrices for $\cF_{1}$]\label{prop:order1_kindep} 
In the case $\cF_1$, with the same constants as in~\eqref{eq:def_trs} and~\eqref{eq:def_cgg}, which are still real, the matrices defined in~\eqref{eq:def_matrices} are
\begin{align*}
 S = \mat{\tdr \1 + \sdr \sigma_2 & \rdr D \\ \rdr D^* & \tdr \1 - \sdr \sigma_2}, \qquad M = -\mat{\tur \1 + \sur \sigma_2 & \rur D \\ \rur D^* & \tur \1 - \sur \sigma_2}
\end{align*}
and $T = (T^1,T^2)$ and $L = (L^1,L^2)$ with
\begin{align*}
T^1 = \mat{ \tur & -i\sur & \rur & i \rur \\ \rur & -i\rur & \tur & i\sur}, \qquad T^2 =  \mat{i\sur & \tur & i\rur & -\rur \\ -i\rur & -\rur & -i\sur & \tur}
\end{align*}
and
\begin{align*}
	L^1 = \icol{
		\chir & 0 & 2\gaur + \gadr & -i\gaur \\
		0 & -\chir & -i \gaur & \gadr \\
		2 \gaur + \gadr & i\gaur & \chir & 0 \\
		i\gaur & \gadr & 0 & -\chir
	}, \qquad 
L^2 = \icol{
		0 & -\chir & -i\gadr & \gaur \\
		-\chir & 0 & \gaur & -i\pa{2\gaur + \gadr} \\
		i \gadr & \gaur & 0 & -\chir \\
		\gaur & i \pa{2\gaur + \gadr} & -\chir & 0
}.
\end{align*}
\end{proposition}
The proof is similar to the one of Proposition~\ref{prop:order1_kdep} and uses Lemmas~\ref{lem:D} and~\ref{lem:three_derivatives}.

\subsection{Values of the parameters for graphene}%
\label{sub:Values of the parameters for graphene}

We computed these effective parameters with DFTK~\cite{DFTKjcon}, using DFT Bloch functions $u_m^K$ for standard physical graphene. The computation is performed in atomic units. Let $E_{\rm h}$ and $a_0$ denote the Hartree energy and the Bohr radius, respectively. Restoring dimensions in the definitions~\eqref{eq:def_trs} and~\eqref{eq:def_cgg}, we have
\begin{align}\label{eq:num_vals} 
	\begin{aligned}
	\frac{E_{\rm h}\rdr}{a_0^2} = -30\,\textup{eV}, \qquad
	\frac{E_{\rm h}\tdr}{a_0^2} &= 32\,\textup{eV}, \qquad
	\frac{E_{\rm h}\sdr}{a_0^2} = 29\,\textup{eV}, \\
	\frac{\rur}{a_0^2} = 13\,\textup{eV}, \qquad
	\frac{\tur}{a_0^2} &= -15\,\textup{eV}, \qquad
	\frac{\sur}{a_0^2} = -11\,\textup{eV}, \\
	\frac{\chir}{a_0^3} \simeq -2\cdot10^{-2}\,\textup{eV}, \qquad
	&\frac{\gaur}{a_0^3} = -1.2\,\textup{eV}, \qquad
	\frac{\gadr}{a_0^3} = 0.6\,\textup{eV}.
	\end{aligned}
\end{align}
The numerical value of $\chir$ is not very precise. In this convention, the Dirac slope is $\frac{v\fer}{a_0} = 11\,\textup{eV}$.




\subsection{Smallness of $\chi$}%
\label{sub:Smallness of chi}

The tight-binding approximation replaces the continuum wave function by amplitudes on localized atomic orbitals and retains hopping between lattice sites~\cite{Wallace47,FefLeeWei18,FefWei20}.

To define the nearest-neighbour model in our lattice convention, set $\delta:=(a_2-a_1)/3$. The two sublattices are $\bbL$ and $\bbL+\delta$, and each site $n\in\bbL$ has neighbours $n+\delta$, $n+a_1+\delta$ and $n-a_2+\delta$. Writing the amplitudes as $Z=(Z_A,Z_B)\in\ell^2(\bbL,\C^2)$, where $Z_B(n)$ is the amplitude at the physical site $n+\delta$, define the shift operator $T_a$ by $(T_a f)(n):=f(n+a)$ for $a,n\in\bbL$. With hopping strength $\tau>0$ and zero on-site energy, the Hamiltonian is
\begin{align}\label{eq:tb_real_space}
H_{\rm TB}=-\tau\mat{0&1+T_{a_1}+T_{-a_2}\\1+T_{-a_1}+T_{a_2}&0}.
\end{align}
 In the periodic Bloch gauge associated with the cell index $n$, its Bloch matrix is
\begin{align}\label{eq:tb_bloch_matrix}
H_{\rm TB}(q)=-\tau\mat{0&F(q)\\\overline{F(q)}&0},
\qquad \text{where } F(q):=1+e^{iq\cdot a_1}+e^{-iq\cdot a_2}.
\end{align}
The sublattice operator is $\sigma_3$, and it satisfies $\sigma_3 H_{\rm TB}(q)\sigma_3=-H_{\rm TB}(q)$.

Let $P_{\rm TB}$ be the orthogonal projection onto $\Ker H_{\rm TB}(K)$, let $R_{\rm TB}$ be the inverse of $-H_{\rm TB}(K)$ on $\ran(1-P_{\rm TB})$, extended by zero on $\ran P_{\rm TB}$, and set $H^1_{\rm TB}:=\left.\partial_{q_1}H_{\rm TB}(q)\right|_{q=K}$. The tight-binding analogue of $\chi$ is $\chir_{\rm TB}:=\ps{u_{\rm TB},H^1_{\rm TB}u_{\rm TB}}$, where $u_{\rm TB}:=R_{\rm TB}H^1_{\rm TB}w_{1,\rm TB}$ and $w_{1,\rm TB}$ is a normalized zero mode. Since $\sigma_3$ preserves the zero eigenspace, this mode can be chosen with $\sigma_3 w_{1,\rm TB}=\eta w_{1,\rm TB}$, $\eta\in\{\pm1\}$. The identities $\sigma_3 H^1_{\rm TB}\sigma_3=-H^1_{\rm TB}$ and $\sigma_3 R_{\rm TB}\sigma_3=-R_{\rm TB}$ then give $\sigma_3 u_{\rm TB}=\eta u_{\rm TB}$, and hence
\begin{align*}
\chir_{\rm TB}=\ps{\sigma_3 u_{\rm TB},\sigma_3 H^1_{\rm TB}u_{\rm TB}}
=-\ps{u_{\rm TB},H^1_{\rm TB}u_{\rm TB}}=-\chir_{\rm TB},
\end{align*}
hence $\chir_{\rm TB} = 0$. In the minimal two-band model~\eqref{eq:tb_bloch_matrix}, $P_{\rm TB}=1_{\C^2}$ and $R_{\rm TB}=0$, so this vanishing is immediate, the same symmetry argument applies to chiral tight-binding models with additional orbitals, where the reduced inverse need not vanish. The continuum/DFT Hamiltonian has no exact sublattice symmetry of this kind, and effective same-sublattice hoppings can break it. Thus the tight-binding cancellation gives a qualitative interpretation of the small numerical value of $\chir$.

\subsection{Using excited states}%
\label{sub:Variational space from excited states}

We can also form families $\cF$ of microscopic functions from Bloch eigenstates at the Dirac point. For $n \in \N$, $n \ge 3$, we will study
\begin{align}\label{eq:cF_exc} 
\cF^{n} := \pa{u_a^K}_{a \in \{m\fer , \dots, m\fer + n-1\}}.
\end{align}
We have $\ab{\cF^n} = n$ and remark that $\cF^2 = \cF_0$. We recall that the usual Bloch eigenfunctions studied in~\cite{FefWei12} are $w_a = u_{m\fer+a-1}^K$ for $a \in \{1,2\}$. We define $w_3 := u_{m\fer +2}^K$ and we only explicitly write the case $\cF^3 = \{w_1,w_2,w_3\}$. In general, the energy $E^3 := E^K_{m\fer + 2} = \ps{w_3, h_K w_3}$ is different to the Fermi energy $E\fer$, and $\phi_3(x) := e^{iKx} \w_3(x)$ can be defined as in~\eqref{eq:def_phi} and let us assume that $\phi_3$ is in a different representation of $\cR_{\f{2\pi}{3}}$ than $\phi_1$ and $\phi_2$ so assume that $\cR_{\f{2\pi}{3}} \phi_3 = \phi_3$. The only additional parameter is $\vpp := -i \ps{w_1, (-i\partial_{K,1}) w_3}$.

We define the parity and complex conjugation operators $\pa{\cP f}(x) := f(-x)$, $\pa{\cC f}(x) := \overline{f}(x)$. 

\begin{proposition}[Effective matrices for $\cF^3$]\label{prop:excited} 
In the case $\cF^3$, assume that $\cC \cP \phi_3 = \phi_3$ and that $\phi_3$ has a definite reflection parity $\mathfrak{R}\phi_3=\eta_3\phi_3$, where $\eta_3\in\{\pm1\}$ satisfies $\eta_{\mathfrak R}\eta_3=-1$. Then the matrices defined in~\eqref{eq:def_matrices} are $S=1$, $M = E^3 - E\fer$, $L = 0$ and
\begin{align*}
	T = -\vpp \mat{\icol{ -i \\ 1 } \\ \icol{ i \\ 1 }},
\end{align*}
where $\vpp  \in \R$.
\end{proposition}
The computations are provided in Section~\ref{sub:Case of excited states}. We see that one advantage of this kind of choice is that the effective matrices are simple and involve less constants to compute.

\subsection{Exact modes at each $k$}
\label{sub:exact each k}

Another natural choice is to take the family $\cF = \{u^k_{m\fer}, u^k_{m\fer +1}\}$, it has the benefit to be an exact choice when $V=0$, so the eigenmodes of the effective operators give some of the exact ones. But the drawback is that in this case one needs to solve $h_k$ at each $k$ in the Brillouin zone. We will not explore this choice in the numerical investigation.


\subsection{Error and limit of high $N_{\cF}$}%
\label{sub:Limit of high}

Let us denote by $P_{\cJ}$ the orthogonal projection onto $\ran \cJ$, and $P_{\cJ}^\perp := 1 - P_{\cJ}$. We recall that we denoted by $\Phi\elk$ the eigenvector of the exact operator $H\elk$ defined in~\eqref{eq:exact_op}. As presented in~\cite[Chapter II, Section 8, (8.46)]{BabOsb91} or in~\cite[Proposition 3.2]{GarSta24}, the error between $\Phi\elk$ and the approximation $\cJ \alpha\elk$ (produced by our effective operators) is controlled by $P_\cJ^\perp \Phi\elk$. Hence, usually, the larger $\ran \cJ$ the smaller $P_\cJ^\perp \Phi\elk$ and the smaller the error $\nor{\Phi\elk - \cJ \alpha\elk}{L^2\per(\Omega)}$. So making $\cF$ larger improves the approximation but produces effective operators that are more expensive to deal with.

When $N_{\cF}$ is large, one can expect that formally, $\ran \cJ$ becomes close to
\begin{align*}
	 &\acs{ \begingroup\textstyle\sum\endgroup_{j=1}^{N_{\cF}}  \alpha_j(y) \vp_j\xep \;|\; \alpha_j \in \cC^{\infty}\per(\Omega), \vp_j \in H^1\per(\Omega) , j\in \{1, \dots, N_{\cF}\}} \\
	&\qquad \simeq \cC^{\infty}\per(\Omega) \otimes H^1\per(\Omega),
\end{align*}
If the increasing microscopic families are complete in the relevant $H^1_{\rm per}(\Omega)$ sector, the associated Galerkin spaces have the same dense limiting space and their projection errors tend to zero. It is in this formal Galerkin sense that approximations built from excited states and those built from momentum derivatives become equivalent as $N_{\cF}\to +\infty$. No rate with respect to $N_{\cF}$ is asserted here.

\subsection{Schur reduction to a $2\times 2$ matrix-valued operator}%
\label{sub:Reduction to a matrix}

We recall the formalism of Schur's reduction in Appendix~\ref{sec:Schur complement}. We formally apply it to the effective operator $\bbH\elk$, where the Schur projection $P_2$ is the projection onto the first two components of $\C^{N_{\cF}}$, the computations are detailed in Section~\ref{sub:Schur complement applied here}. We define the $2 \times 2$ matrix-valued operator
\begin{multline}\label{eq:schur_operator} 
\cH\elk :=\1 \otimes V + v\fer \sigma \otimes \cdot (-i\na_k) \\
+ \ep \1 \otimes \pa{\tfrac{1}{2} (-i\na_k)^2 - T \cdot(-i\na_k) (M^{-1} T^*) \cdot(-i\na_k) },
\end{multline}
acting on $\C^2 \otimes L^2\per(\Omega)$, where more precisely, for any $i \in \{1,2\}$ and any $\p \in \C^2 \otimes L^2\per(\Omega)$,
\begin{multline*}
\pa{T \cdot(-i\na_k) (M^{-1} T^*) \cdot(-i\na_k) \psi}_i \\
	= \sum_{\substack{1 \le a,b,j \le 2}} \pa{ T^a M^{-1} \bpa{T^b}^*}_{ij} \pa{-i\na_k}_a \pa{-i\na_k}_b \psi_j.
\end{multline*}
The operator $\cH\elk$ is similar to the one derived in~\cite{CanGarGon23b}, but is has a corrected second-order differential operator. It is the operator $\bbH\elk$ on which we applied a Schur reduction, as explained in Appendix~\ref{sec:Schur complement}. 
Hence we ``reduced'' the operator $\bbH\elk$, which had divergences via the term in $\ep^{-1}$, to $\cH\elk$, which has no divergence, i.e. no term in $\ep^{-1}$. A drawback of the Schur reductions is that one needs to apply an operator to the eigenvectors of $\cH\elk$ to obtain an approximation of the exact eigenvectors, see Appendix~\ref{sec:Schur complement}.

\section{Improved validity time}%
\label{sec:Improved validity time}

\subsection{Definitions}%
\label{sub:Definitions vt}

In this section and in the proof of Theorem~\ref{thm:main_derivative_enriched_validity_time_scale}, we assume in addition that $v\in L^\infty\per(\Omega)$, we set $A=0$ and $V=0$. We regard the corresponding operators as acting on the whole space $\R^2$, rather than on a supercell, so that $\ep$ is an arbitrary parameter of $]0,1[$ and no longer needs to belong to $1/\N$. We write
\begin{align*}
h:= -\tfrac12\Delta+v.
\end{align*}
Since $v$ is bounded, $D(h_q)=H^2\per(\Omega)$ for any $q$, and the graph norm of $h_{K+k}$ is uniformly equivalent to the $H^2\per(\Omega)$ one when $k$ ranges over a bounded set.

Let $\cF=(\psi_j)_{1\le j\le N_{\cF}}$ be a family as in~\eqref{eq:families} which does not depend on $k$, so that $\cM$, $\cL$ and $\cS$ are constant matrices and $\bbH\elk$ has constant coefficients. Among the families of~\eqref{eq:cFs}, this allows $\cF_0$, $\cF_1$ and $\cF_2$. As in Section~\ref{sub:First simplification}, we assume that
\begin{align}\label{eq:derivative_family_first_vectors}
\psi_1=\w_1,\qquad \psi_2=\w_2,\qquad \w_a\perp\psi_b \; \text{ for } a\in\{1,2\},\; 3\le b\le N_{\cF},
\end{align}
so that $\cM$, $\cL$ and $\cS$ have the block form~\eqref{eq:M_L_S}. Define
\begin{align*}
\iota(z_1,z_2):=(z_1,z_2,0,\ldots,0)^\top, \qquad \phi_j(x):=e^{iK\cdot x}\psi_j(x), \quad \forall j \in\{1,\ldots,N_{\cF}\}.
\end{align*}
We assume that the Rayleigh--Ritz compression at $K$ has no zero mode besides the two Dirac ones,
\begin{align}\label{eq:no_spurious_zero_mode}
\Ker\cM=\ran\iota.
\end{align}
By~\eqref{eq:M_L_S}, this amounts to the invertibility of $M$. In the case of $\cF_1$, using Proposition~\ref{prop:order1_kindep} together with $\sigma_2 D = D$ and $D^*D = 2(\1 - \sigma_2)$, a Schur complement computation gives
\begin{align}\label{eq:det_M_F1}
\det M = \pa{\tur - \sur}^2 \bpa{\pa{\tur + \sur}^2 - 4 \rur^2}.
\end{align}
The rounded values~\eqref{eq:num_vals} give $\det M = 0$ but a more accurate 3-digits numerical simulation shows that $\det M \neq 0$, so~\eqref{eq:no_spurious_zero_mode} is valid.

On $L^2(\R^2_X,\C^{N_{\cF}})$, the effective operator~\eqref{eq:effective_op_general} associated with $h^\ep$, written in real space, is
\begin{align}\label{eq:constant_potential_effective_operator}
\bbH^\ep
:=\ep^{-1}\cM
+\cL\cdot(-i\na_X)
+\tfrac12\ep\cS(-\Delta_X).
\end{align}

The effective equation is written in the slow time $T=\ep t$. For $|k|$ small, define
\begin{align*}
P(k):=\frac{1}{2\pi i}\oint_{\Gamma}
\bpa{z-h_{K+k}}^{-1}\,\d z,
\end{align*}
where $\Gamma$ is a positively oriented circle enclosing $E\fer$ and no other point of $\sigma(h_K)$. The map $k\mapsto h_{K+k}$ is an analytic family with common domain
$H^2\per(\Omega)$, so the Riesz formula shows that $k\mapsto P(k)$ is
real-analytic with values in
$\mathcal L\bpa{L^2\per(\Omega),H^2\per(\Omega)}$.  Set
\begin{align}\label{eq:vjj} 
    v_j(k):=P(k)\w_j,
\end{align}
Since $v_j(0)=\w_j$, after reducing the neighbourhood of $0$, the vectors $v_1(k),v_2(k)$ form a basis of $\ran P(k)$.

For any $a \in \mathscr S(\R^2,\C^{N_\cF})$, we denote the reconstruction map by
\begin{align}\label{eq:reconstruction_map}
    (\cG a)(x)
    :=\ep\sum_{j=1}^{N_{\cF}}a_j(\ep x)\phi_j(x).
\end{align}
It is the analogue of $\cJ$ in the unscaled variable and carrying the Bloch phase $e^{iKx}$, and for any normalized $a$, $\nor{\cG a}{L^2(\R^2)}$ is bounded away from $0$ and $+\infty$. Let $\mathfrak a\in\mathscr S(\R^2,\C^2)$ and set $\alpha^0:=\iota \mathfrak a$. 

\subsection{The result}%
\label{sub:The result}

We compare the solution $\alpha$ of the effective evolution
\begin{align}\label{eq:effective_evolution}
\begin{cases}
i\cS\partial_T\alpha=\bbH^\ep\alpha,\\
\alpha(X,0)=\alpha^0(X),
\end{cases}
\end{align}
with the solution $\varphi$ of
\begin{align}\label{eq:exact_evolution}
\begin{cases}
i\partial_t\varphi=h\varphi,\\
\varphi(\cdot,0)=\cG\alpha^0.
\end{cases}
\end{align}
Both equations are posed on $\R^2$, the first one in the slow variables $(X,T) = (\ep x, \ep t)$. Define the remainder 
\begin{align}\label{eq:def_remainder}
        r^\ep(\cdot,t) =e^{-it(h-E\fer)}\cG\alpha^0 -\cG\alpha(\cdot,\ep t),
    \end{align}
so $\varphi(\cdot,t)=e^{-iE\fer t}\bpa{\cG\alpha(\cdot,\ep t)+r^\ep(\cdot,t)}$, $r^\ep(\cdot,0)=0$, and $r^\ep$ measures the quality of the approximation.

\begin{theorem}[Validity time scales for derivative-enriched microscopic families]\label{thm:main_derivative_enriched_validity_time_scale}
Let $\ell\in\N\cup\{0\}$, take $A=0$ and $V=0$, and assume~\eqref{eq:derivative_family_first_vectors} and~\eqref{eq:no_spurious_zero_mode}. Consider the basis~\eqref{eq:vjj}, and assume that
\begin{align}\label{eq:derivative_enriched_family_assumption}
\partial_k^\beta v_j(0)\in\Span\cF,
\qquad \forall j\in\{1,2\},\quad 0 \le \ab{\beta}\le\ell.
\end{align}
Consider the error $r^\ep$~\eqref{eq:def_remainder} between the effective evolution~\eqref{eq:effective_evolution} and the exact one~\eqref{eq:exact_evolution}. Then there are $\ep_0>0$ and $C>0$, independent of $\ep$ and of $t_\ep$, such that for any $0<\ep\le\ep_0$ and any $t_\ep>0$, 
\begin{align}\label{eq:constant_potential_sharp_error}
\sup_{0\le t\le t_\ep}\nor{r^\ep(\cdot,t)}{H^2(\R^2)} \le C\bpa{\ep+\ep^{2(\ell+1)}t_\ep}.
\end{align}
\end{theorem}

A proof is presented in Section~\ref{sec:Proof of Theorem }. Assumption~\eqref{eq:derivative_enriched_family_assumption} is respected by the families $\cF_\ell$ of Section~\ref{sub:Dirac modes and their derivatives}. For any $\eta>0$, taking $t_\ep=\ep^{-2(\ell+1)+\eta}$ makes the right-hand side of~\eqref{eq:constant_potential_sharp_error} equal to $C(\ep+\ep^\eta)$, which tends to zero. For $\ell=0$, this was already obtained in~\cite{FefWei14}. Thus, when $V = 0$, enriching the microscopic family with momentum derivatives improves the time scale on which the resulting dynamics is valid.

\section{Simulations}%
\label{sec:Simulations}

In this section we are going to explore the way band diagrams vary with change of parameters, and the errors between the effective and exact models. The simulations are done using DFTK~\cite{DFTKjcon} but not taking the physical graphene units.

\subsection{The parameters}%
\label{sub:The parameters}

The lattice constant used in~\eqref{eq:def_lattice} will always be $a_0 = 5$. We will use the notation $ma^* := m_1 a_1^* + m_2 a_2^*$. We define $m^1 := \pa{1,0}$, $m^2 := \pa{0,-1}$, $m^3 := \pa{-1,1}$ and define the honeycomb potential
\begin{align*}
v_0(x) := 2 \sum_{i=1}^{3} \cos ( (m^i a^*) \cdot x ),
\end{align*}
having the symmetries~\eqref{eq:honeycomb}. The microscopic potential will always be the same, $v = 10 v_0$. Moreover, the macroscopic one $V$ will mostly be 
\begin{align}\label{eq:def_lambda_V_honeycomb} 
V_{\text{honeycomb}} = \lambda v_0, \qquad \qquad  \lambda \in \R.
\end{align}
Our ``standard'' choice for $\ep$ will be $\ep = \f 17$. The index of the lower Fermi level will always be $m\fer = 1$.

In the band diagrams, the momentum follows a triangular path in momentum space, it will go through $\kappa := (-2,1)a^*/3$ (which is a Dirac point), $\Gamma := (0,0)$ and $M := -(1,0)a^*$ as indicated on the horizontal axis, the path will be 
\begin{align}\label{eq:path_momentum} 
\kappa \to \Gamma \to M \to \kappa.
\end{align}
 The vertical axis will represent the energy differences. The origin of energies (i.e. the zero) is chosen such that when $V = 0$, the Fermi level is at zero, for the effective and exact systems.

On Figure~\ref{fig:zero} we display the path in momentum space, as well as the band diagram of the exact model when $V=0$ and $\ep = 1$ (no superlattice).
\begin{figure}[h!]
\begin{center}
\begin{subfigure}[b]{0.45\linewidth}
\centering
\vbox to 4.2cm{\vfill\begin{tikzpicture}[scale=0.6, >=Stealth]
  \pgfmathsetmacro{\G}{2.0} 
  \pgfmathsetmacro{\R}{\G/sqrt(3)} 
  
  \coordinate (Gamma) at (0,0);
  \coordinate (M) at ({- \G * sqrt(3)/2}, {\G/2});
  \coordinate (kappa) at ({-\G/(2*sqrt(3))}, {\G/2});
  
  \coordinate (C_a1) at ({\G * sqrt(3)/2}, {- \G/2});
  \coordinate (C_a2) at ({- \G * sqrt(3)/2}, {- \G/2});
  \coordinate (C_ma2) at ({\G * sqrt(3)/2}, {\G/2});
  \coordinate (C_ma1a2) at (0, {\G});
  \coordinate (C_a1a2) at (0, {- \G});
  
  \fill[blue!5] (0:\R) -- (60:\R) -- (120:\R) -- (180:\R) -- (240:\R) -- (300:\R) -- cycle;
  \fill[gray!8] ($(M) + (0:\R)$) -- ($(M) + (60:\R)$) -- ($(M) + (120:\R)$) -- ($(M) + (180:\R)$) -- ($(M) + (240:\R)$) -- ($(M) + (300:\R)$) -- cycle;

  \draw[blue!60!black, thick] (0:\R) -- (60:\R) -- (120:\R) -- (180:\R) -- (240:\R) -- (300:\R) -- cycle;
  
  \newcommand{\drawbz}[1]{
    \draw[gray!50, dashed, thin] (#1) +(0:\R) -- +(60:\R) -- +(120:\R) -- +(180:\R) -- +(240:\R) -- +(300:\R) -- cycle;
  }
  
  \drawbz{M}
  \drawbz{C_a1}
  \drawbz{C_a2}
  \drawbz{C_ma2}
  \drawbz{C_ma1a2}
  \drawbz{C_a1a2}

  \draw[->, thick, gray!80!black] (Gamma) -- (C_a1) node[below right, black] {$\bm{a}_1^*$};
  \draw[->, thick, gray!80!black] (Gamma) -- (C_ma2) node[above right, black] {$\bm{a}_2^*$};
  
  \coordinate (mid_kG) at ($(kappa)!0.55!(Gamma)$);
  \coordinate (mid_GM) at ($(Gamma)!0.55!(M)$);
  \coordinate (mid_Mk) at ($(M)!0.55!(kappa)$);
  
  \draw[very thick, red!85!black, ->] (kappa) -- (mid_kG);
  \draw[very thick, red!85!black] (mid_kG) -- (Gamma);
  
  \draw[very thick, red!85!black, ->] (Gamma) -- (mid_GM);
  \draw[very thick, red!85!black] (mid_GM) -- (M);
  
  \draw[very thick, red!85!black, ->] (M) -- (mid_Mk);
  \draw[very thick, red!85!black] (mid_Mk) -- (kappa);
  
  \fill[brown] (Gamma) circle (1.8pt) node[right, xshift=2pt, yshift=-2pt, text=brown] {$\Gamma$};
  \fill[brown] (M) circle (1.8pt) node[above left, text=brown] {$M$};
  \fill[brown] (kappa) circle (1.8pt) node[above, yshift=2pt, text=brown] {$\kappa$};
  
  \coordinate (K) at ({-\R}, 0);
  \fill[blue!80!black] (K) circle (1.8pt) node[below left, xshift=-2pt] {$K$};
  
  \node[blue!60!black] at (0, -0.4) {$\mathcal{B}$};
\end{tikzpicture}%
\vfill}
\caption{Path~\eqref{eq:path_momentum} in momentum space.}
\end{subfigure}
\hspace{0.5cm}
\begin{subfigure}[b]{0.35\linewidth}
\centering
\vbox to 4.2cm{\vfill\includegraphics[width=3cm,trim={0.35cm 0cm 0.34cm 0cm},clip]{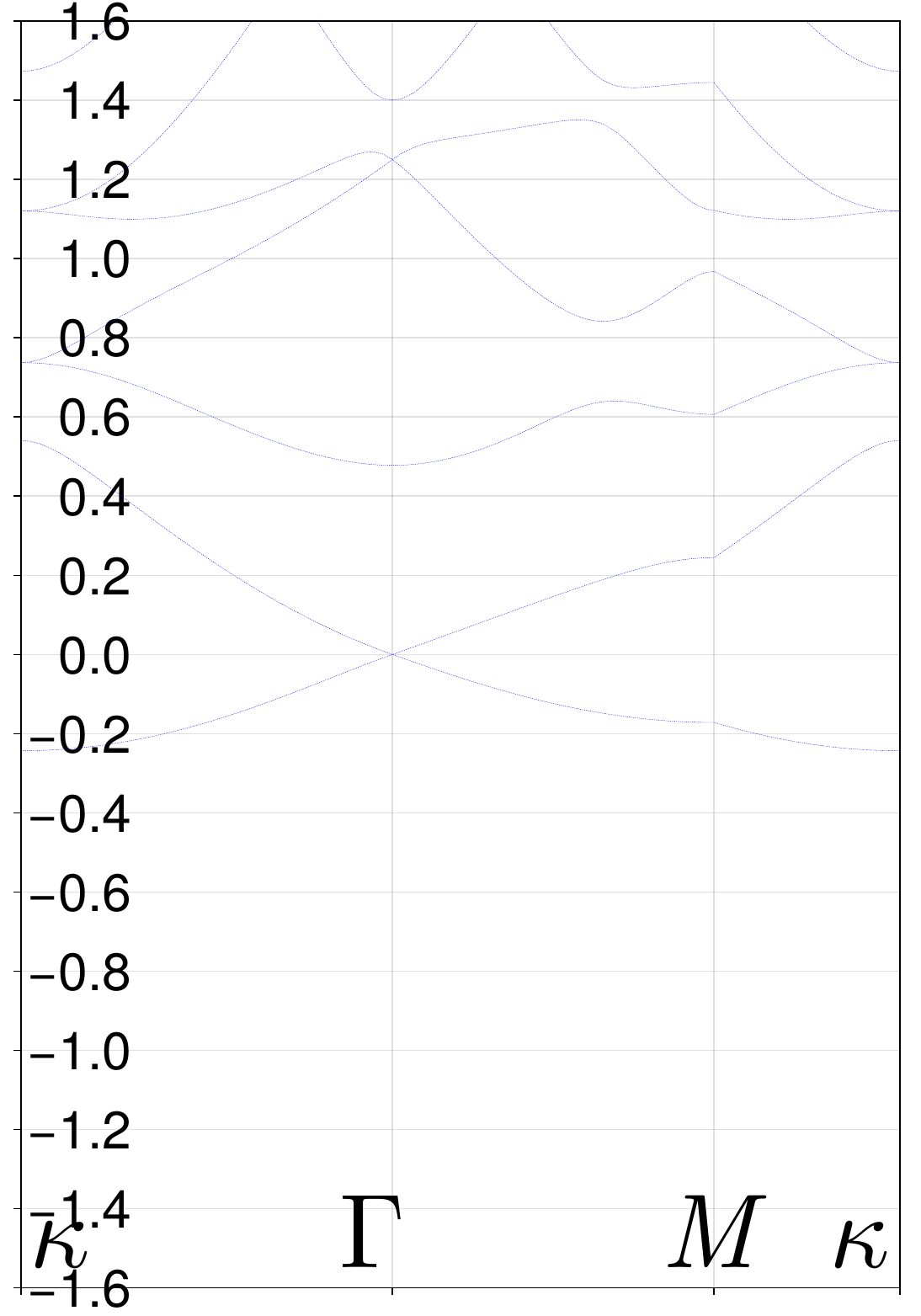}\vfill}
\caption{Band diagram when $V=0$ and $\ep=1$.}
\end{subfigure}
\caption{Path in momentum space and band diagram when $V=0$ and $\ep=1$.}\label{fig:zero}
\end{center}
\end{figure}

\subsection{The cutoff $\nu$}%
\label{sub:The cutoff}

\subsubsection{Necessity of the cutoff}%
\label{ssub:Necessity of the cutoff}

We introduce an integer $\nu \in \N$. For $k \in \Z^2$ we define plane waves $p_k(x) := e^{i x \cdot (k a^*)}$ and the set of $\Omega$-periodic functions being composed of momenta lower than $\nu$ (in $\ell^\infty$ norm), more precisely
\begin{align*}
\cC^{\infty}_{\text{per},\nu} := \Span \{ p_k \;|\; k \in \Z^2, \forall i \in \{1,2\}, \ab{k_i} \le \nu \}.
\end{align*}
We actually do not use $\bbH\elk$ defined in~\eqref{eq:effective_op_general} but $P_\nu \bbH\elk P_\nu$ where $P_\nu$ is the orthogonal projection onto $\C^{N_{\cF}} \otimes \cC^{\infty}_{\text{per},\nu}$. This is necessary because increasing $N_{\cF}$ adds undesired bands which pollute the spectrum, as we see on Figure~\ref{fig:increase_nu}. The band of interest is the one coming from the Dirac cone. This first series of band diagrams represents the effective (in black) and exact (in blue) bands for $\lambda = 0$, $\ep = \f{1}{7}$, with $\cF_0$, so the effective operator is $v\fer \sigma \cdot (-i\na) + \f 12 \ep (-\Delta)$. We see that as we increase $\nu$, additional spectral ``pollution'' arises, 
 but their number stabilizes after some value of $\nu$, for instance in the case of Figure~\ref{fig:increase_nu}, taking $\nu \ge 6$ does not change with respect to $\nu = 6$. The same phenomenon happens for other families $\cF$, and at any value of $\ep$.

Written differently, at $\ep$ fixed, when we increase $\nu$ we observe the spectral pollution as in Figure~\ref{fig:increase_nu}, due to the presence of the Laplace operator $-\Delta$. Moreover, at $\nu$ fixed, the spectral pollution disappears and the artificial bands are ``expelled'' as we decrease $\ep$, as we can see on Figure~\ref{fig:increase_invep_fixed_nu}.

In all the rest of the simulations, we will hence use $\nu = 2$, for any situation, it is a good compromise to obtain the sought-after bands, and small spectral pollution.

\begin{figure}[h!]
\begin{center}
\newgeometry{left=0.5cm,right=0.5cm}  
\hspace{-6cm}
\begin{subfigure}[b]{0.15\linewidth}\includegraphics[width=\linewidth,trim={0.35cm 0cm 0.34cm 0cm},clip]{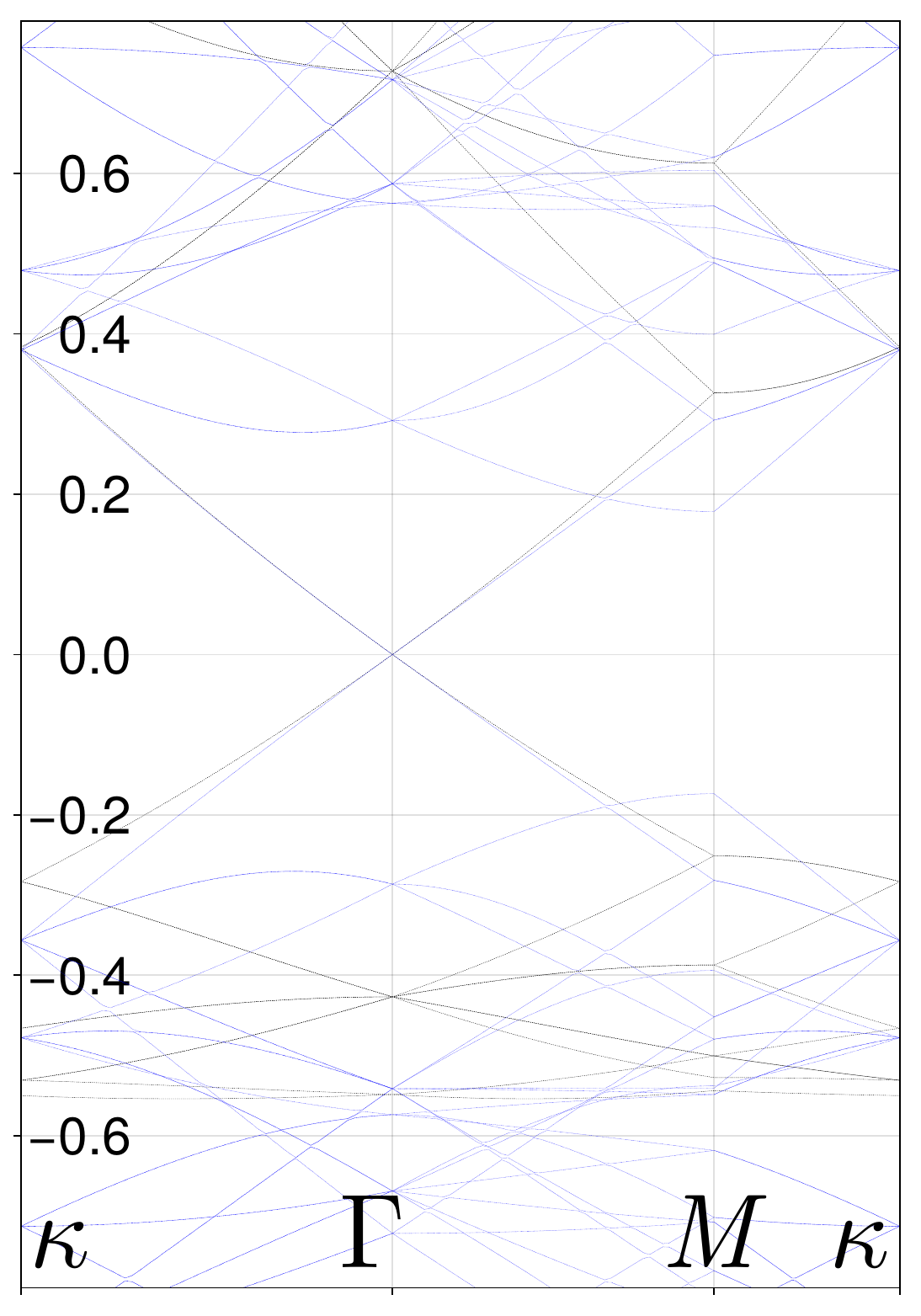}\caption{$\nu = 1$}\end{subfigure}
\begin{subfigure}[b]{0.15\linewidth}\includegraphics[width=\linewidth,trim={0.35cm 0cm 0.34cm 0cm},clip]{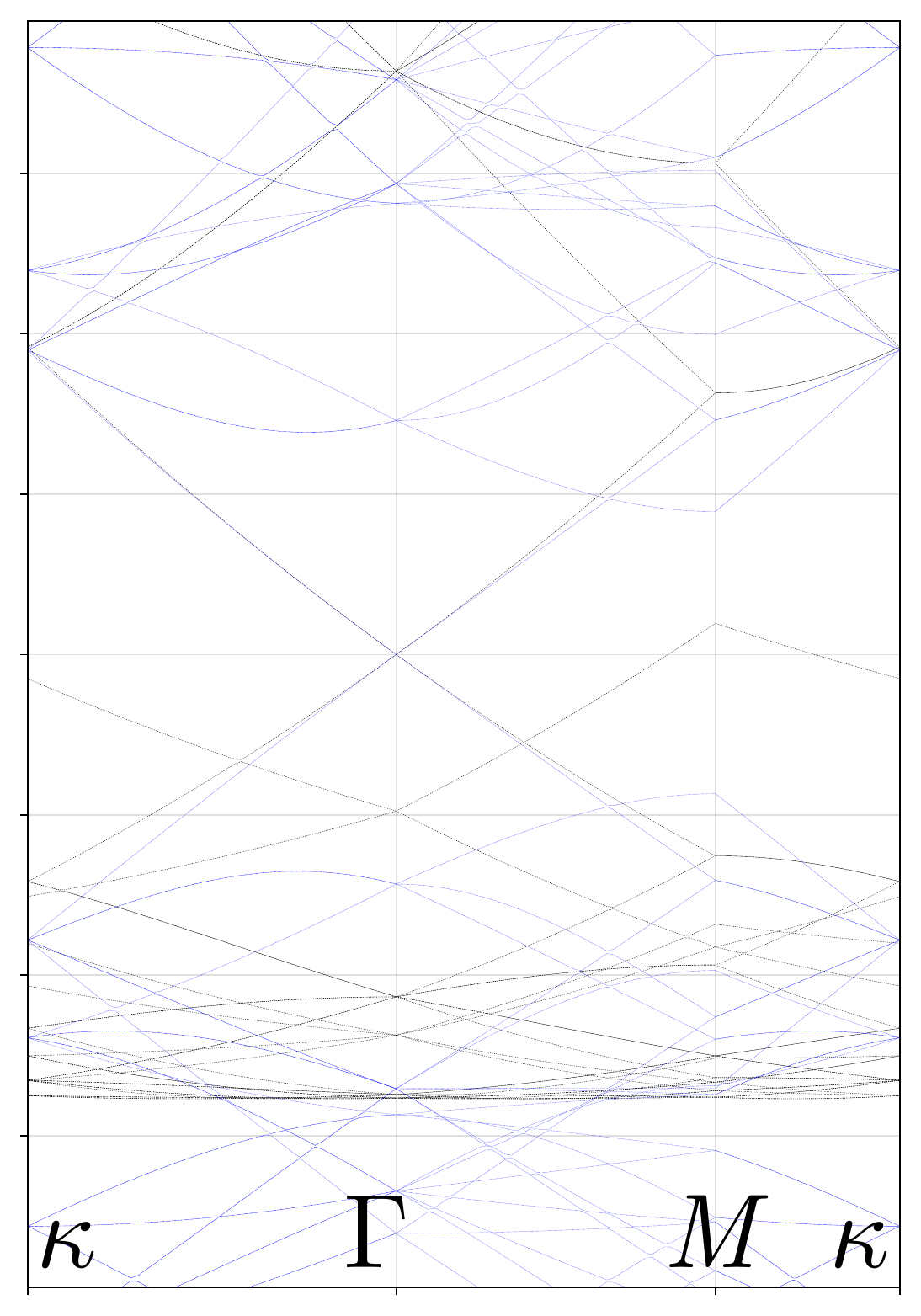}\caption{$\nu = 2$}\end{subfigure}
\begin{subfigure}[b]{0.15\linewidth}\includegraphics[width=\linewidth,trim={0.35cm 0cm 0.34cm 0cm},clip]{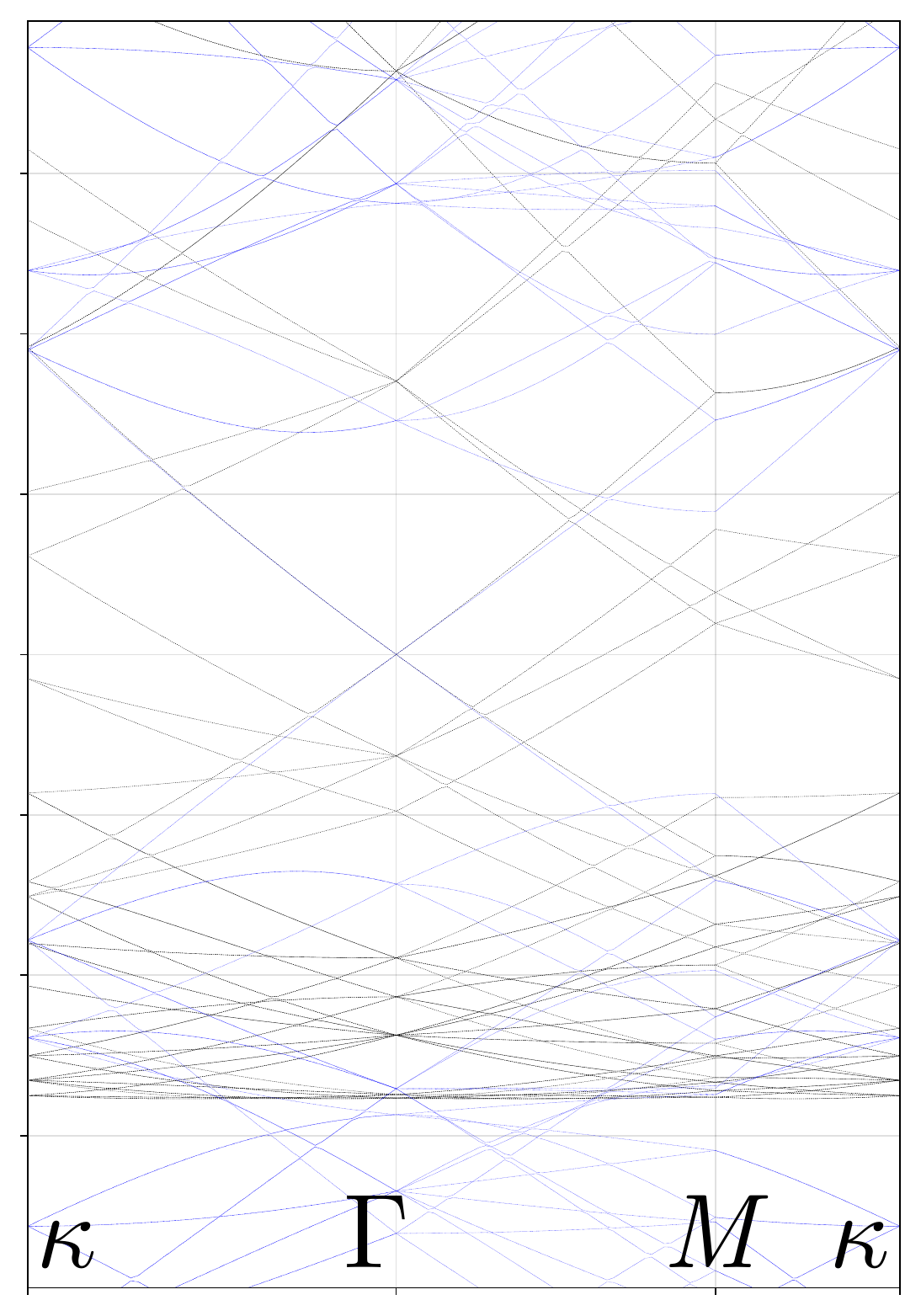}\caption{$\nu = 3$}\end{subfigure}
\begin{subfigure}[b]{0.15\linewidth}\includegraphics[width=\linewidth,trim={0.35cm 0cm 0.34cm 0cm},clip]{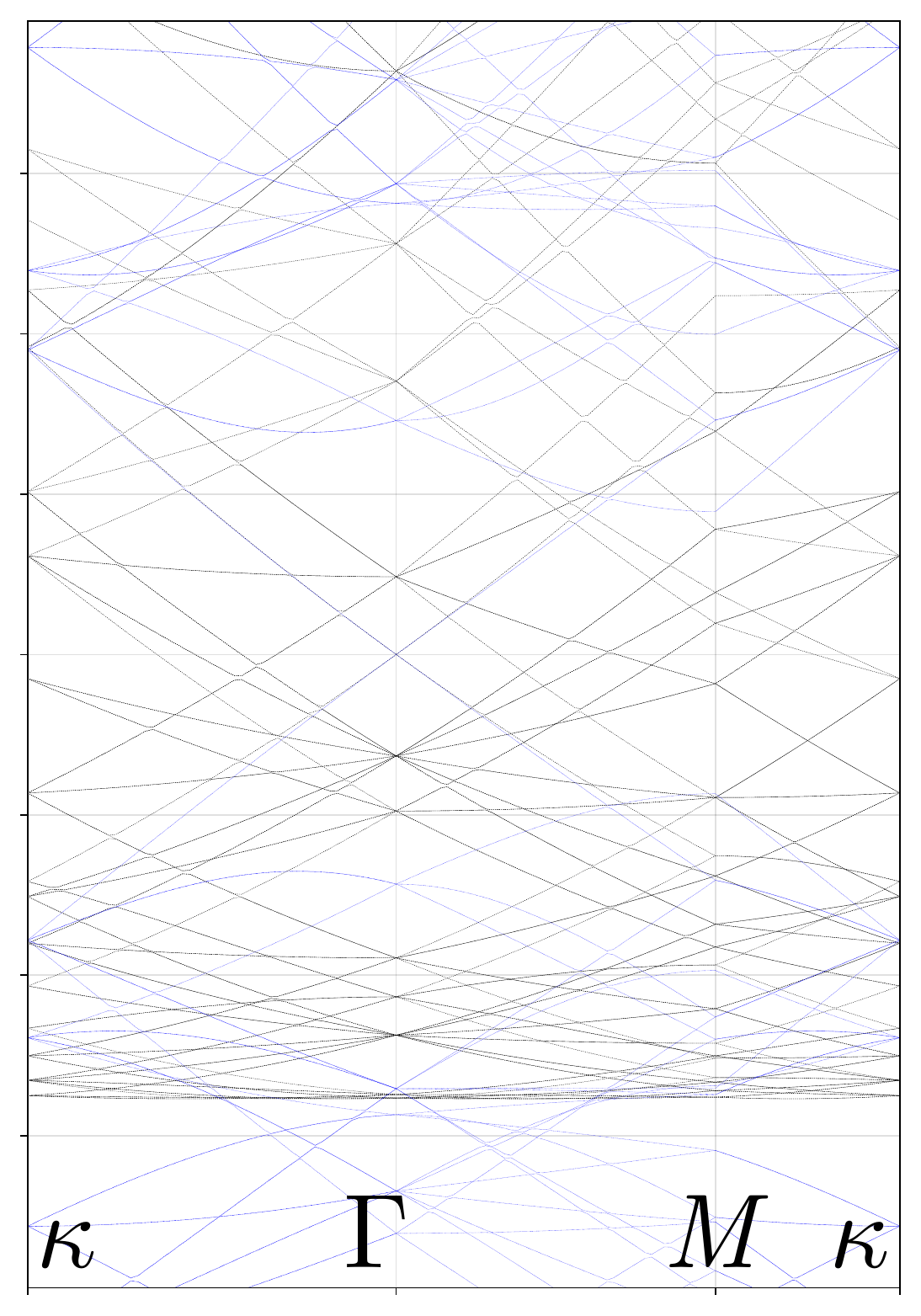}\caption{$\nu = 4$}\end{subfigure}
\begin{subfigure}[b]{0.15\linewidth}\includegraphics[width=\linewidth,trim={0.35cm 0cm 0.34cm 0cm},clip]{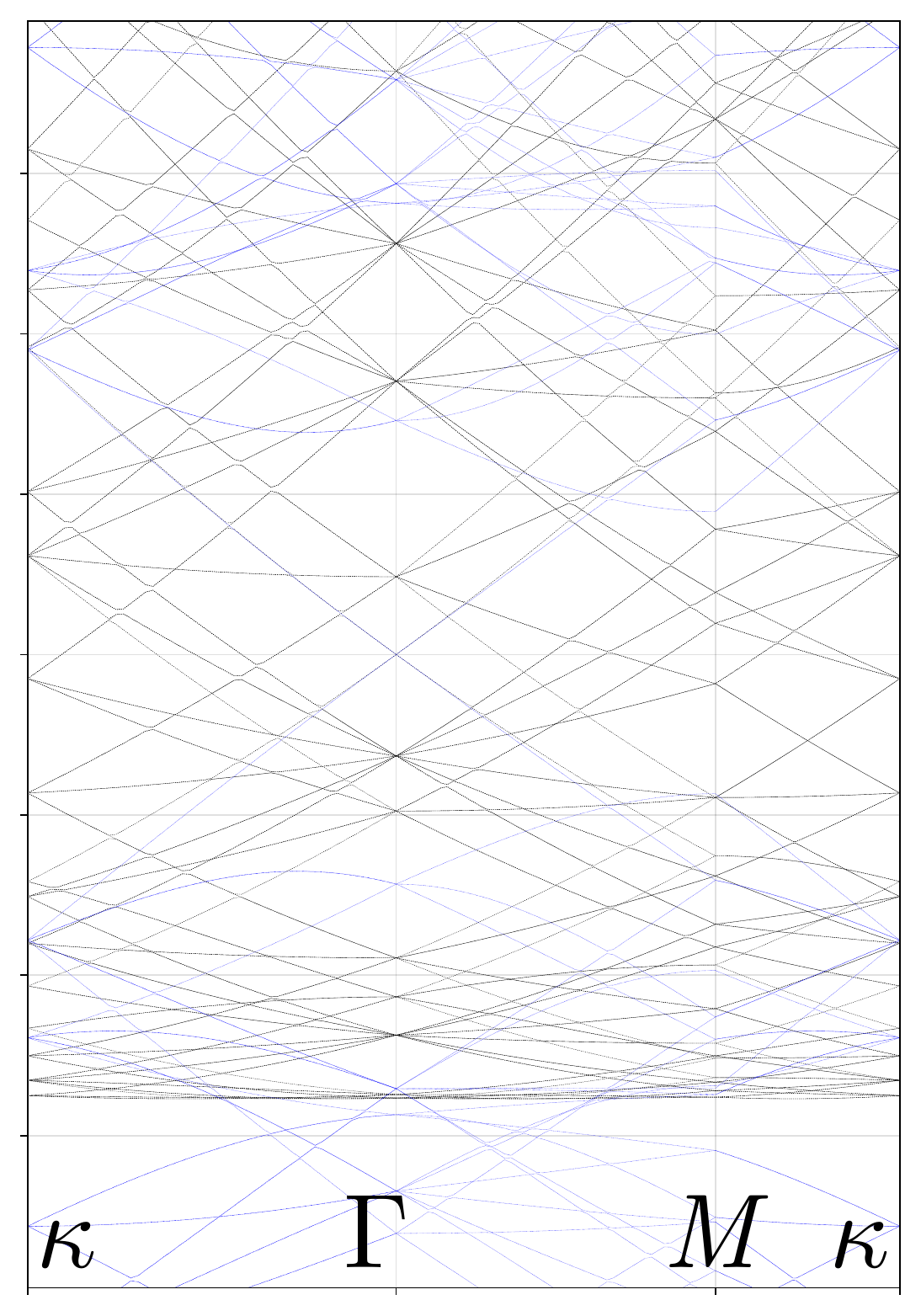}\caption{$\nu = 6$}\end{subfigure}
\begin{subfigure}[b]{0.15\linewidth}\includegraphics[width=\linewidth,trim={0.35cm 0cm 0.34cm 0cm},clip]{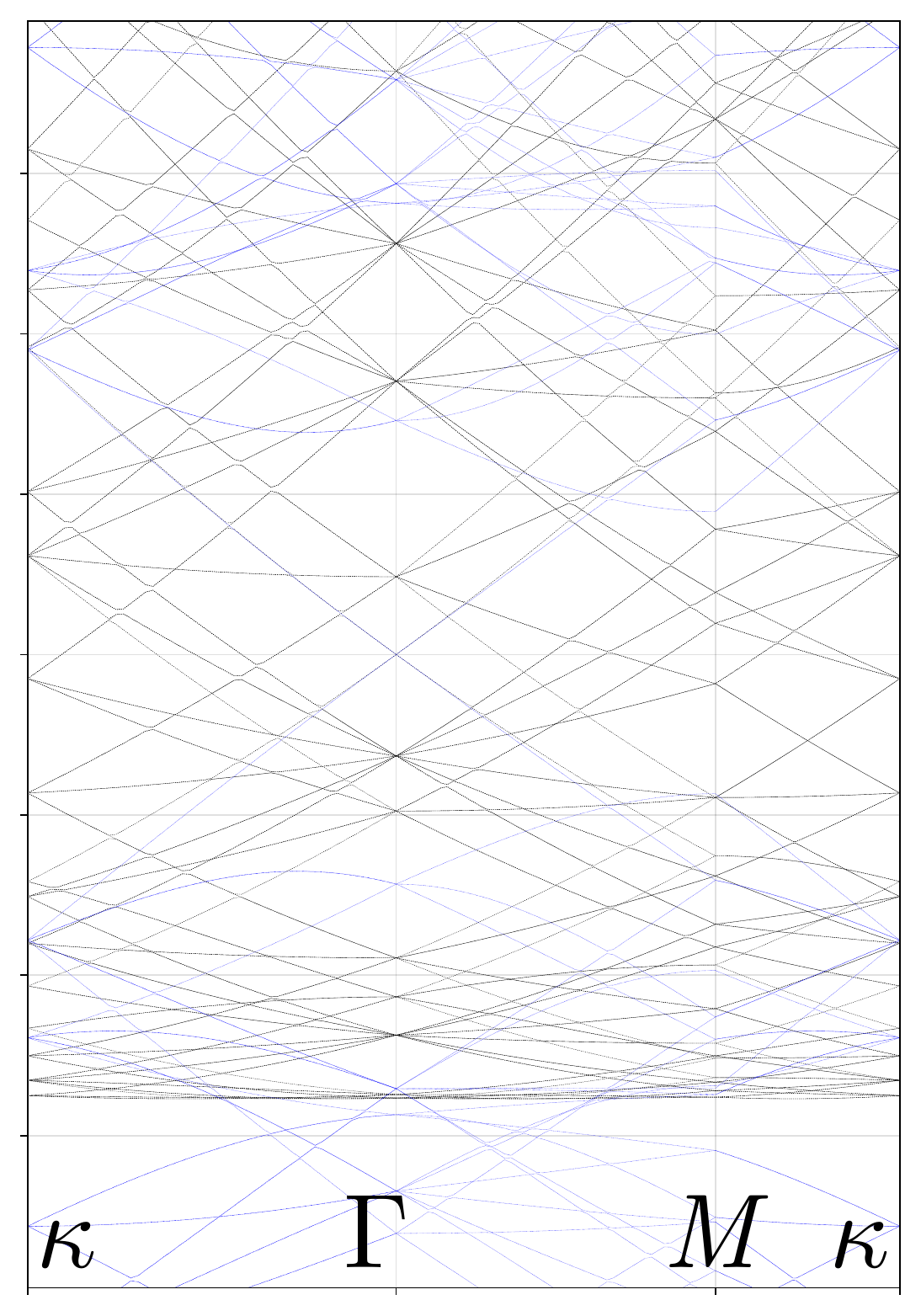}\caption{$\nu = 8$}\end{subfigure} 
\restoregeometry
\caption{Letting $\nu$ vary at fixed $\ep = \f{1}{7}$, and $\cF = \cF_0$, $V = 0$.}\label{fig:increase_nu}
\end{center}
\end{figure}

\begin{figure}[h!]
\begin{center}
\newgeometry{left=0.5cm,right=0.5cm}  
\hspace{-6cm}
\begin{subfigure}[b]{0.15\linewidth}\includegraphics[width=\linewidth,trim={0.35cm 0cm 0.34cm 0cm},clip]{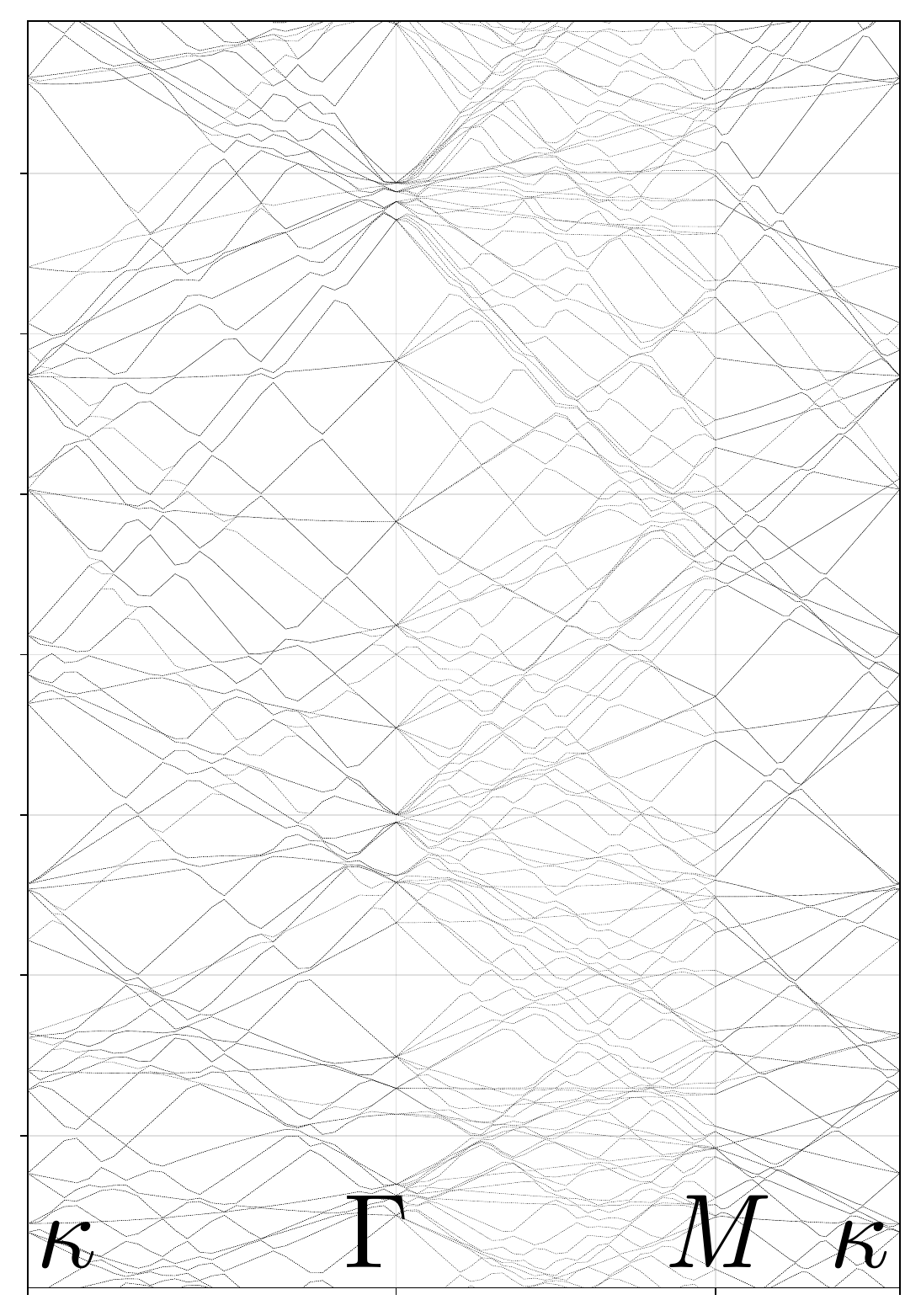}\caption{$\f{1}{\ep} = 7$}\end{subfigure}
\begin{subfigure}[b]{0.15\linewidth}\includegraphics[width=\linewidth,trim={0.35cm 0cm 0.34cm 0cm},clip]{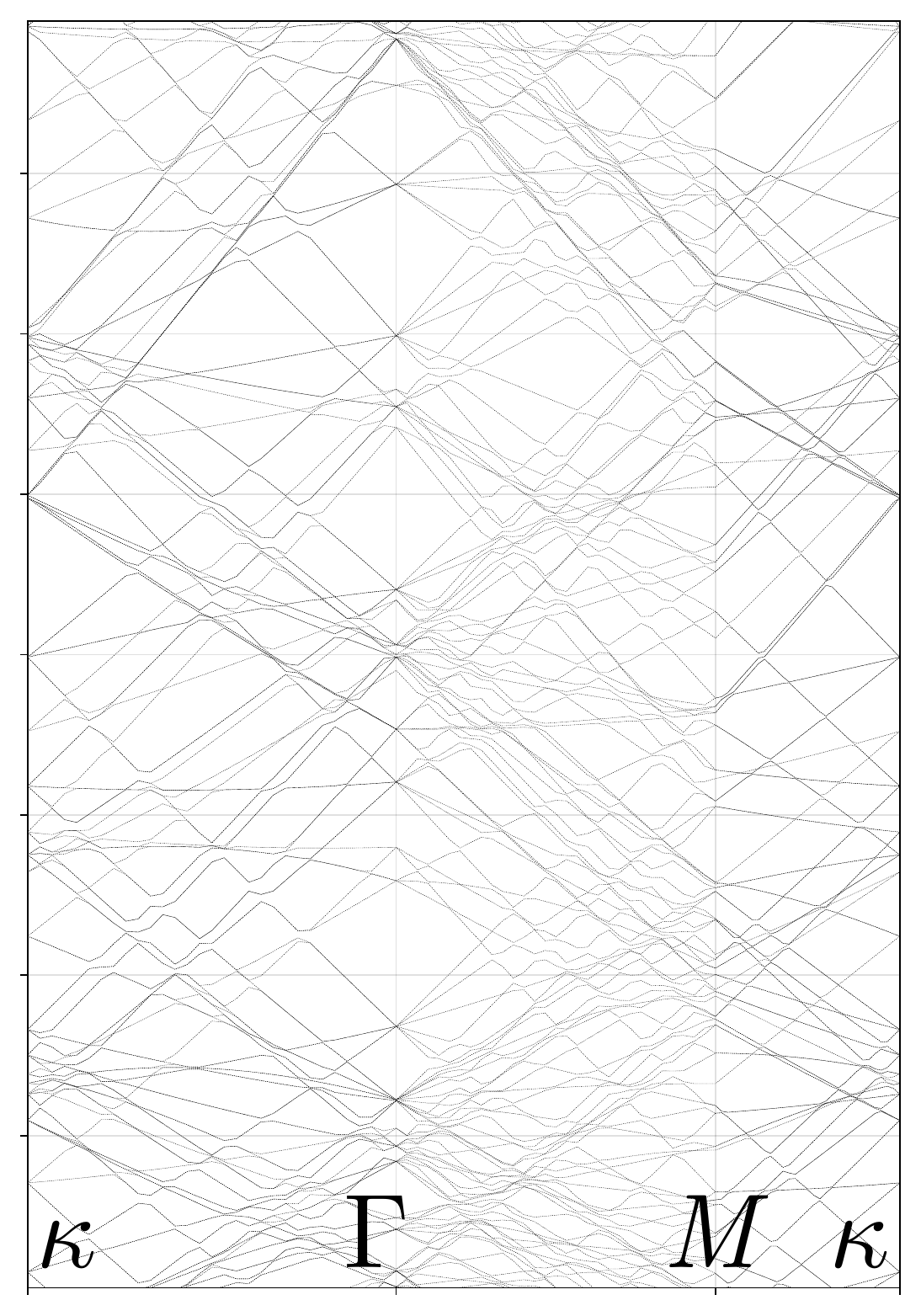}\caption{$\f{1}{\ep} = 10$}\end{subfigure}
\begin{subfigure}[b]{0.15\linewidth}\includegraphics[width=\linewidth,trim={0.35cm 0cm 0.34cm 0cm},clip]{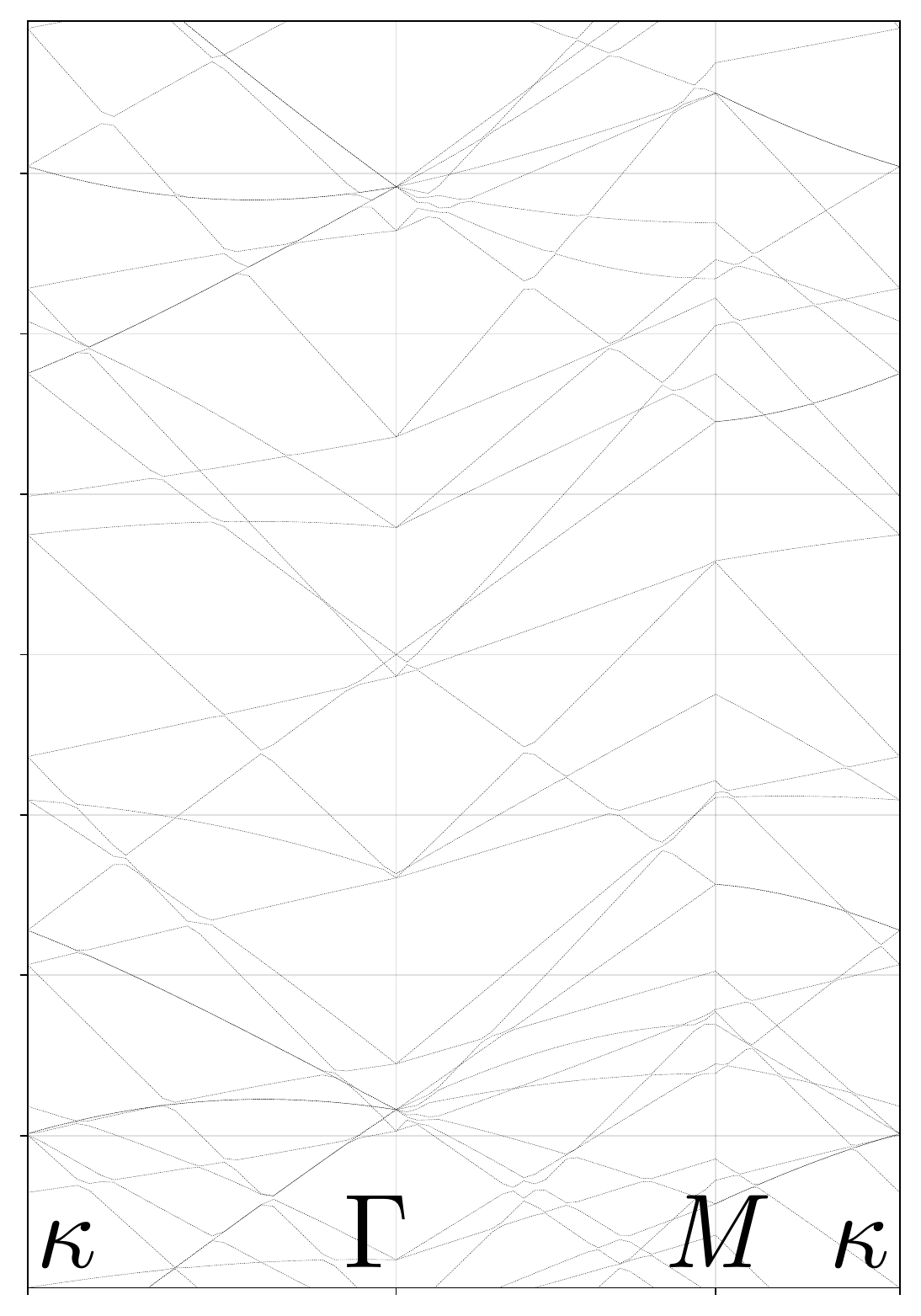}\caption{$\f{1}{\ep} = 19$}\end{subfigure}
\begin{subfigure}[b]{0.15\linewidth}\includegraphics[width=\linewidth,trim={0.35cm 0cm 0.34cm 0cm},clip]{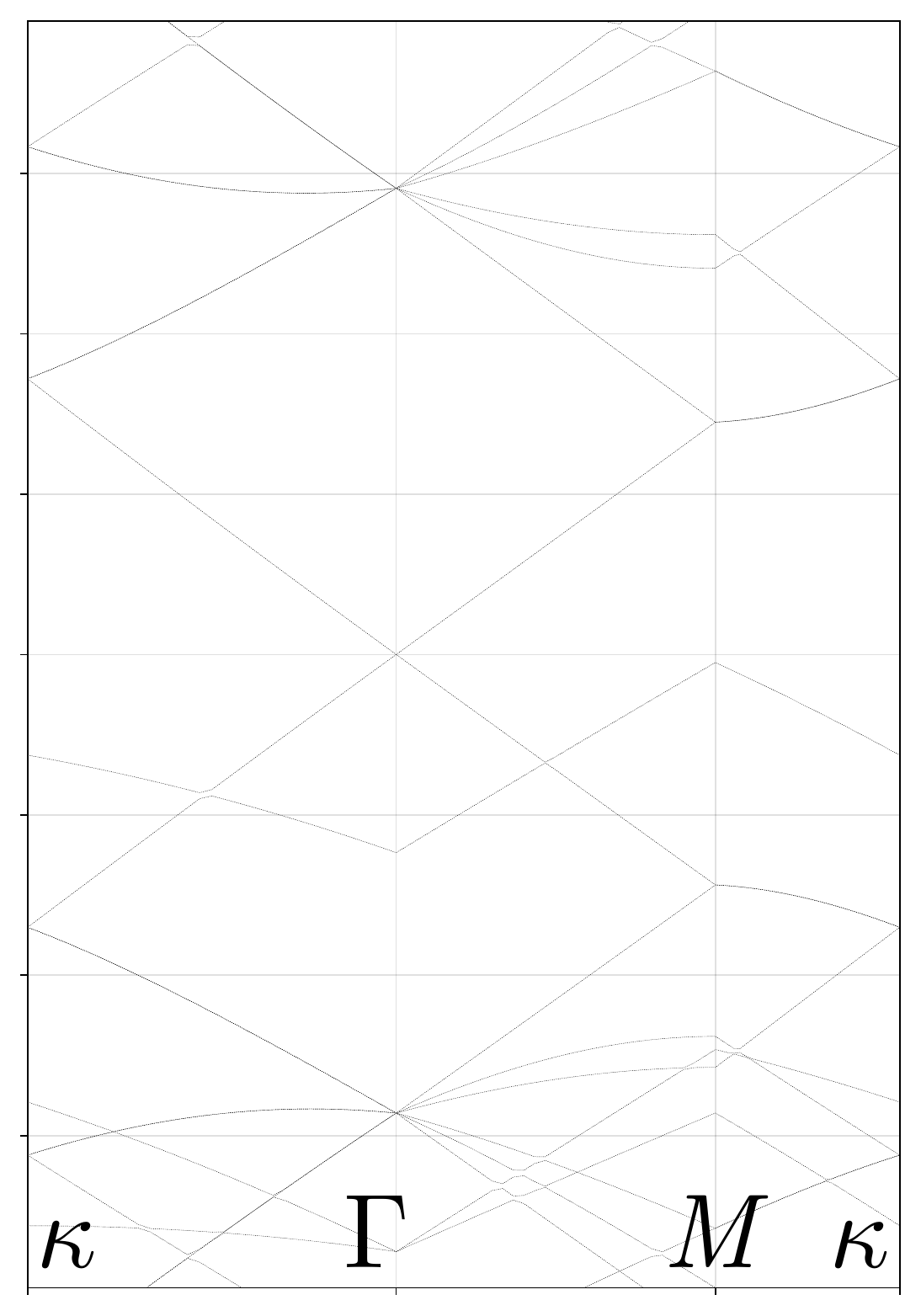}\caption{$\f{1}{\ep} = 31$}\end{subfigure}
\begin{subfigure}[b]{0.15\linewidth}\includegraphics[width=\linewidth,trim={0.35cm 0cm 0.34cm 0cm},clip]{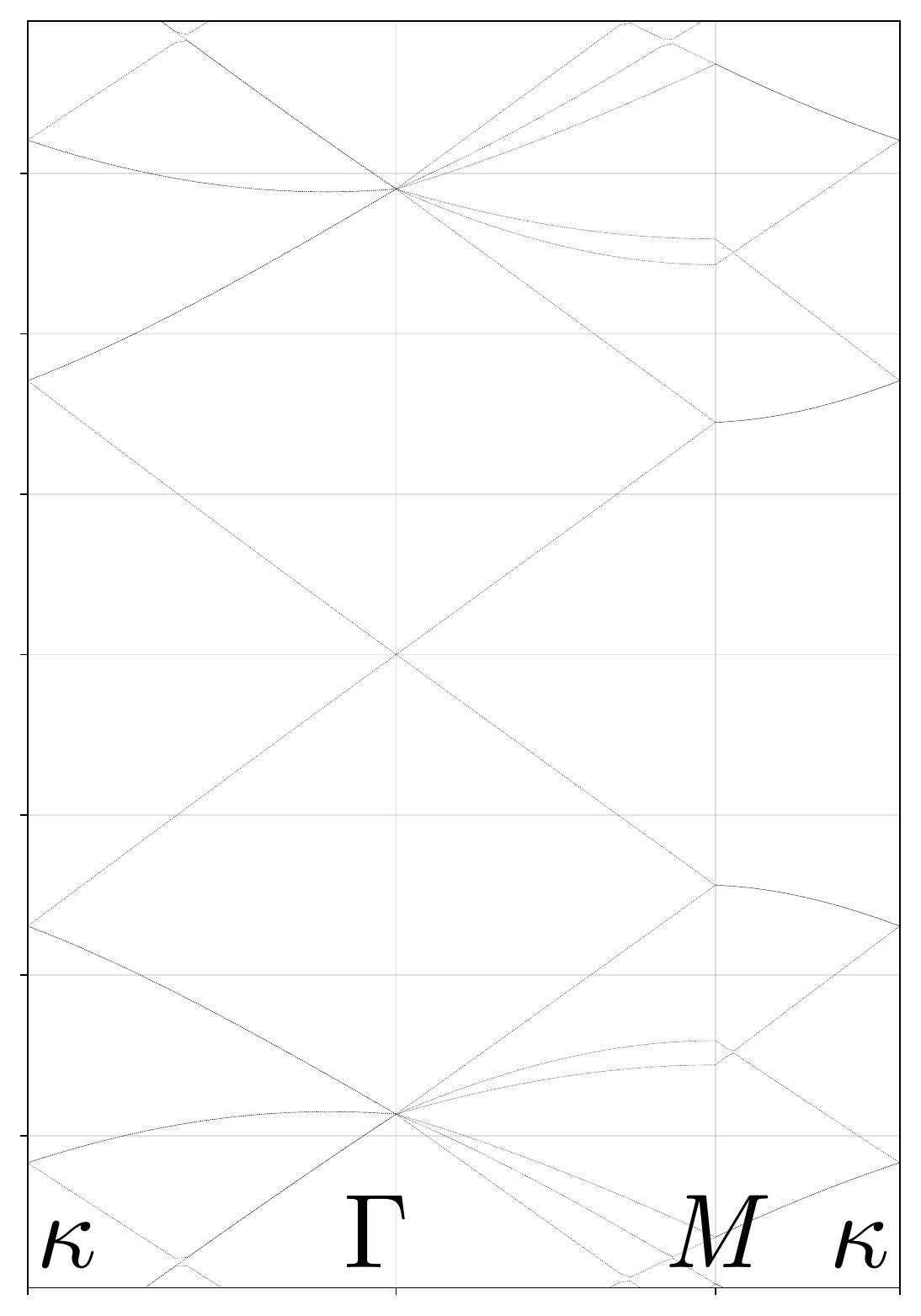}\caption{$\f{1}{\ep} = 40$}\end{subfigure}
\begin{subfigure}[b]{0.15\linewidth}\includegraphics[width=\linewidth,trim={0.35cm 0cm 0.34cm 0cm},clip]{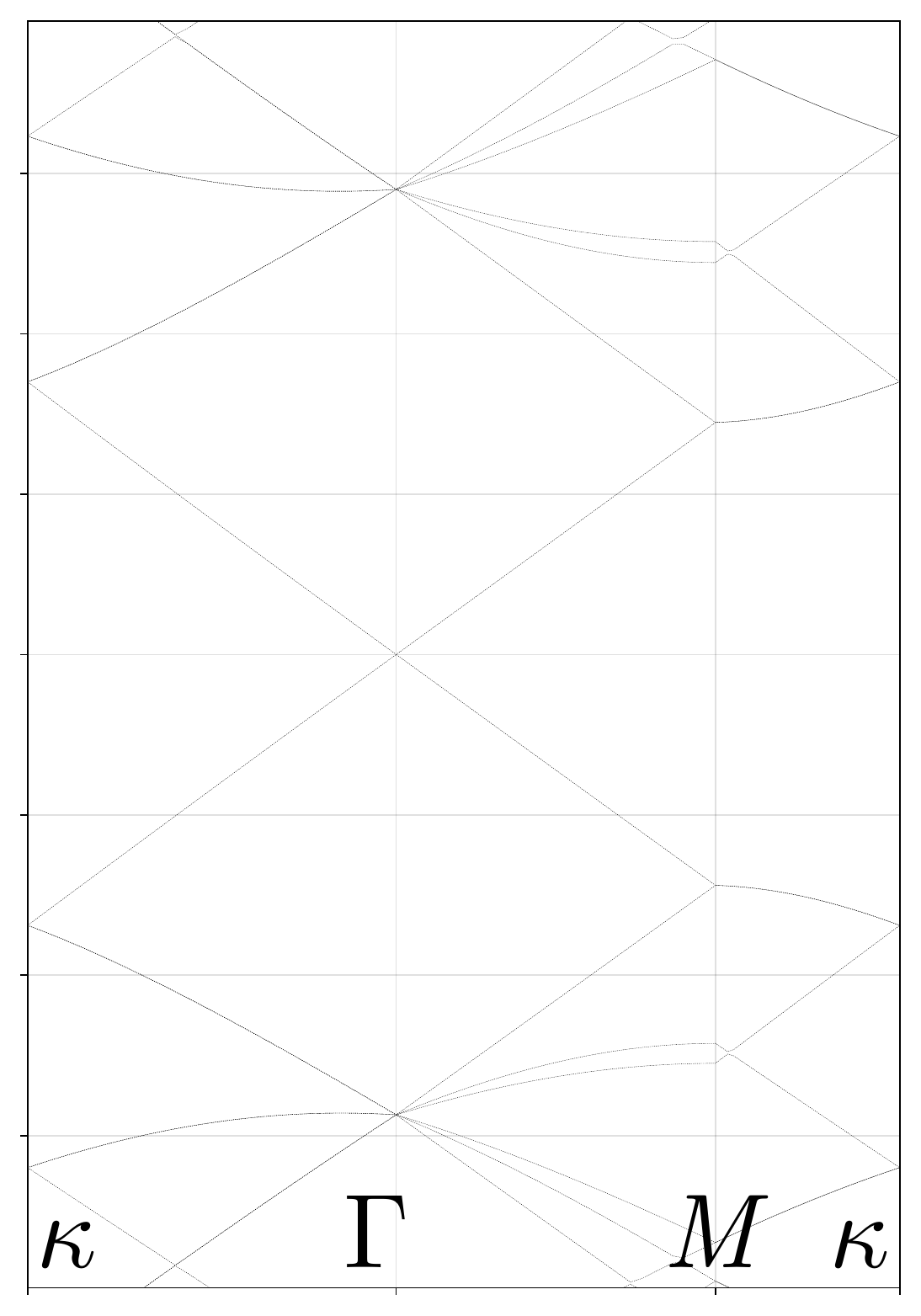}\caption{$\f{1}{\ep} = 49$}\end{subfigure}
\restoregeometry
\caption{Letting $\ep$ decrease at fixed but large $\nu = 21$, and where only bands of the effective model are plotted. We chose $\cF_{1}$, $\nu = 21$, $V = 0$.}\label{fig:increase_invep_fixed_nu}
\end{center}
\end{figure}

\subsubsection{Origin of the spectral pollution}%
\label{ssub:Origin of the spectral pollution}

The mechanism can already be seen explicitly for $V=A=0$ and the family $\cF_0$. For $m\in\Z^2$, set $\xi_{m,k}:=m a^*+k$. On the plane wave $p_m$, the two eigenvalues of $v\fer\sigma\cdot(-i\na_k)+\frac12\ep(-i\na_k)^2$ are
\begin{align}\label{eq:pollution_branches}
 \lambda_\pm(\xi_{m,k})=\frac{\ep}{2}|\xi_{m,k}|^2\pm v\fer|\xi_{m,k}|.
\end{align}
The upper branch is increasing as a function of $|\xi_{m,k}|$, whereas the lower branch $\lambda_-(r)=\frac12\ep r^2-v\fer r$ reaches its minimum at $r=v\fer/\ep$ and returns to the Fermi energy $0$ at $r=2v\fer/\ep$. Thus, at fixed $\ep$, increasing $\nu$ eventually includes a high-frequency shell $|\xi_{m,k}|\simeq2v\fer/\ep$ whose eigenvalues lie again in the low-energy window. Once the cutoff contains that shell, the number of such eigenvalues in a fixed bounded energy window stabilizes. Conversely, at fixed $\nu$, decreasing $\ep$ moves the shell to frequencies of order $\ep^{-1}$ and hence beyond the numerical space, which explains the expulsion observed in Figure~\ref{fig:increase_invep_fixed_nu}. Said differently, those artifical eigenvalues are present because the effective operator has been obtained from information localized near the Dirac point and is then being used at macroscopic Fourier frequencies of order $\ep^{-1}$.


Following~\cite{GerMarSjo91}, one could introduce a smooth phase-space cutoff $\chi(\ep\xi)$ supported in a neighborhood of the Dirac point, and compare the full Grushin effective Hamiltonian with the finite-order matrix model on the support of $\chi$. For the Fourier discretization used here, this suggests replacing the purely empirical choice of $\nu$ by a condition of the form $\ep\max_{|m_i|\le\nu}|m a^*+k|\le\rho$, where $\rho$ is chosen inside the momentum neighborhood on which the symbol expansion at the Dirac point is controlled.

\subsection{Varying $\ell$ and $k$-dep}%
\label{sub:Varying ell}

We take $\nu = 2$, $\ep = \f{1}{7}$. On Figure~\ref{fig:increase_ell}, we take $V = 0$, we let $\ell$ vary in $\{0,1,2\}$, and for $\ell \in \{1,2\}$, we provide the band diagrams of the $k$-dependent families $\cF_{\ell,k}$ and the ones of the $k$-independent families $\cF_{\ell}$. We also provide the result for the ``6 states'' family $\cF^6$ defined in Section~\ref{sub:Variational space from excited states}. It is natural to compare families having the same number of states, so we can focus on comparing $\cF_{1}$, $\cF_{2,k}$ and $\cF^6$, which all have $6$ states. We see that the other families than $\cF_0$ reproduce the main branch better than $\cF_0$.

\begin{figure}[h!]
\begin{center}
\newgeometry{left=0.5cm,right=0.5cm}  
\hspace{-6cm}
\begin{subfigure}[b]{0.15\linewidth}\includegraphics[width=\linewidth,trim={0.35cm 0cm 0.34cm 0cm},clip]{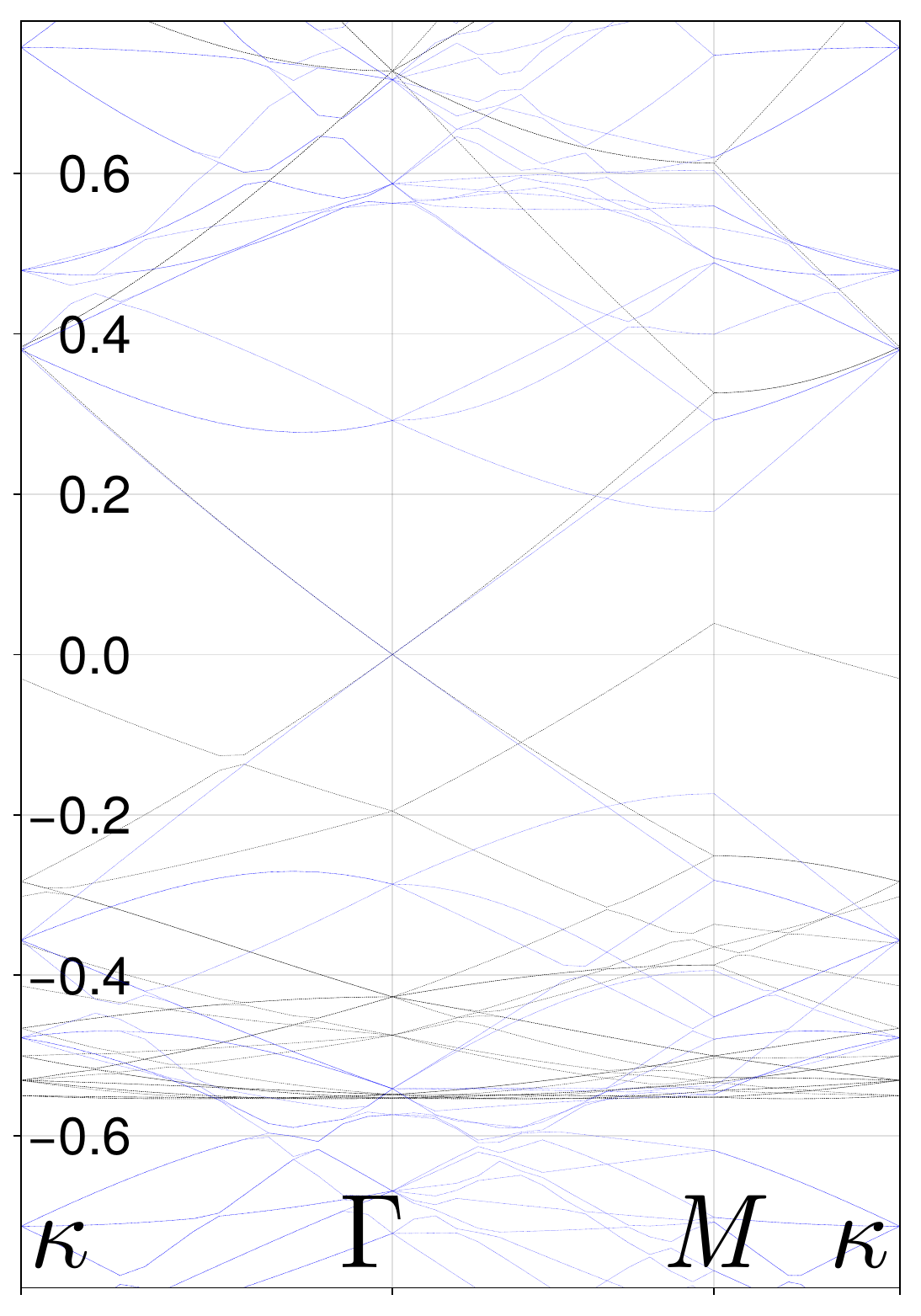}\caption{$\cF = \cF_0$}\end{subfigure}
\begin{subfigure}[b]{0.15\linewidth}\includegraphics[width=\linewidth,trim={0.35cm 0cm 0.34cm 0cm},clip]{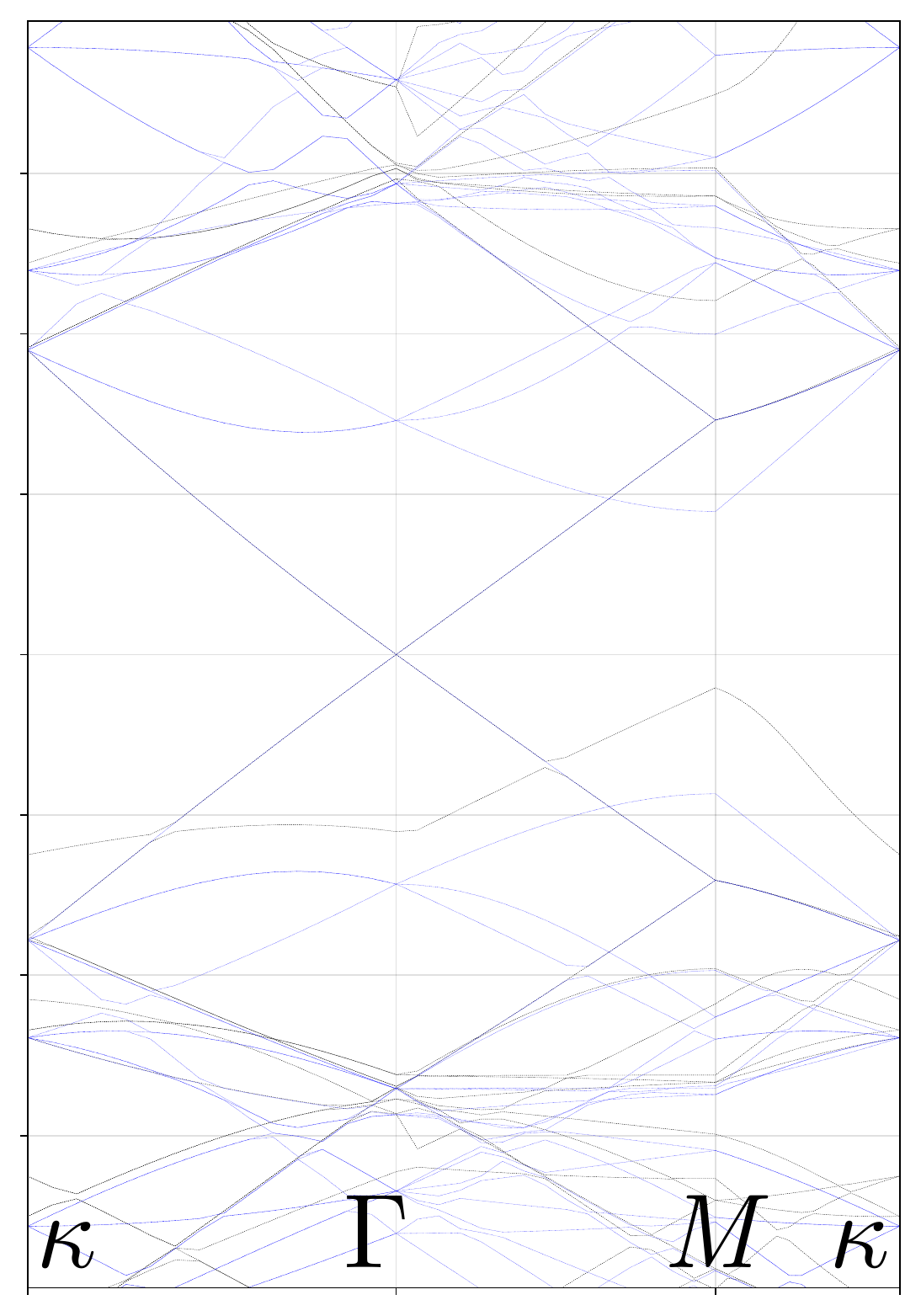}\caption{$\cF = \cF_{1,k}$}\end{subfigure}
\begin{subfigure}[b]{0.15\linewidth}\includegraphics[width=\linewidth,trim={0.35cm 0cm 0.34cm 0cm},clip]{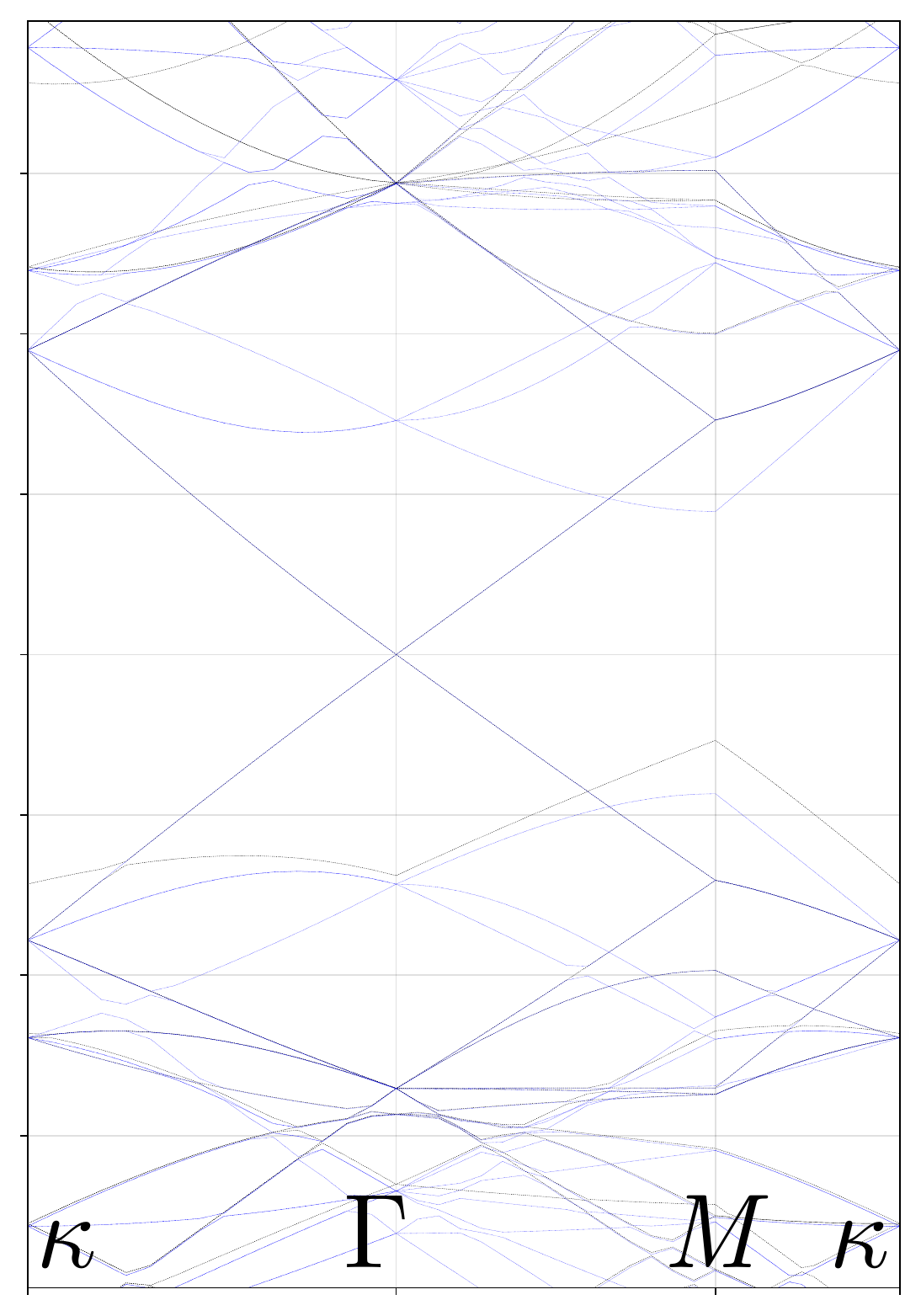}\caption{$\cF = \cF_{1}$}\end{subfigure}
\begin{subfigure}[b]{0.15\linewidth}\includegraphics[width=\linewidth,trim={0.35cm 0cm 0.34cm 0cm},clip]{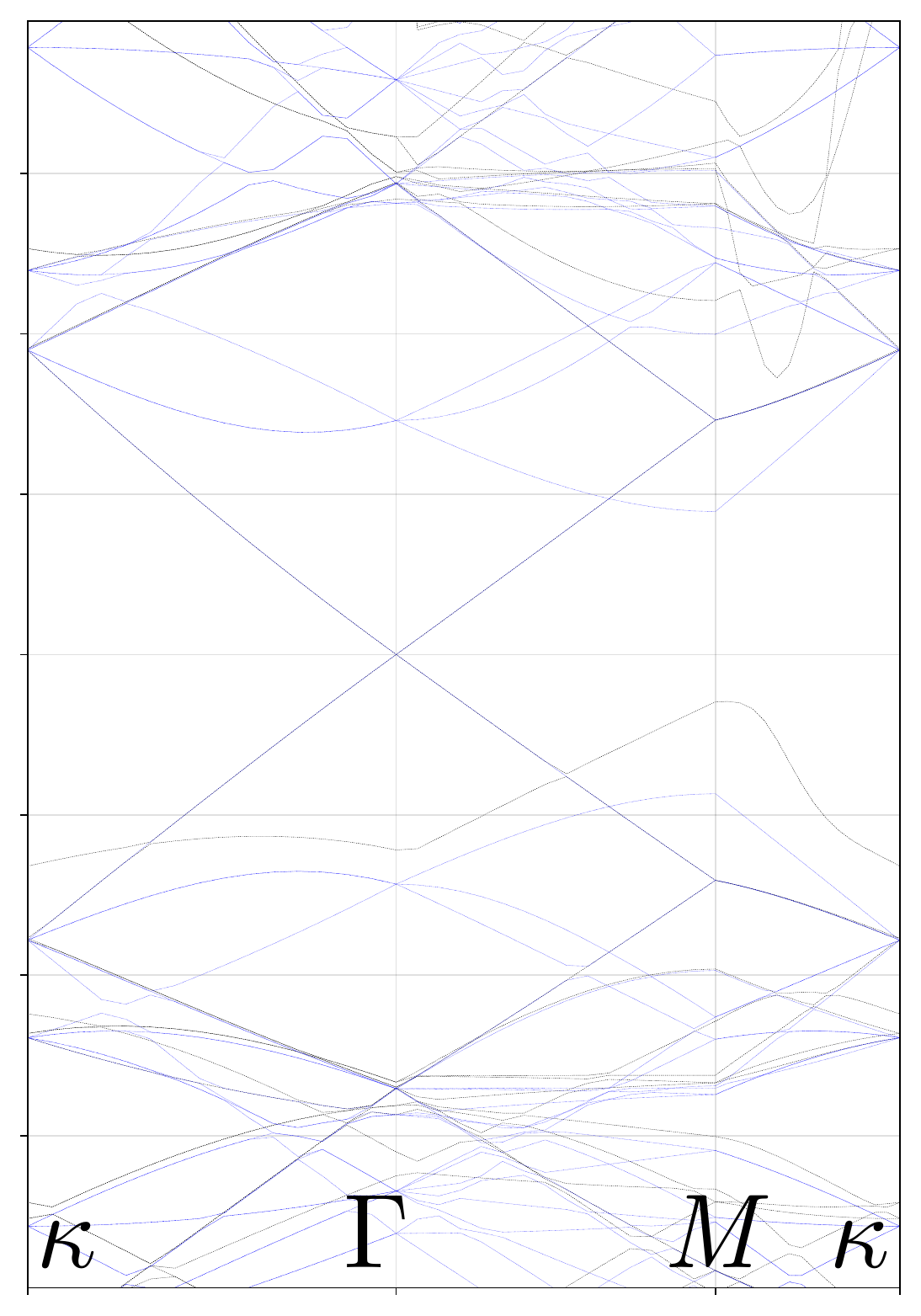}\caption{$\cF = \cF_{2,k}$}\end{subfigure}
\begin{subfigure}[b]{0.15\linewidth}\includegraphics[width=\linewidth,trim={0.35cm 0cm 0.34cm 0cm},clip]{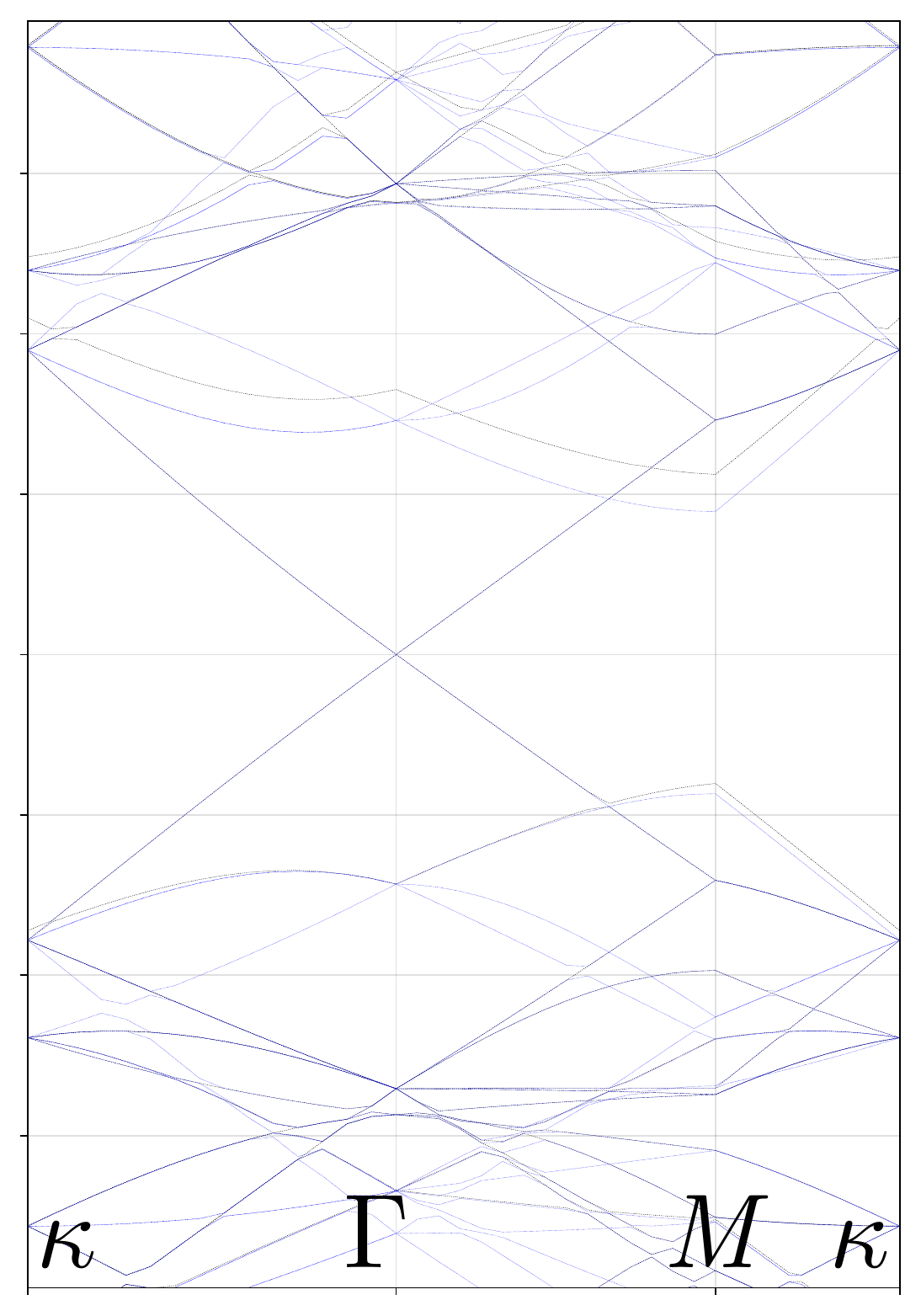}\caption{$\cF = \cF_{2}$}\end{subfigure}
\begin{subfigure}[b]{0.15\linewidth}\includegraphics[width=\linewidth,trim={0.35cm 0cm 0.34cm 0cm},clip]{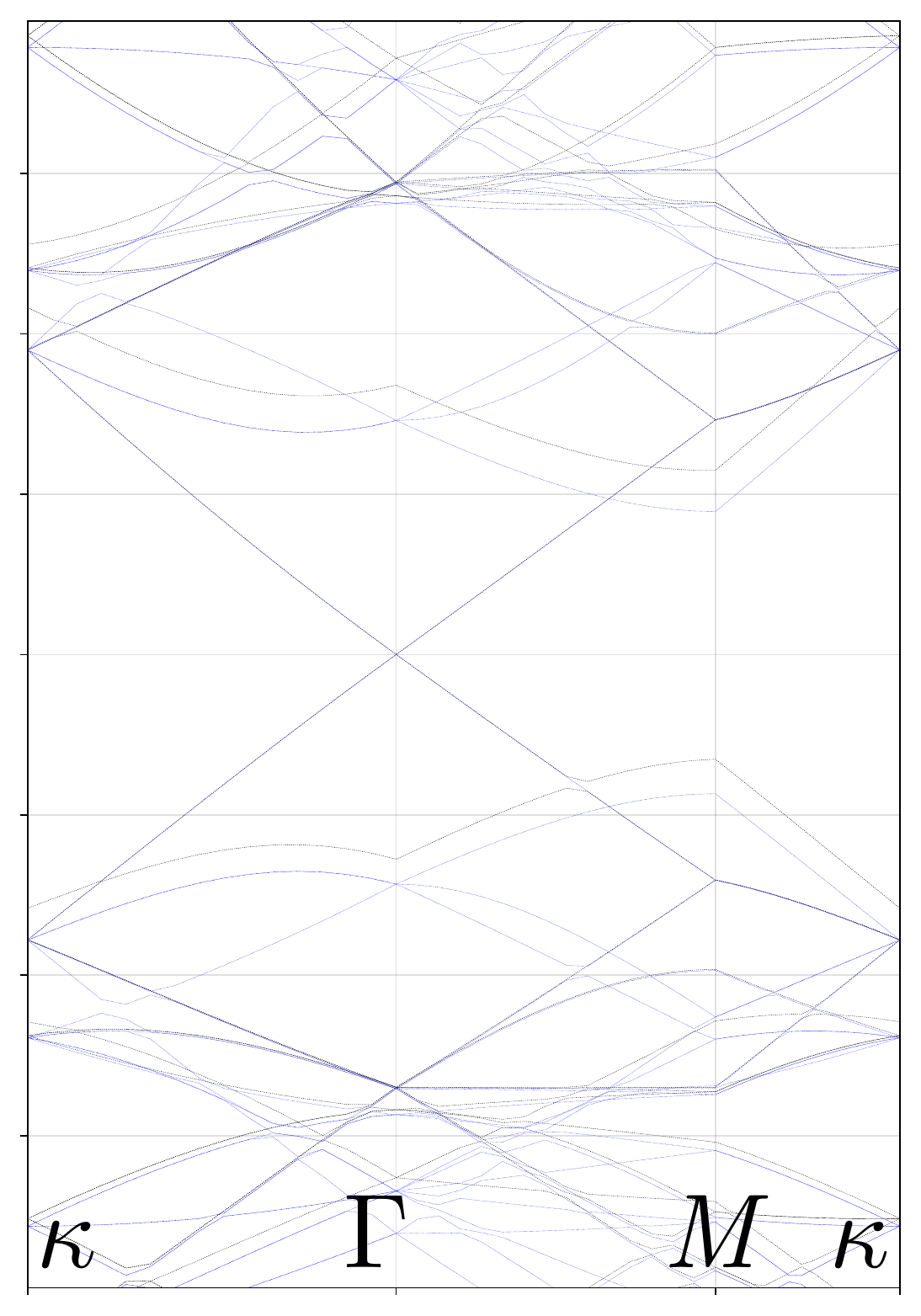}\caption{$\cF = \cF^6$}\end{subfigure} 
\restoregeometry
\caption{Letting $\ell$ vary, with $\ep = \f{1}{7}$, $\nu = 2$ and $V = 0$.}\label{fig:increase_ell}
\end{center}
\end{figure}

To better see the usefulness of the method, we also provide in Figure~\ref{fig:increase_ell_v_non_zero} the same band diagrams but with a potential $V = V_{\text{ng}}$ which does not have honeycomb symmetry. Let us define $\widetilde{m}^1 := \pa{1,0}$, $\widetilde{m}^2 := \pa{0,1}$, $\widetilde{m}^3 := \pa{0,2}$. We used
\begin{align}\label{eq:V_non_graphene} 
V_{\text{ng}}(x) := 2 \lambda \sum_{i=1}^{3} \cos ( (\widetilde{m}^i a^*) \cdot x ).
\end{align}
Here again we see that the effective operator reproduces relatively well the exact bands.

\begin{figure}[h!]
\begin{center}
\newgeometry{left=0.5cm,right=0.5cm}  
\hspace{-6cm}
\begin{subfigure}[b]{0.15\linewidth}\includegraphics[width=\linewidth,trim={0.35cm 0cm 0.34cm 0cm},clip]{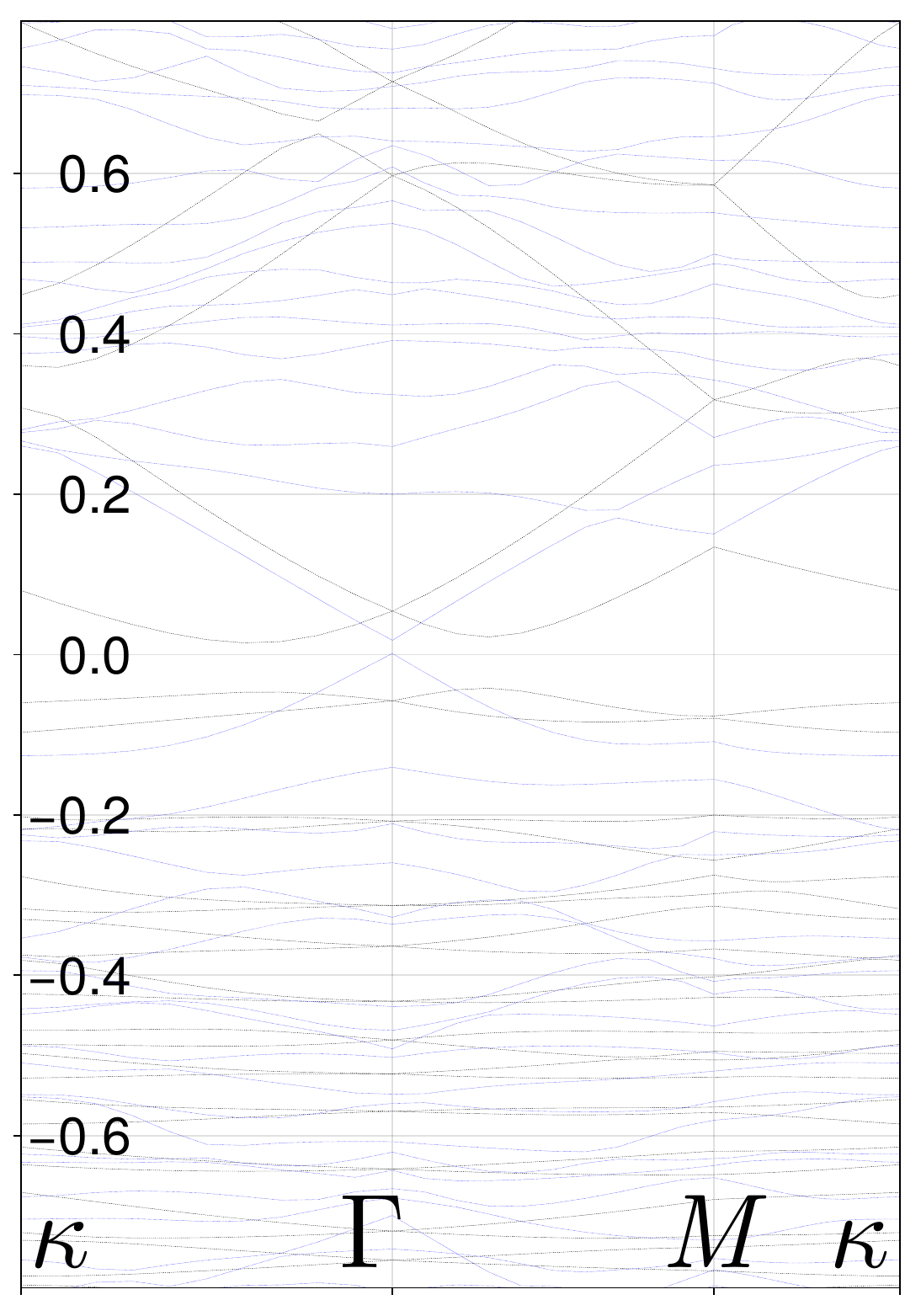}\caption{$\cF = \cF_0$}\end{subfigure}
\begin{subfigure}[b]{0.15\linewidth}\includegraphics[width=\linewidth,trim={0.35cm 0cm 0.34cm 0cm},clip]{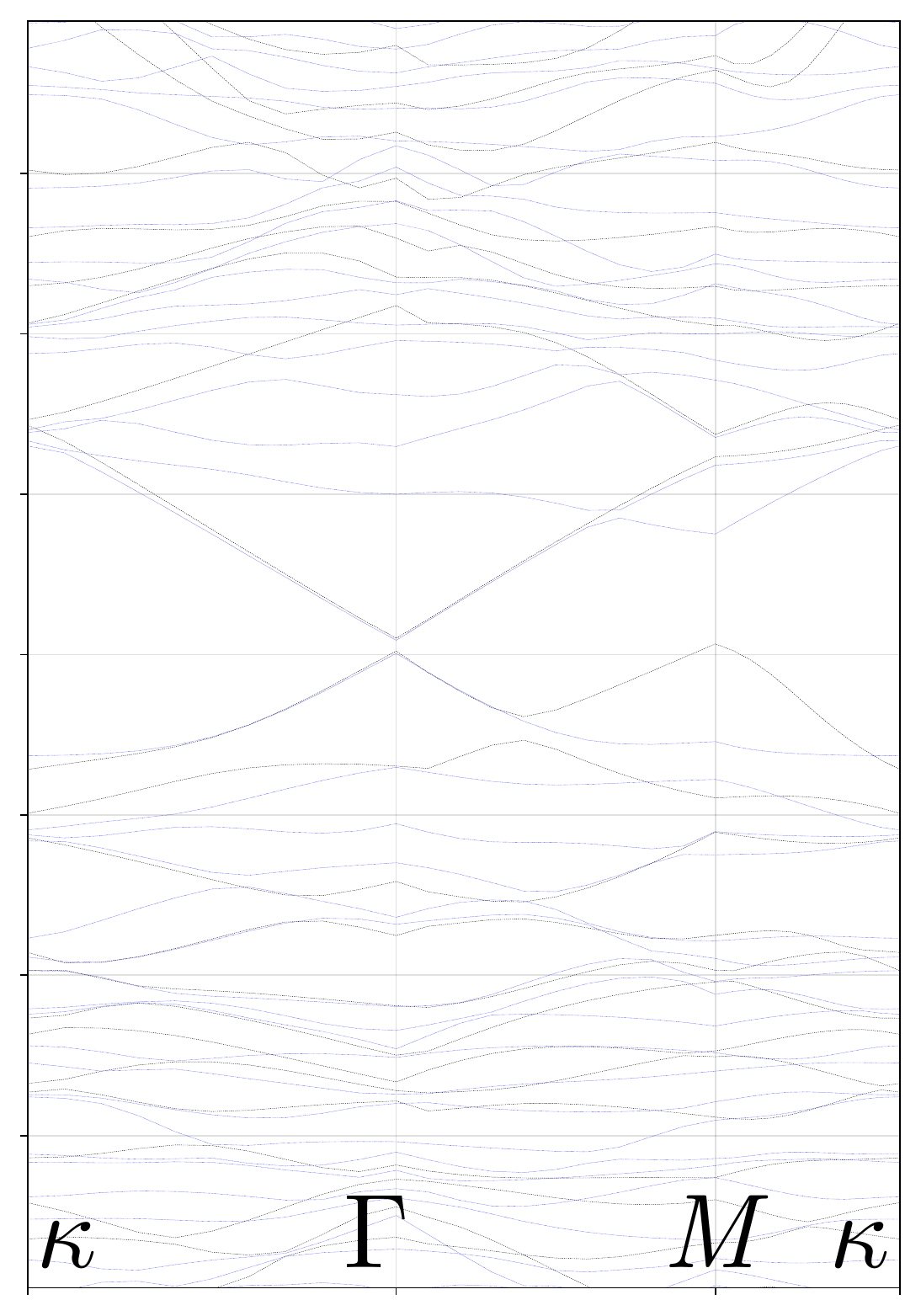}\caption{$\cF = \cF_{1,k}$}\end{subfigure}
\begin{subfigure}[b]{0.15\linewidth}\includegraphics[width=\linewidth,trim={0.35cm 0cm 0.34cm 0cm},clip]{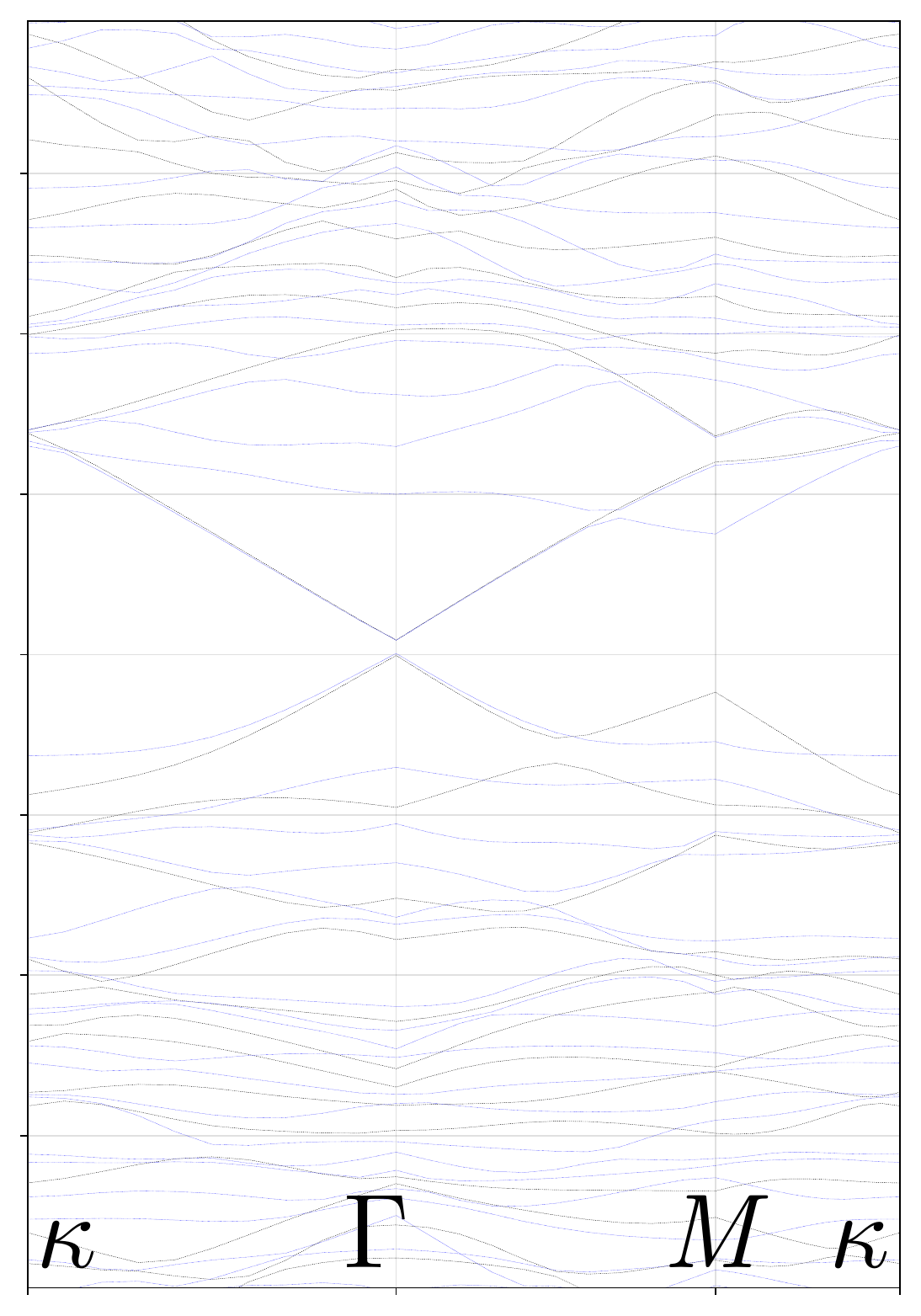}\caption{$\cF = \cF_{1}$}\end{subfigure}
\begin{subfigure}[b]{0.15\linewidth}\includegraphics[width=\linewidth,trim={0.35cm 0cm 0.34cm 0cm},clip]{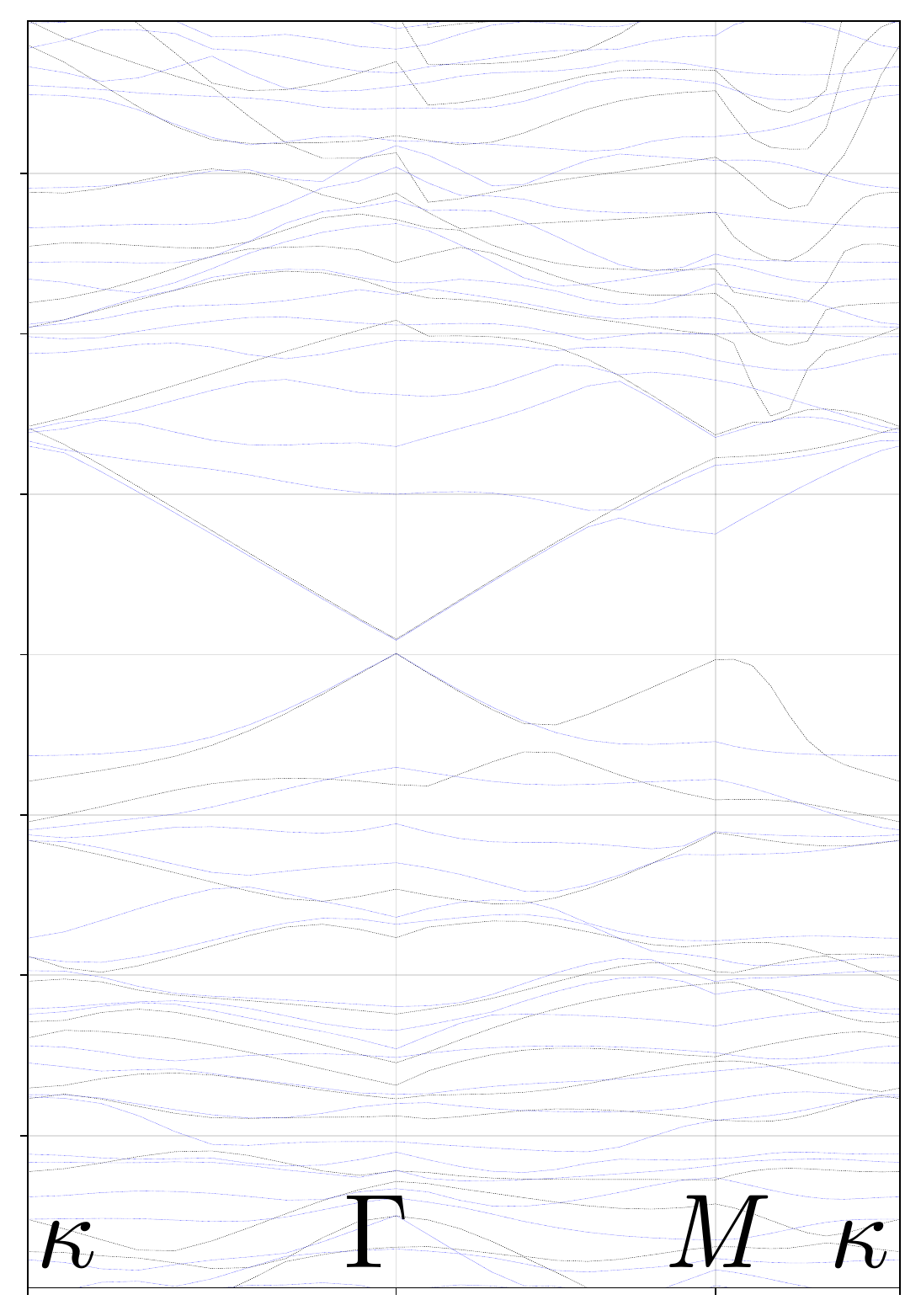}\caption{$\cF = \cF_{2,k}$}\end{subfigure}
\begin{subfigure}[b]{0.15\linewidth}\includegraphics[width=\linewidth,trim={0.35cm 0cm 0.34cm 0cm},clip]{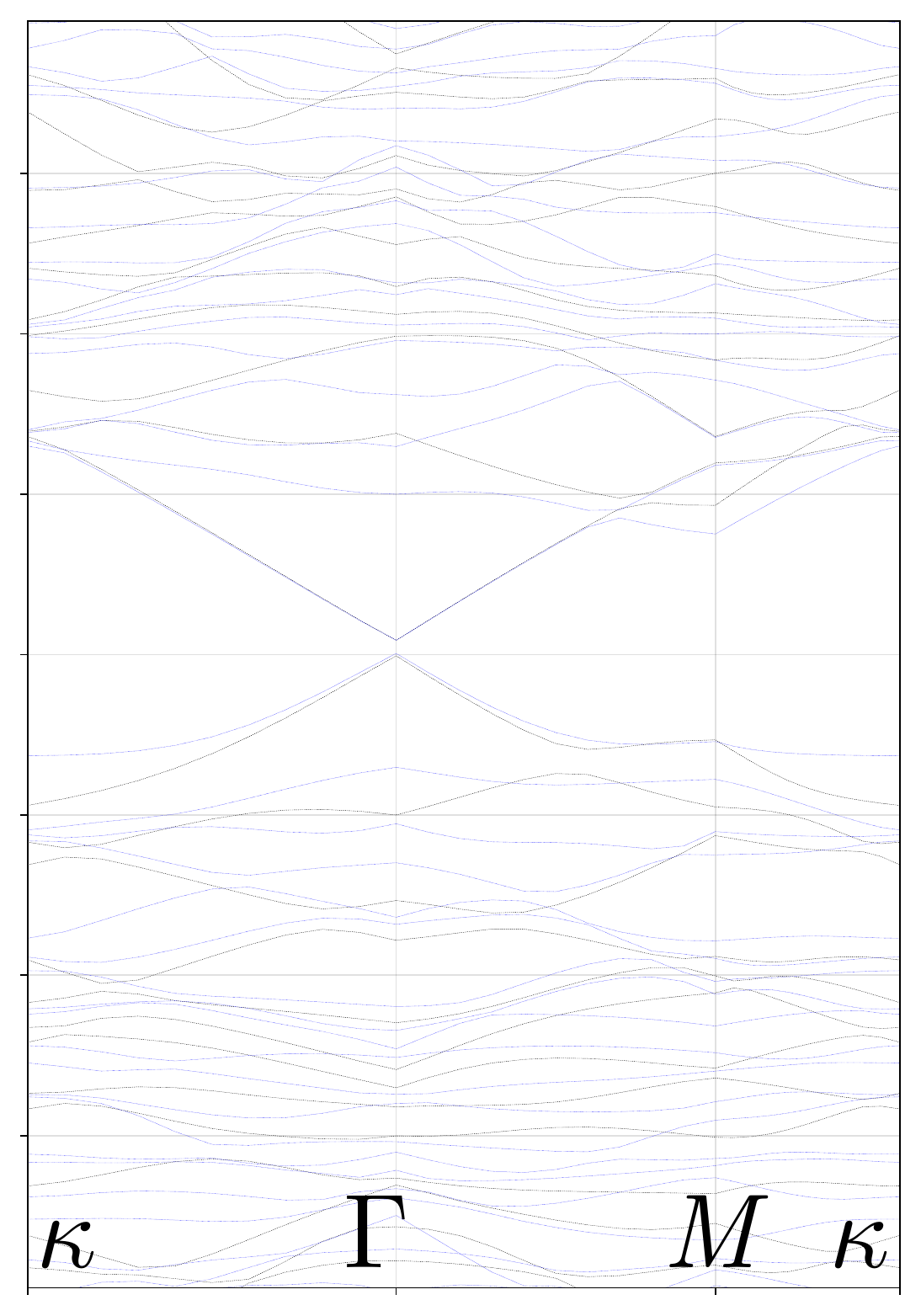}\caption{$\cF = \cF_{2}$}\end{subfigure}
\begin{subfigure}[b]{0.15\linewidth}\includegraphics[width=\linewidth,trim={0.35cm 0cm 0.34cm 0cm},clip]{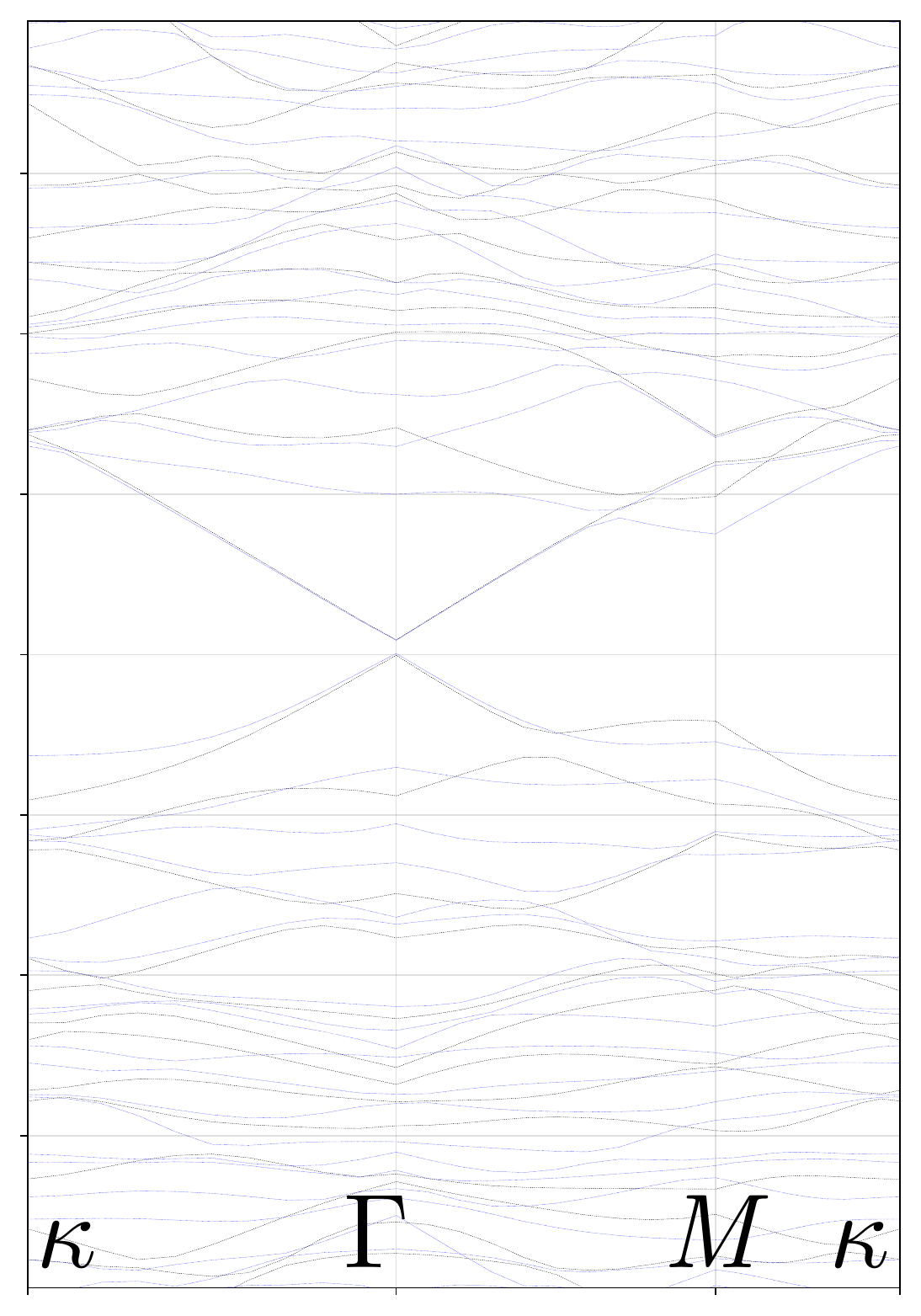}\caption{$\cF = \cF^6$}\end{subfigure} 
\restoregeometry
\caption{Letting $\ell$ vary, with $\ep = \f{1}{7}$, $\nu = 2$ and a non-zero $V = V_{\text{ng}}$ (with $\lambda = 5$) which does not have honeycomb symmetry, presented in~\eqref{eq:V_non_graphene}.}\label{fig:increase_ell_v_non_zero}
\end{center}
\end{figure}

\subsection{Varying $\ep$}%
\label{sub:Varying ep}

Now let us take $V = V_{\text{honeycomb}}$ defined in~\eqref{eq:def_lambda_V_honeycomb}, with $\lambda = 6$, $\cF_1$ and $\nu = 2$, and let $\ep$ vary. On Figure~\ref{fig:increase_ep_V0}, we display the band diagrams. The bands are well reproduced and similarly as in Figure~\ref{fig:increase_invep_fixed_nu}, the artificial bands are expelled as $\ep$ decreases.

\begin{figure}[h!]
\begin{center}
\newgeometry{left=0.5cm,right=0.5cm}  
\hspace{-6cm}
\begin{subfigure}[b]{0.15\linewidth}\includegraphics[width=\linewidth,trim={0.35cm 0cm 0.34cm 0cm},clip]{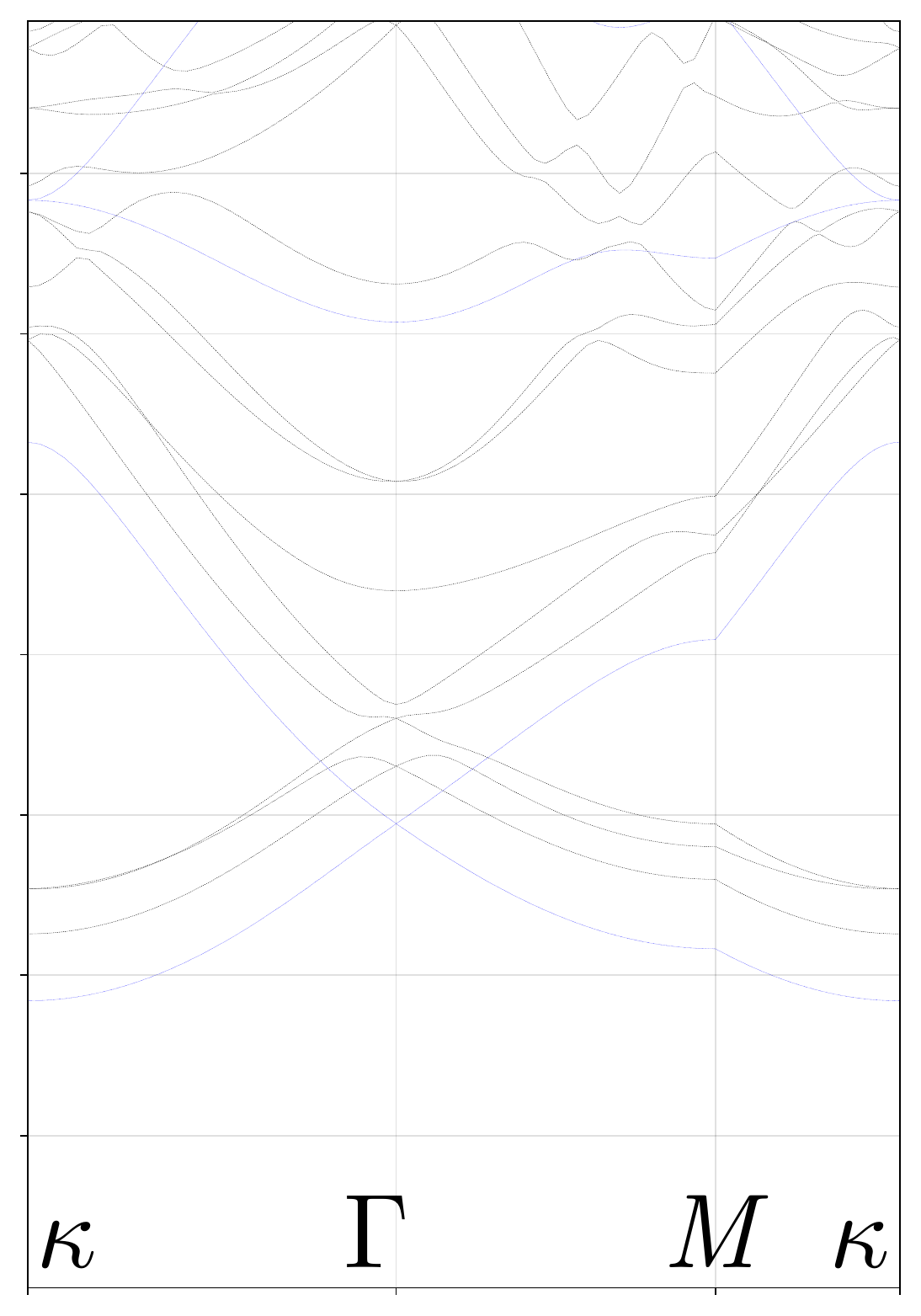}\caption{$\f{1}{\ep} = 1$}\end{subfigure}
\begin{subfigure}[b]{0.15\linewidth}\includegraphics[width=\linewidth,trim={0.35cm 0cm 0.34cm 0cm},clip]{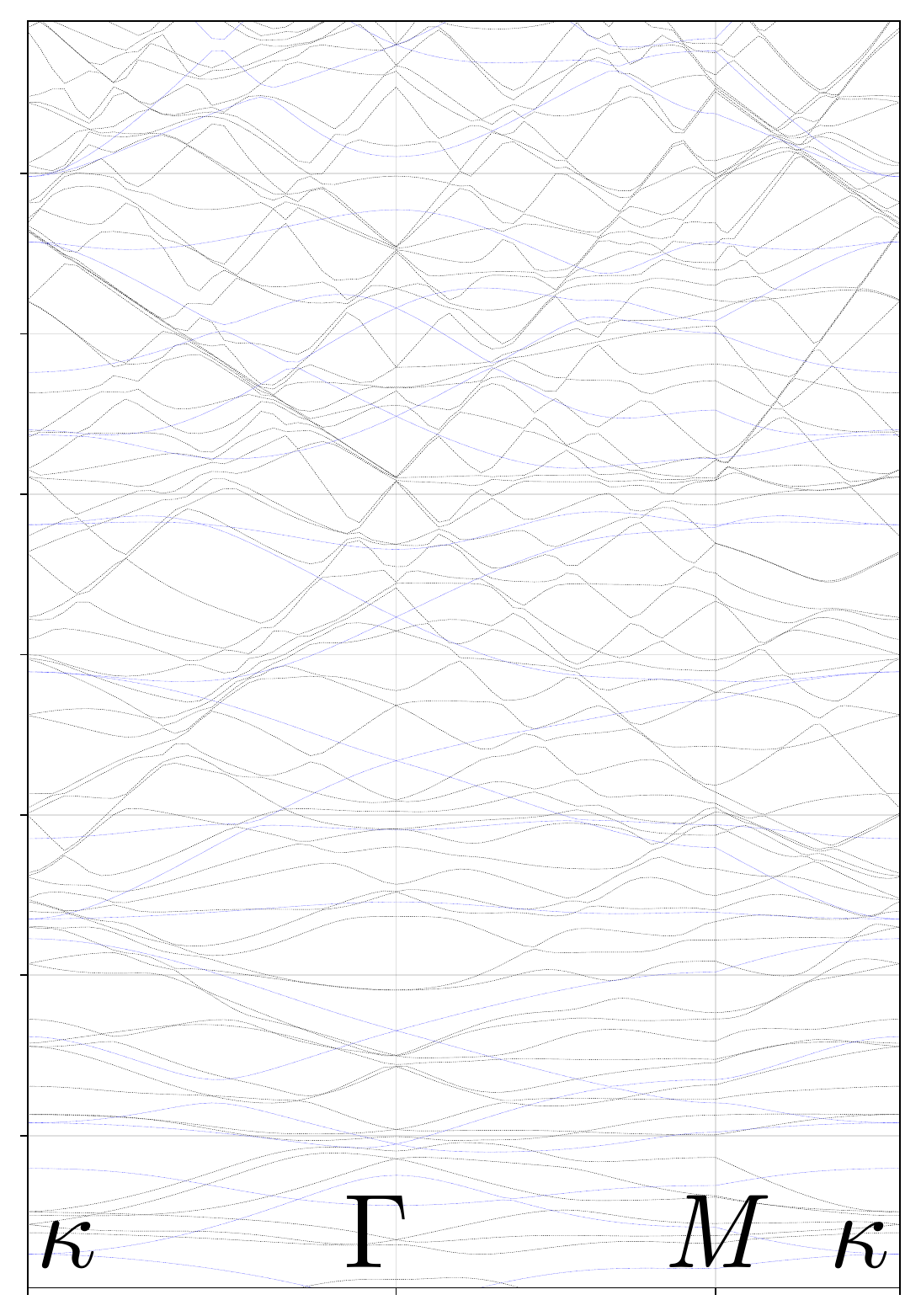}\caption{$\f{1}{\ep} = 2$}\end{subfigure}
\begin{subfigure}[b]{0.15\linewidth}\includegraphics[width=\linewidth,trim={0.35cm 0cm 0.34cm 0cm},clip]{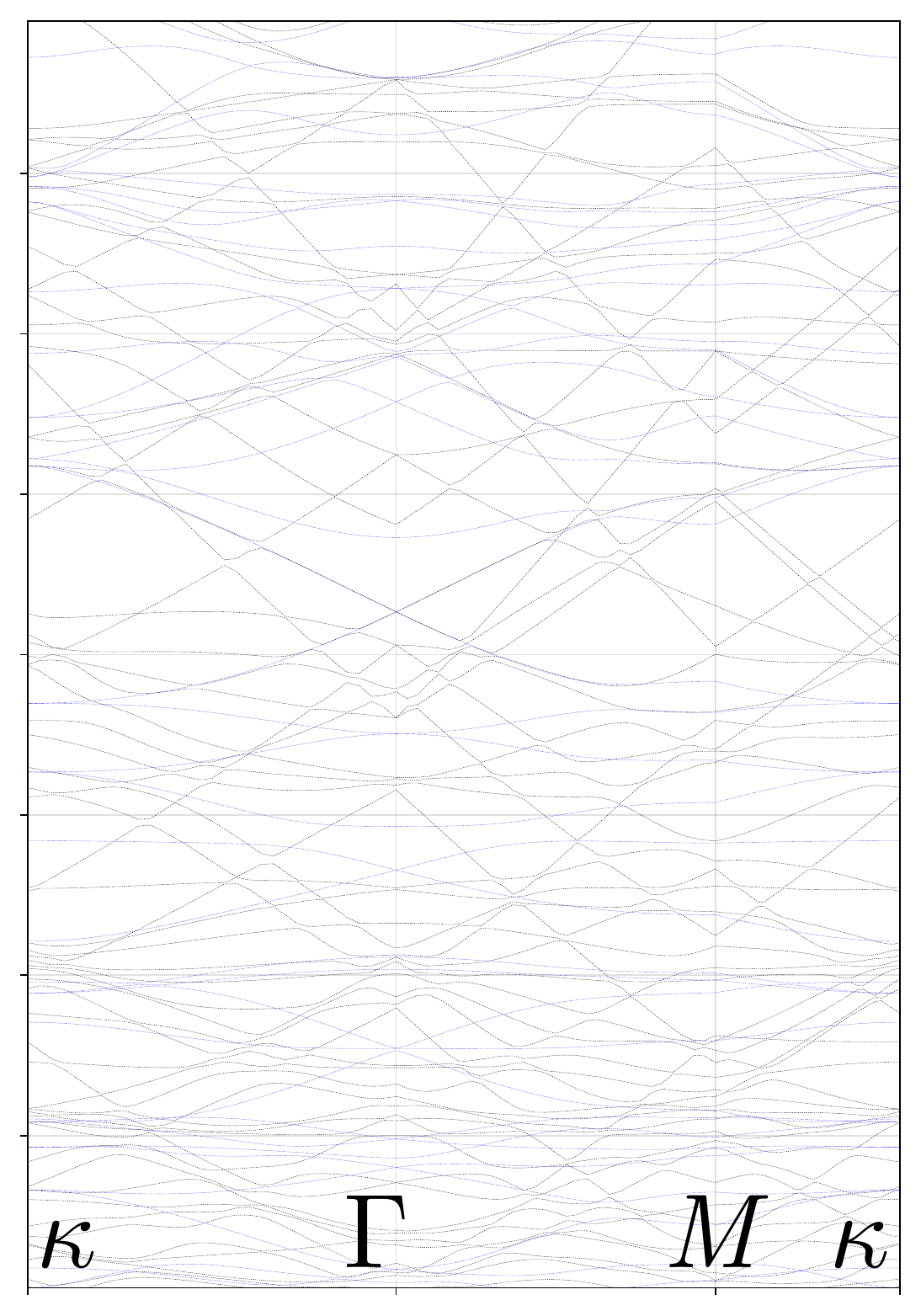}\caption{$\f{1}{\ep} = 7$}\end{subfigure}
\begin{subfigure}[b]{0.15\linewidth}\includegraphics[width=\linewidth,trim={0.35cm 0cm 0.34cm 0cm},clip]{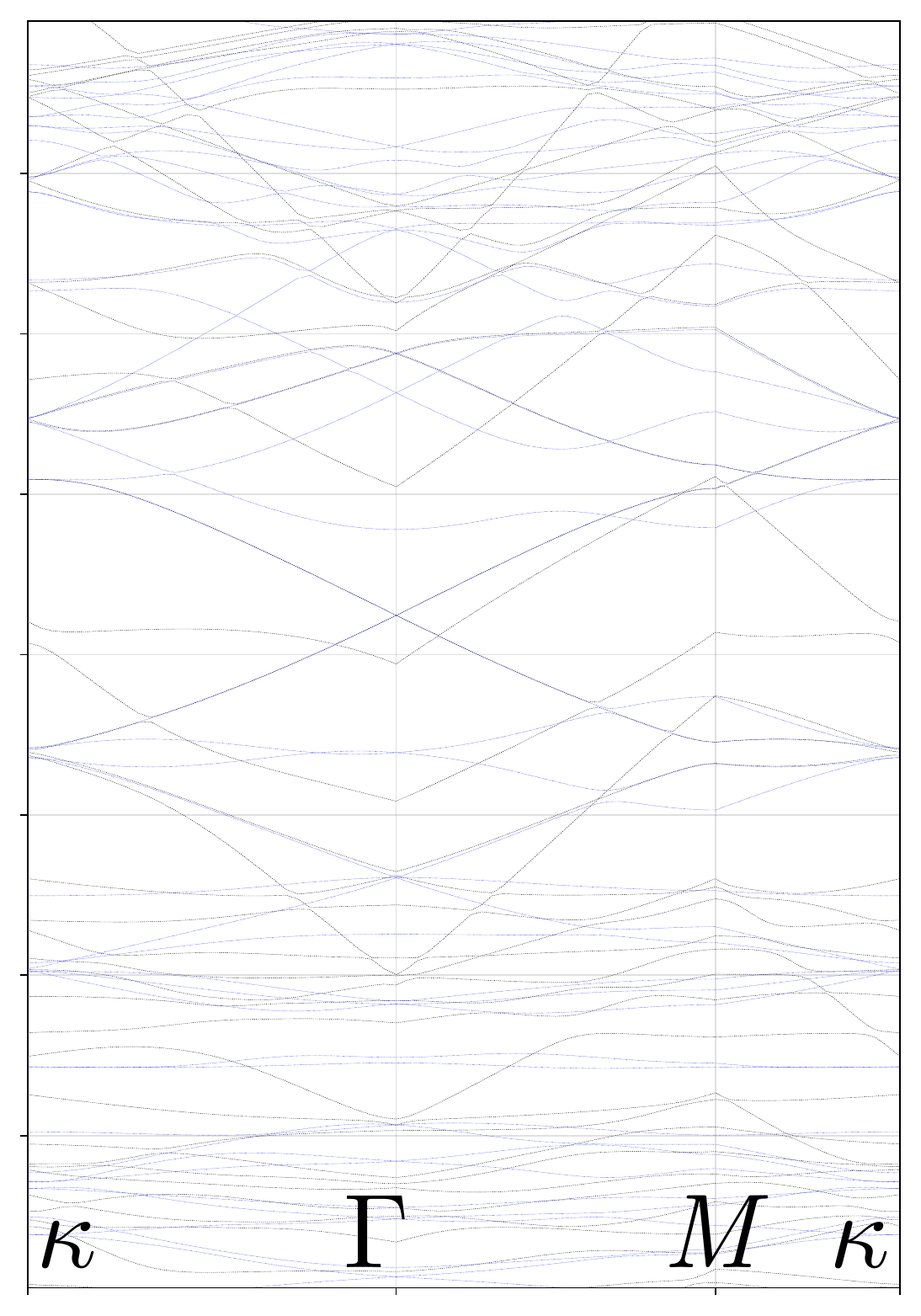}\caption{$\f{1}{\ep} = 11$}\end{subfigure}
\begin{subfigure}[b]{0.15\linewidth}\includegraphics[width=\linewidth,trim={0.35cm 0cm 0.34cm 0cm},clip]{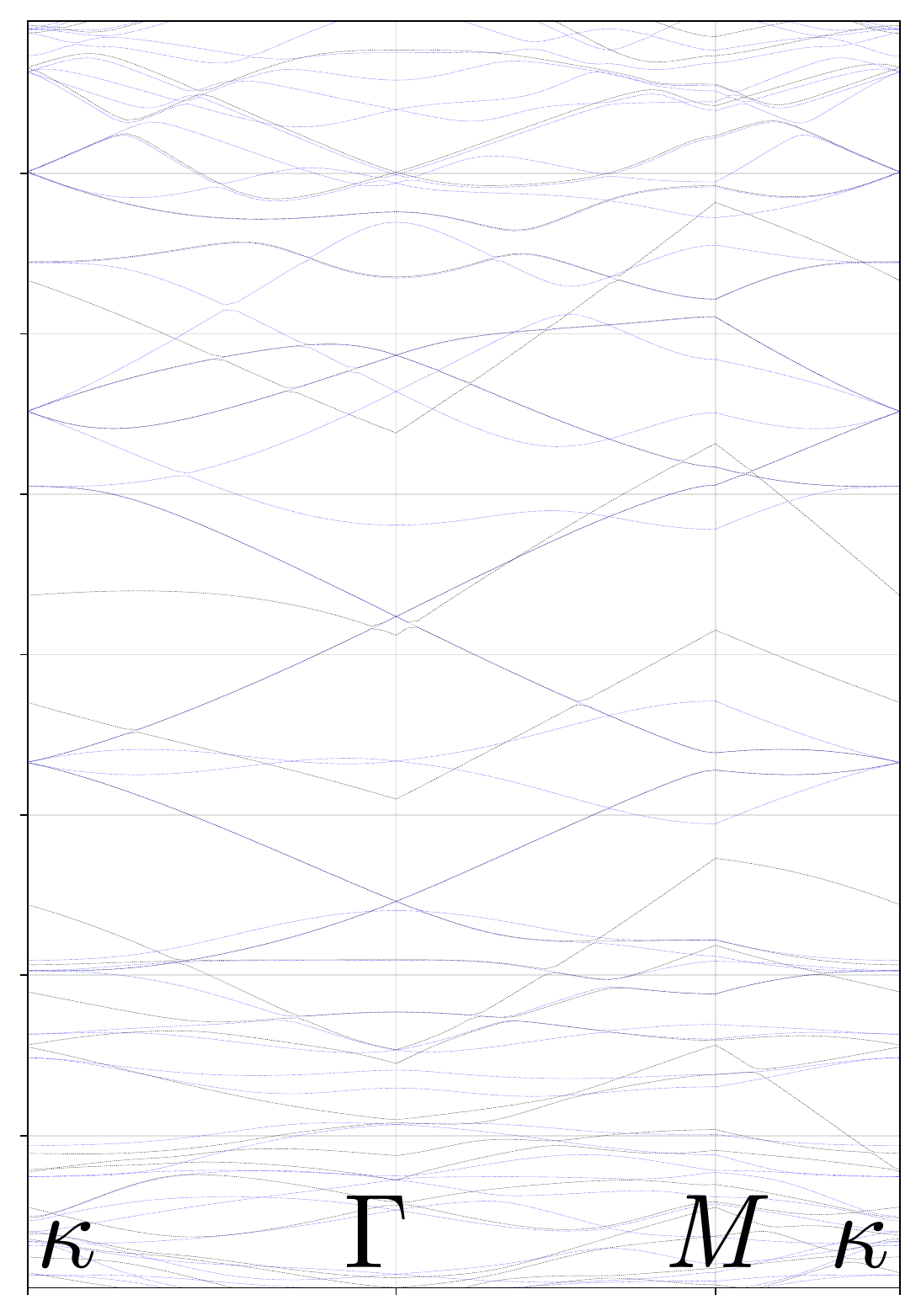}\caption{$\f{1}{\ep} = 14$}\end{subfigure}
\begin{subfigure}[b]{0.15\linewidth}\includegraphics[width=\linewidth,trim={0.35cm 0cm 0.34cm 0cm},clip]{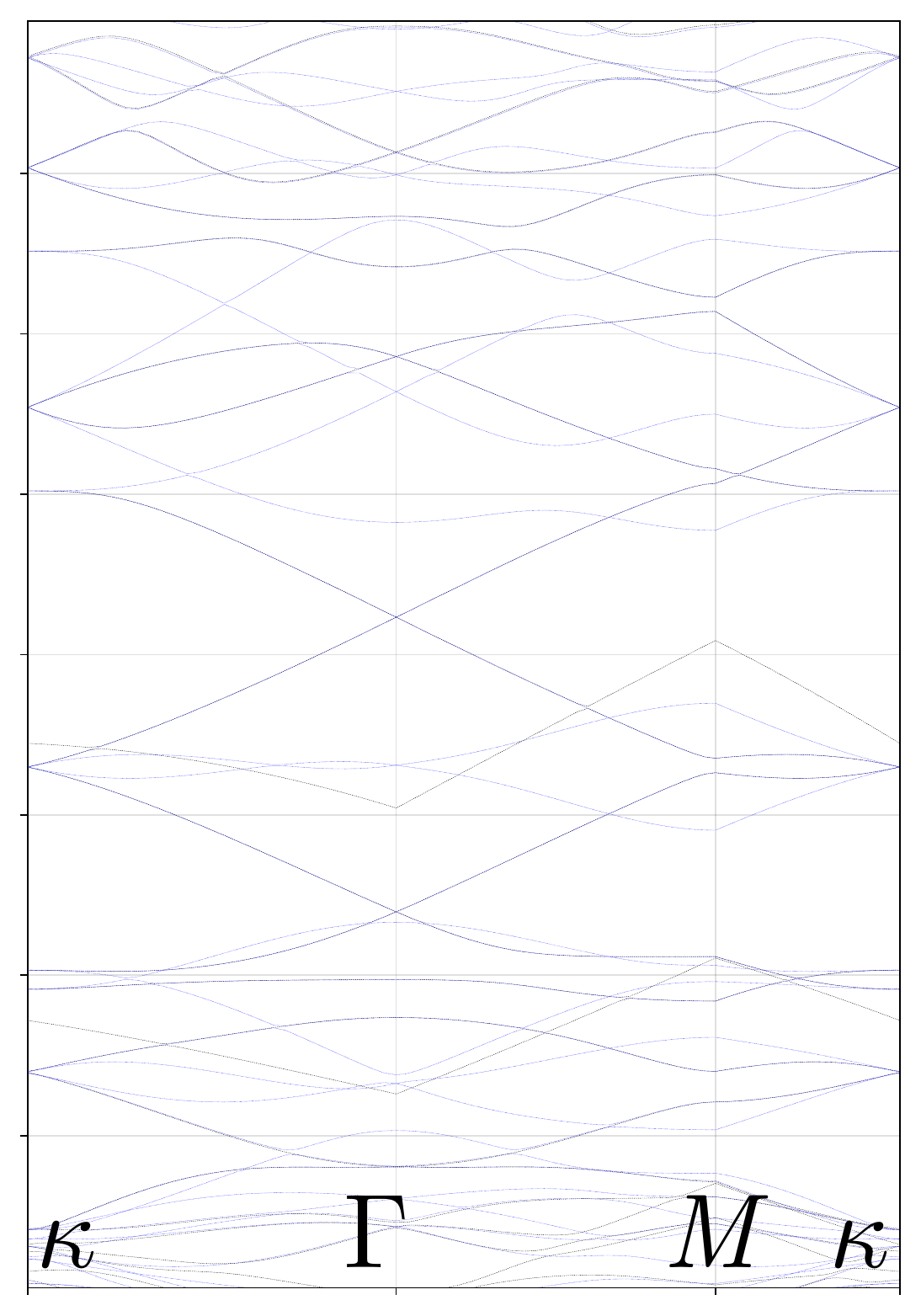}\caption{$\f{1}{\ep} = 19$}\end{subfigure} 
\restoregeometry
\caption{Letting $\ep$ vary, with $\cF_{1}$, $\nu = 2$, $V = V_{\text{honeycomb}}$ and $\lambda = 6$.}\label{fig:increase_ep_V0}
\end{center}
\end{figure}

\subsection{Varying $\lambda$}
\label{sub:Varying v}

On Figure~\ref{fig:increase_v}, we take $\ep = \f{1}{7}$, $\cF = \cF_1$, $V = V_{\text{honeycomb}}$, and we increase $\lambda$ from $0$ to $9$, see~\eqref{eq:def_lambda_V_honeycomb}. As $\lambda$ increases, we observe that the effective branches coming from the Dirac cone progressively lose accuracy relative to the exact branches.

\begin{figure}[h!]
\begin{center}
\newgeometry{left=0.5cm,right=0.5cm}  
\hspace{-6cm}
\begin{subfigure}[b]{0.15\linewidth}\includegraphics[width=\linewidth,trim={0.35cm 0cm 0.34cm 0cm},clip]{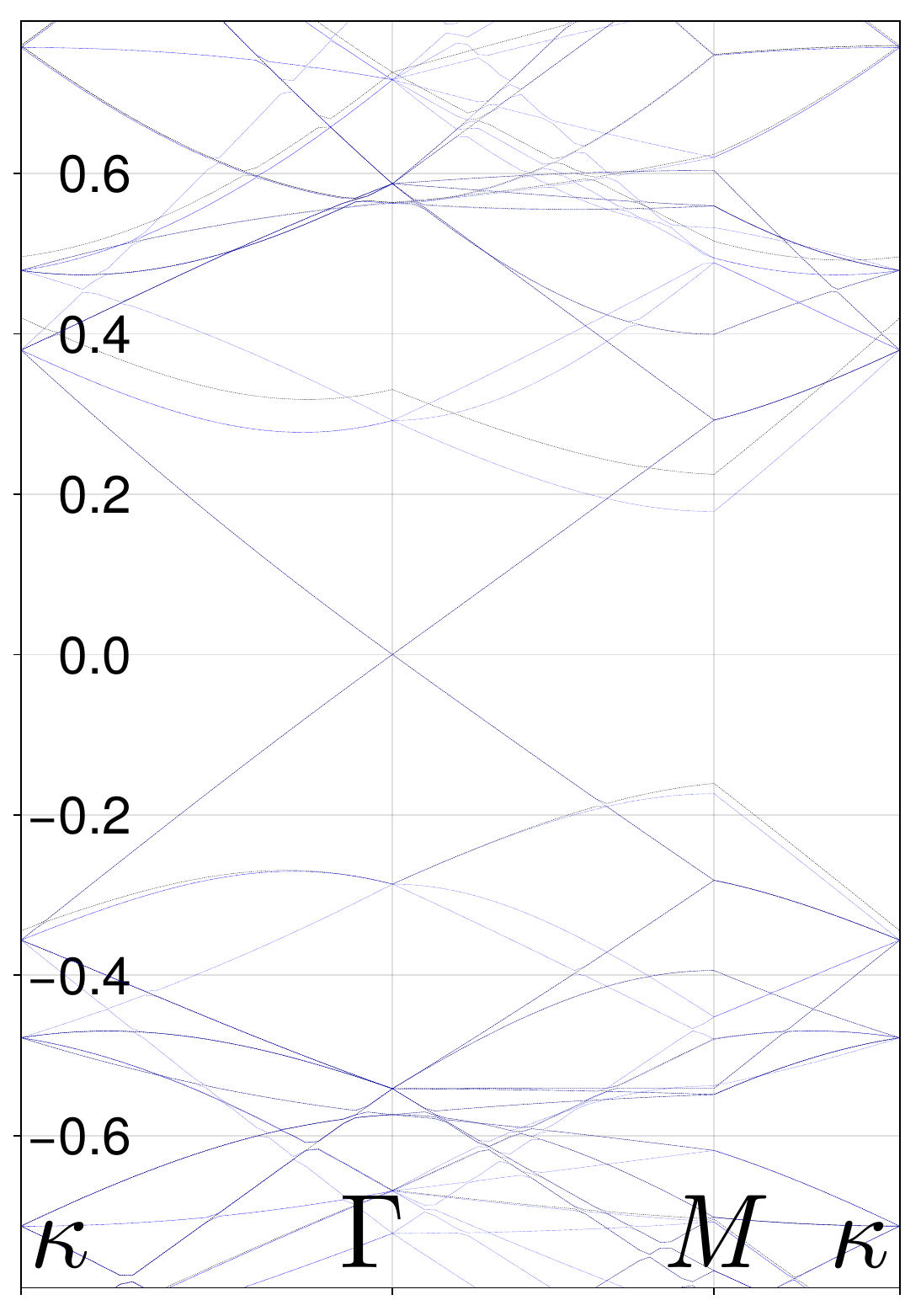}\caption{$\lambda = 0$}\end{subfigure}
\begin{subfigure}[b]{0.15\linewidth}\includegraphics[width=\linewidth,trim={0.35cm 0cm 0.34cm 0cm},clip]{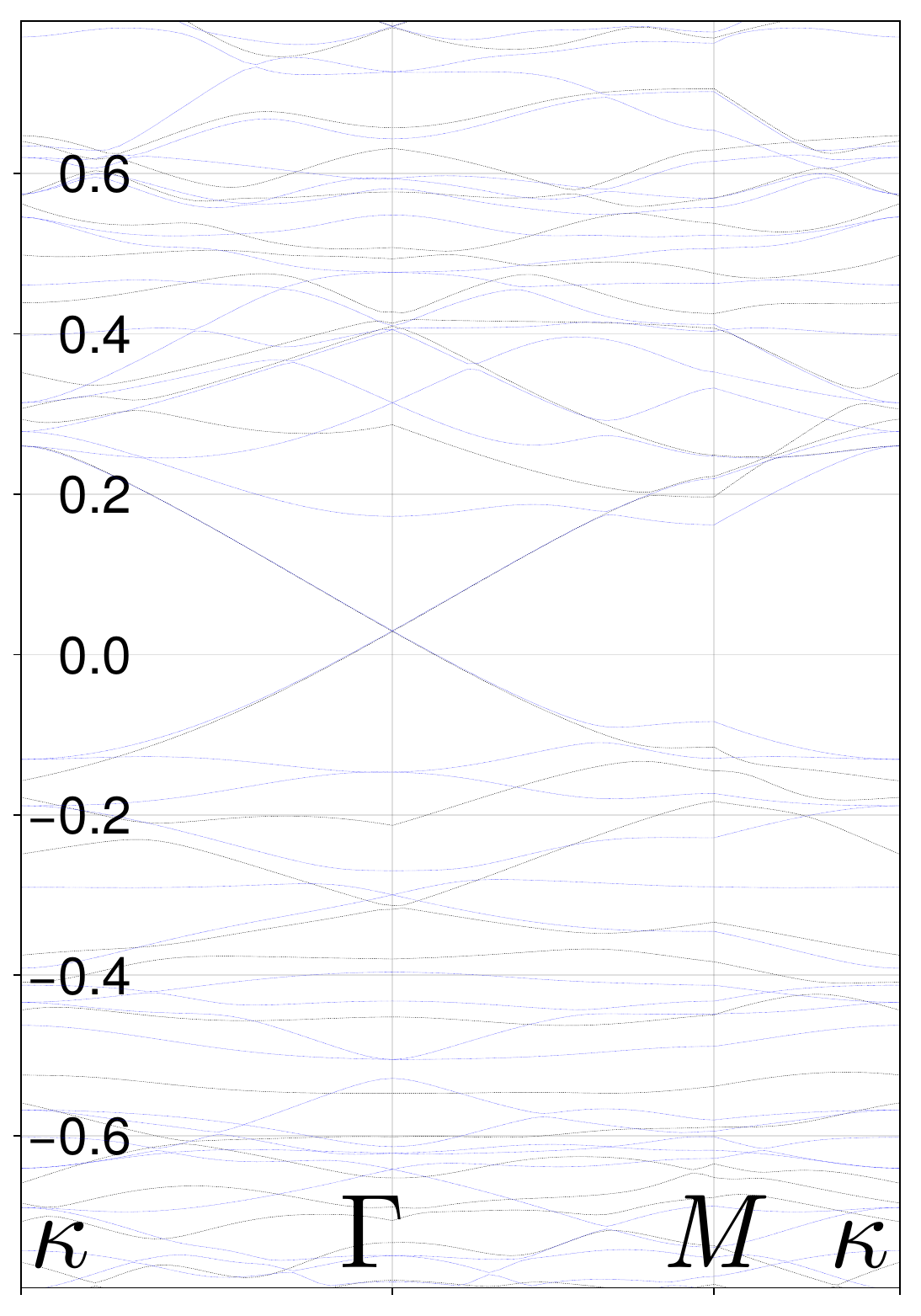}\caption{$\lambda = 5$}\end{subfigure}
\begin{subfigure}[b]{0.15\linewidth}\includegraphics[width=\linewidth,trim={0.35cm 0cm 0.34cm 0cm},clip]{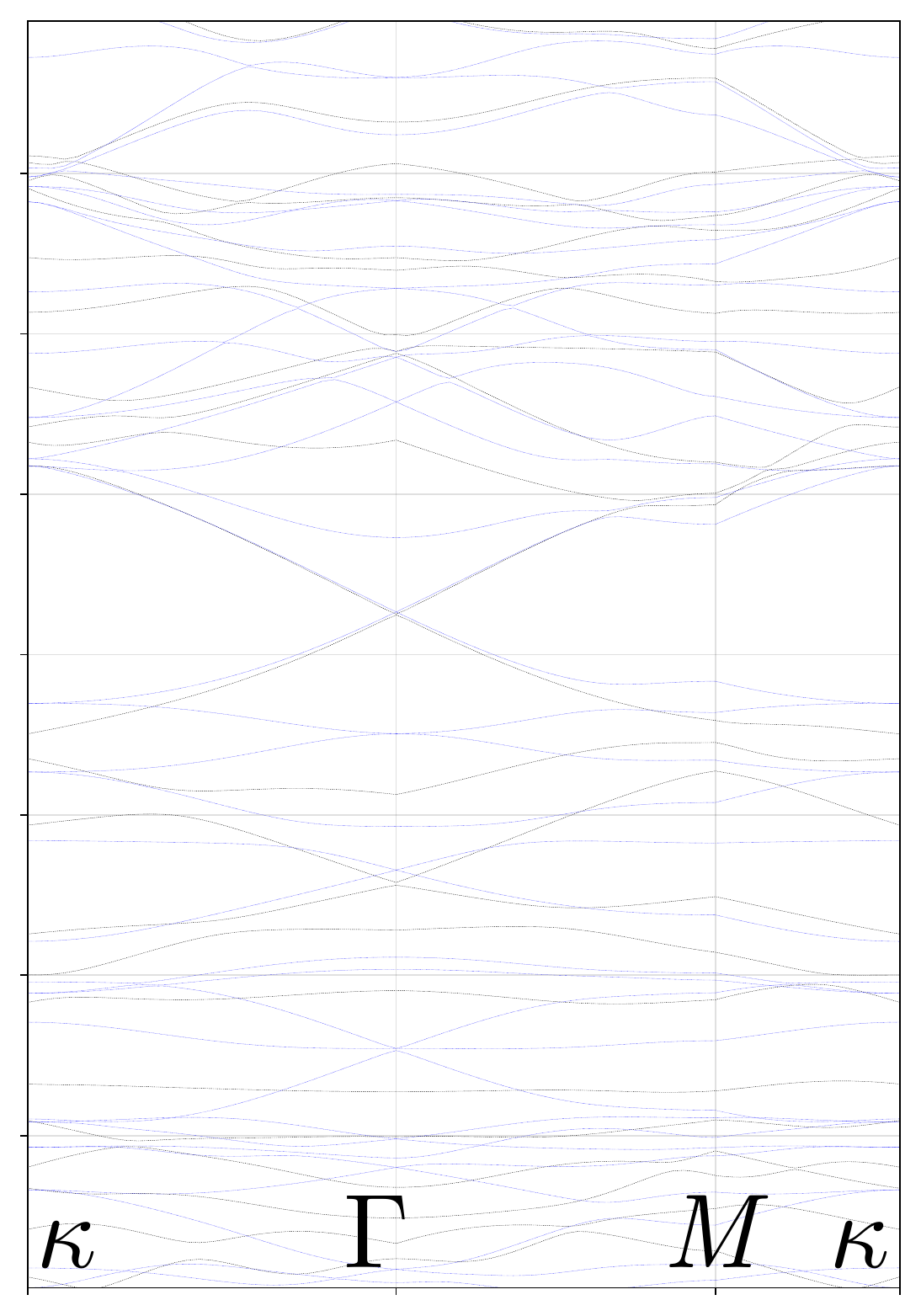}\caption{$\lambda = 6$}\end{subfigure}
\begin{subfigure}[b]{0.15\linewidth}\includegraphics[width=\linewidth,trim={0.35cm 0cm 0.34cm 0cm},clip]{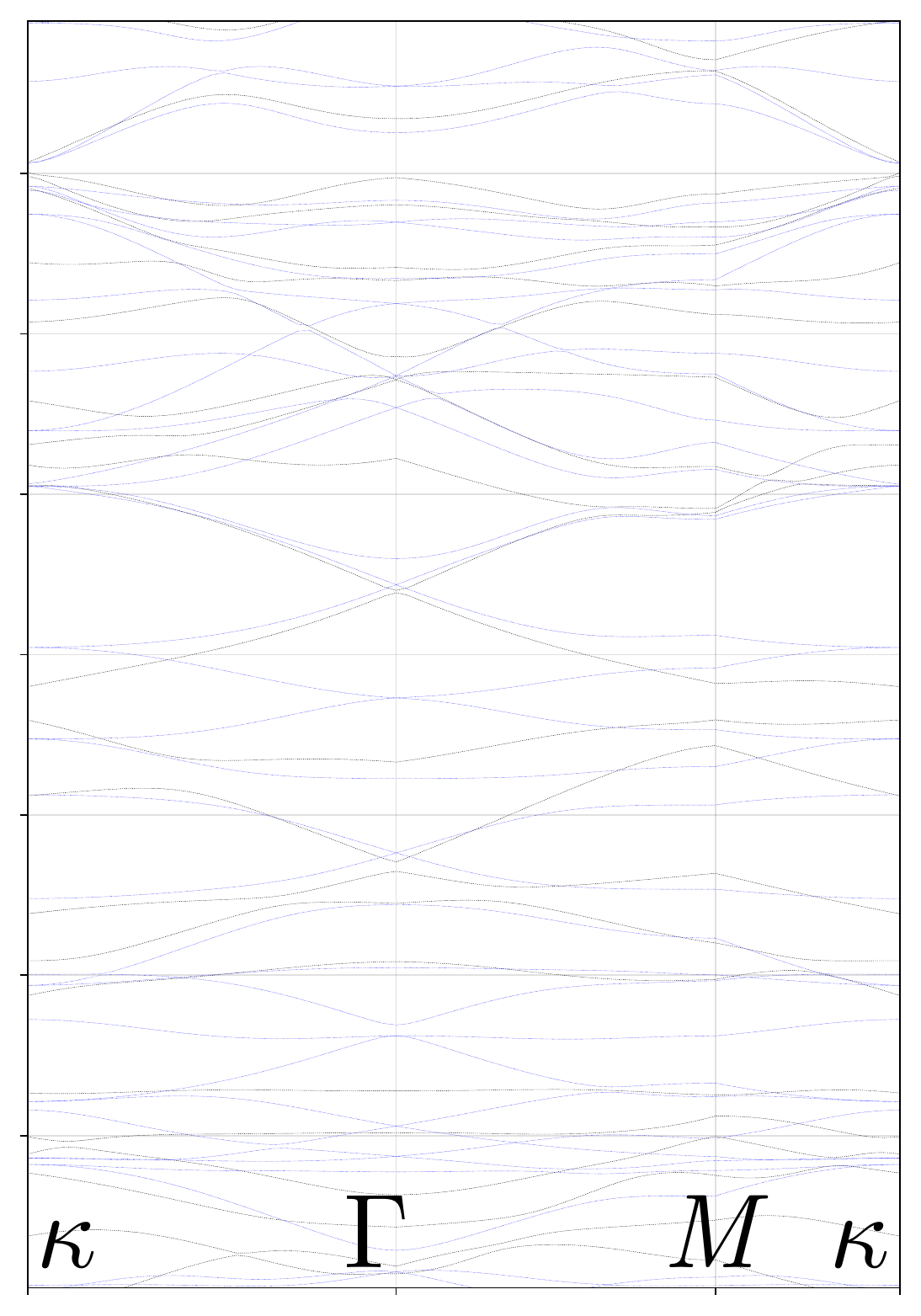}\caption{$\lambda = 7$}\end{subfigure}
\begin{subfigure}[b]{0.15\linewidth}\includegraphics[width=\linewidth,trim={0.35cm 0cm 0.34cm 0cm},clip]{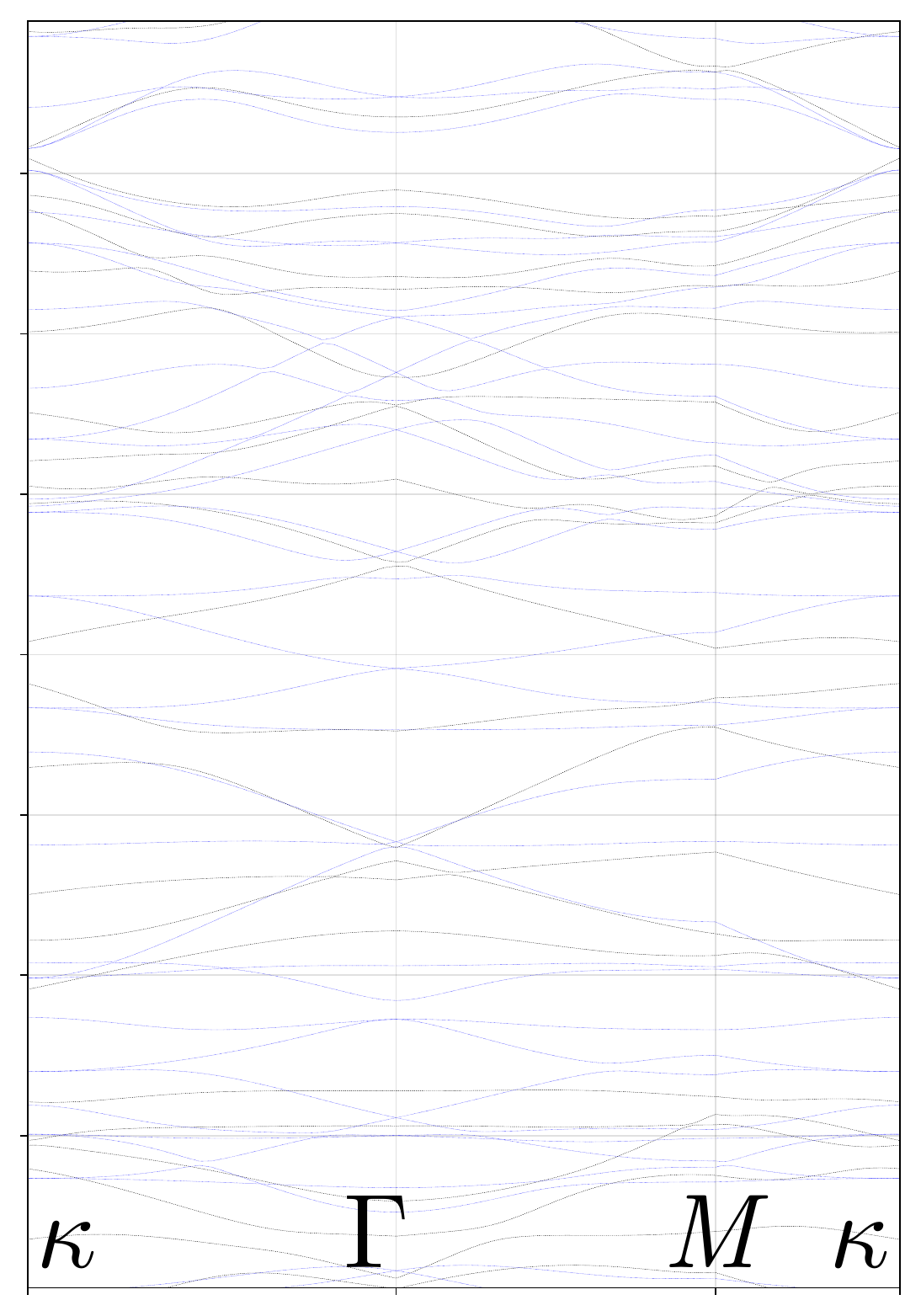}\caption{$\lambda = 8$}\end{subfigure}
\begin{subfigure}[b]{0.15\linewidth}\includegraphics[width=\linewidth,trim={0.35cm 0cm 0.34cm 0cm},clip]{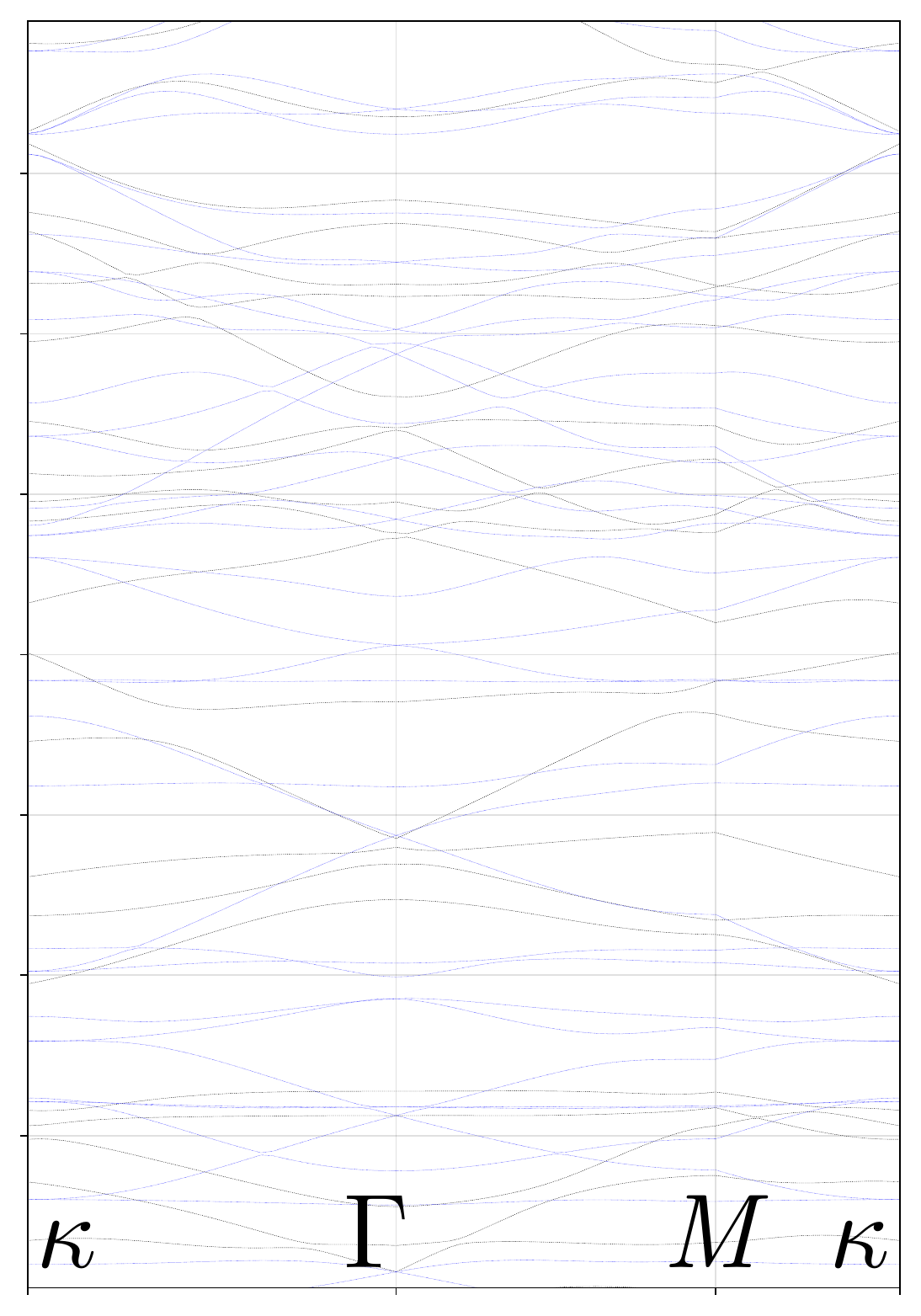}\caption{$\lambda = 9$}\end{subfigure} 
\restoregeometry
\caption{Letting $\lambda$ vary, which is the coupling constant of $V$, with $V = V_{\text{honeycomb}}$, $\ep = \f{1}{7}$, $\nu = 2$ and $\cF = \cF_1$.}\label{fig:increase_v}
\end{center}
\end{figure}

\subsection{Asymptotics}%
\label{sub:Asymptotics}


We will use the relative distance
\begin{align*}
	\d(\p,\phi) :=  \f{\mymin{\gamma \in [0,2\pi[}\nor{\p - e^{i\gamma} \phi}{L^2\per(\Omega)}}{\f{1}{2} \pa{\nor{\p}{L^2\per(\Omega)}  + \nor{\phi}{L^2\per(\Omega)}  }}.
\end{align*}
Since the eigenvectors $\w_1$ and $\w_2$ defined in~\eqref{eq:w1w2} are exact when $k=0$ and $V = 0$, we can consider that the two sources of error are the fact that $k \neq 0$ and $\lambda \neq 0$. Hence we want to better estimate the errors in $k$ and $\lambda$, independently from each other. Thus on Figure~\ref{fig:Vk_varies}, we represent
\begin{align}\label{eq:dist_eigenvectors} 
	\d\pa{\Phi^\ep_k , \cJ \alpha^\ep_k},
\end{align}
in the two following cases.
\begin{itemize}
\item On the left, we take $\lambda =0$, and $k = \mu K$ for instance, and we let $\mu \ge 0$ vary. Since $\lambda = 0$ then $V = 0$ and we only see the error coming from the fact that we are not at $k=0$. As expected from variational pertubation theory, the error is proportional to $\ab{\mu}^{\ell +1}$ for models of order $\ell$ (we precise that $\cF^6$ is a model with order $\ell = 0$). Both for $\mu$ small and larger, the effective operator coming from perturbation theory (which is $\cF_1$) is slightly more precise than the one coming from adding excited states in the variational space (which is $\cF^6$). We can compare the results of those two operators because the corresponding families have the same number of microscopic states.
\item On the right, we take $k=0$ and let $\lambda$ vary, so we only see the error coming from the fact that we are not at $V=0$. Here, we see that for $\lambda$ small enough, the effective operator built from first order pertubative vectors ($\cF_1$) is slightly better than the effective operator coming from excited states ($\cF^6$). However, this accuracy order is reversed for large $V$'s.
\end{itemize}

\begin{figure}[h!]
\begin{center}
\includegraphics[height=5cm,trim={0cm 0cm 0cm 0cm},clip]{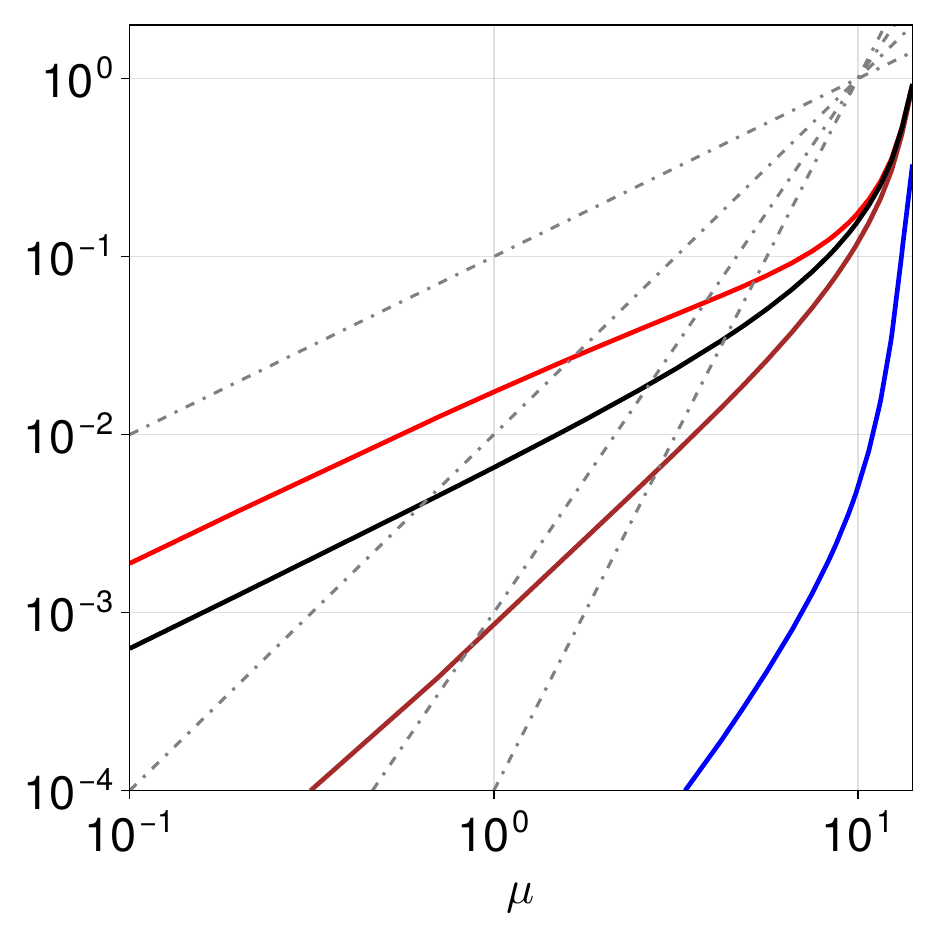} 
\includegraphics[height=5cm,trim={0cm 0cm 0cm 0cm},clip]{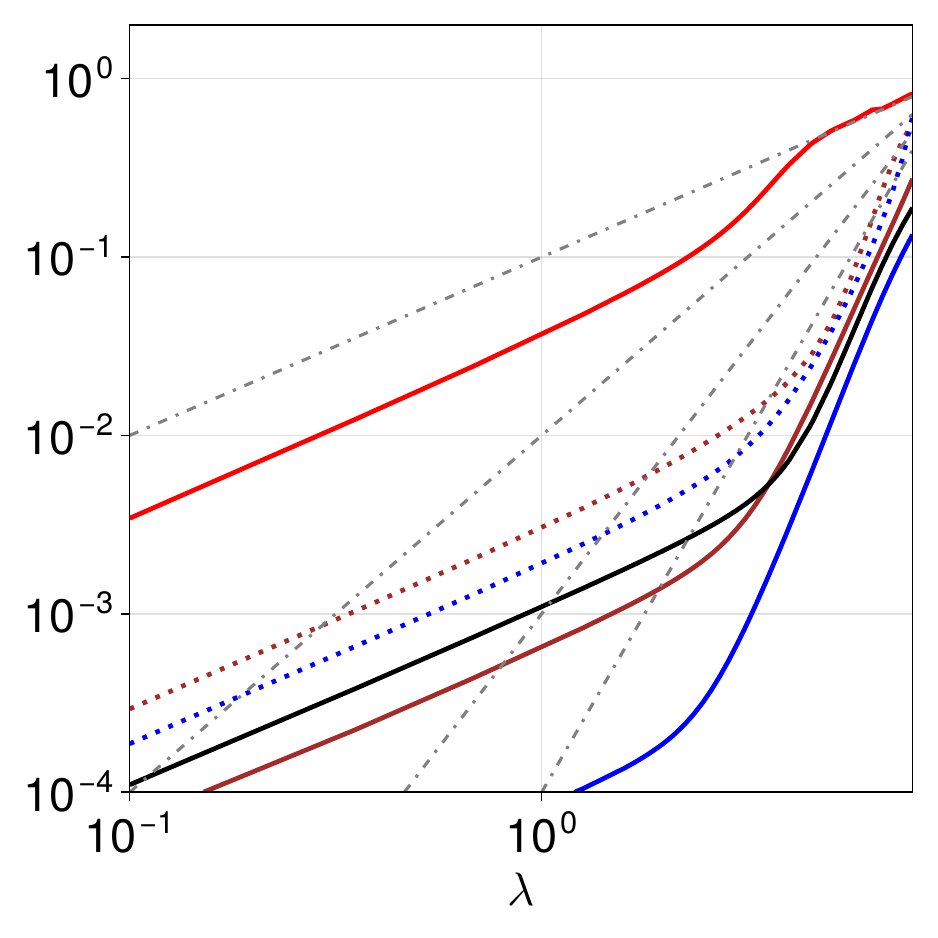} 
\includegraphics[height=5cm,trim={2.9cm 0cm 2.9cm 2cm},clip]{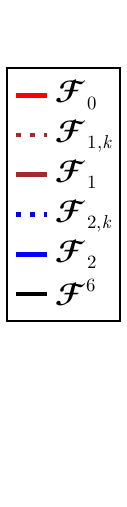}
\end{center}
\caption{We take $V = V_{\text{honeycomb}}$, $\nu = 2$, $\f{1}{\ep} = 7$ and display the relative distance on eigenvectors~\eqref{eq:dist_eigenvectors} (between the exact and effective models). On the left, we take $\lambda = 0$ and let $k = \mu K$ vary on a one-dimensional segment starting from $k = 0$, we do not provide $\cF_{\ell,k}$ and the curve of $\cF_2$ is below the other ones and we do not see it. On the right we take $ k =0$ and we let $\lambda$ (the intensity of $V$) vary. The grey lines are reference curves $(x/10)^{a}$ for $a \in \{1,2,3,4\}$ to enable asymptotic behavior comparisons.} \label{fig:Vk_varies}
\end{figure}

\section{Conclusions}%
\label{sub:Conclusion}

The goal of this document was to obtain accurate effective operators for the Bloch eigenmodes of the two-scales exact Schrödinger operator~\eqref{eq:exact_ep_inv_omega}. 

\subsection{Theoretical results}%
\label{sub:Theoretical results}

We recall that Proposition~\ref{prop:weak_conv} does not provide an approximation estimate on eigenvectors. The approximate eigenvectors given by families $\cF_\ell$ of high order are not expected to give a better $\ep$ order of convergence to the true solution, it is still expected to be of order 1 in $\ep$, for any $\ell$. The proposition was rather presented as a justification for the use of the derived effective Hamiltonians.

Our main theoretical result is Theorem~\ref{thm:main_derivative_enriched_validity_time_scale}. It shows that when $V = 0$, the validity time of the approximation given by our new effective operator is $\ep^{-2(\ell +1) + \eta}$, where $\eta > 0$, for families $\cF_\ell$ containing all momentum derivatives of Dirac eigenfunctions up to order $\ell$.

\subsection{Numerical results}%
\label{sub:Numerical results}

Here are the main conclusions brought by the simulations.
\begin{itemize}
\item The precision given by the massless Dirac operator can be significantly improved to approach the eigenmodes of the exact Bloch operator, as we saw on Figures~\ref{fig:increase_ell} and~\ref{fig:increase_ell_v_non_zero}
\item The cut-off $\nu$ needs to be well-chosen to avoid spectral pollution of the effective models. We can in general choose $\nu = 2$, as explained in Section~\ref{sub:The cutoff}
\item We are not retaining the $k$-dependent families $\cF_{1,k}$ and $\cF_{2,k}$ because of non-smooth behaviors and lack of accuracy seen in Figures~\ref{fig:increase_ell},~\ref{fig:increase_ell_v_non_zero} and~\ref{fig:Vk_varies}
\item As seen in Section~\ref{sub:Asymptotics}, the effective operator $\cF_1$ built with perturbative (with respect to the momentum) vectors at first order has most of the time better precision than the effective operator $\cF^6$ coming from adding eigenstates at $K$, and hence we recover the same kind of conclusions as in~\cite[Section 6.3]{GarSta24}.
\item Consider $d_\ell(\mu) := \d\bpa{\Phi^\ep_{\mu K} , \cJ \alpha^\ep_{\mu K}}$ when $\mu \to 0$. When $\lambda = 0$, as we see on the numerical evidence provided in Figure~\ref{fig:Vk_varies}, one can theoretically expect that $d_\ell(\mu) \sim_{\mu \to 0} C \ab{\mu}^{1 + \ell}$ for families of order $\ell$ ($\cF_0$ and $\cF^6$ are of order $0$)
\end{itemize}
We believe that the use of the proposed effective operators could be of interest in some applications to model low-energy Dirac fermions, when high precision is needed.

\section{Proof of Proposition~\ref{prop:weak_conv}}%
\label{sec:proof of prop}

In this section, we derive the effective operator $\bbH\elk$ written in~\eqref{eq:effective_op_general} from the exact operator $H\elk$ expressed in~\eqref{eq:exact_op}, by using the same method as in~\cite{CanGarGon23b}.

We define $a := \icol{a_1 \\ a_2}$ and for $m \in \Z^2$, we will use the notation $m a := m_1 a_1 + m_2 a_2$. Our Fourier transform will be the operator $\bbF : L^2\per(\Omega) \rightarrow \ell^2(\Z^2)$ such that for any $f \in L^2\per(\Omega)$,
\begin{align}\label{eq:def_fft}
\pa{\bbF f}_m = \widehat{f}_m =  \f{1}{\sao} \int_{\Omega} f(y) e^{-i m a^* y} \d y, \qquad f(x) = \sum_{m \in \Z^2} \widehat{f}_m \f{e^{ima^* x}}{\sao}.
\end{align}

The following result is classical and explains the decoupling between macroscopic and microscopic scales
\begin{lemma}
For every $N\in\N\cup\{0\}$, there is $C_N>0$, depending only on $N$ and the fixed lattice, such that for every $g\in H^N\per(\Omega)$, every $f\in L^2\per(\Omega)$ and every $\ep\in]0,1[\cap(1/\N)$,
\begin{align}\label{eq:oscillating} 
\ab{\int_\Omega g(y)f\pa{\tfrac{y}{\ep}}\,\d y-\ab{\Omega}^{-1}\int_\Omega g\int_\Omega f}
\le C_N\ep^N\nor{g}{H^N\per(\Omega)}\nor{f}{L^2\per(\Omega)}.
\end{align}
\end{lemma}
\begin{proof}
First suppose that $f$ and $g$ are trigonometric polynomials. Since $\ep^{-1}\in\N$, the Fourier convention~\eqref{eq:def_fft} gives
\begin{align*}
\int_\Omega g(y)f\pa{\tfrac{y}{\ep}}\,\d y
&=\ab{\Omega}^{-1}\sum_{m,n\in\Z^2}\widehat g_m\widehat f_n\int_\Omega e^{i(m+\ep^{-1}n)a^*\cdot y}\,\d y =\sum_{n\in\Z^2}\widehat g_{-\ep^{-1}n}\widehat f_n\\
&=\ab{\Omega}^{-1}\int_\Omega g\int_\Omega f+\sum_{n\in\Z^2\setminus\{0\}}\widehat g_{-\ep^{-1}n}\widehat f_n.
\end{align*}
Cauchy--Schwarz and Parseval's identity yield
\begin{align*}
\ab{\sum_{n\in\Z^2\setminus\{0\}}\widehat g_{-\ep^{-1}n}\widehat f_n}^2
\le\nor{f}{L^2\per(\Omega)}^2\sum_{n\in\Z^2\setminus\{0\}}\ab{\widehat g_{-\ep^{-1}n}}^2
\end{align*}
For $n\neq0$, one has $\ab{\ep^{-1}n}\ge\ep^{-1}$. Consequently,
\begin{align*}
\sum_{n\in\Z^2\setminus\{0\}}\ab{\widehat g_{-\ep^{-1}n}}^2
&\le\ep^{2N}\sum_{n\in\Z^2\setminus\{0\}}\ab{\ep^{-1}n}^{2N}\ab{\widehat g_{-\ep^{-1}n}}^2\\
&\le\ep^{2N}\sum_{m\in\Z^2}(1+\ab{m}^2)^N\ab{\widehat g_m}^2 \le C_N\ep^{2N}\nor{g}{H^N\per(\Omega)}^2,
\end{align*}
where the last inequality follows from the equivalence of the Fourier definition of the periodic Sobolev norm and the usual $H^N\per(\Omega)$ norm for the fixed lattice. Combining these estimates proves~\eqref{eq:oscillating} for trigonometric polynomials. Finally, $f\mapsto f(\cdot/\ep)$ is an isometry on $L^2\per(\Omega)$ when $\ep^{-1}\in\N$, so density extends the inequality to all $f\in L^2\per(\Omega)$ and $g\in H^N\per(\Omega)$.
\end{proof}
We will denote by $O(\ep^\infty)$ any quantity $X$ such that for any $N \in \N$, there exists $c_N > 0$ such that for any $\ep \in ]0,1[ \cap (1/\N)$, $\ab{X} \le c_N \ep^N$.


For any functions $f$ and $g$ and any $p,q \in \R^2$, we have
\begin{align}\label{eq:deriv_one} 
	(-i\na_{p+q}) (fg) = f (-i\na_p) g + g (-i\na_q) f.
\end{align}
From~\eqref{eq:deriv_one} and $\na_{\ep^{-1} K} Q = \ep^{-1} Q \na_K$ we deduce that
\begin{multline*}
(-i\na_{\ep^{-1}K + k} + A) (g(y) \vp \xep) \\
= \vp \xep ((-i\na_k + A) g) (y) + \ep^{-1}g(y) \pa{-i\na_K \vp} \xep,
\end{multline*}
and thus
\begin{multline*}
(-i\na_{\ep^{-1}K + k} + A)^2 (g(y) \vp \xep) = \vp\xep ((-i\na_k + A)^2 g) (y) \\
+ 2 \ep^{-1} ((-i\na_k + A) g)(y) \cdot (-i\na_K  \vp) \xep + \ep^{-2} g(y) \pa{(-i\na_K )^2 \vp} \xep.
\end{multline*}
This enables to make the computation
\begin{multline*}
	H_k^{\ep} \pa{g(y) \vp \xep} = g(y) \pa{\pa{h_K - E\fer}\vp} \xep \\
	+ \vp \xep \pa{\tfrac 12 \ep^2 (-i\na_k + A)^2 + \ep V} g(y) \\
	+ \ep \bpa{\pa{-i\na_k +  A }g}(y) \cdot (-i\na_K \vp) \xep.
\end{multline*}
For $\xi,\vp\in D(h_K)$, one has $(h_K-E\fer)\vp\in L^2\per(\Omega)$. Moreover, since $v\in L^p\per(\Omega)$ with $p>1$, periodic elliptic regularity gives $D(h_K)\subset W^{2,r}\per(\Omega)\subset L^\infty\per(\Omega)$ for some $r>1$. Consequently, $\overline{\xi}(h_K-E\fer)\vp$, $\overline{\xi}(-i\na_K)\vp$ and $\overline{\xi}\vp$ belong to $L^2\per(\Omega)$, so the oscillatory estimate~\eqref{eq:oscillating} applies to all the microscopic factors below. We proceed by taking $\Omega$-periodic functions $f, g$ and $\xi, \vp\in D(h_K)$ and computing
\begin{align*}
& \ep^{-1} \ab{\Omega} \ps{f Q \xi, H_k^{\ep} g Q \vp}_{L^2\per(\Omega)} \\
& \qquad = \ep^{-1}\ab{\Omega}\int_{\Omega} \pa{\overline{f}g}(y) \pa{\overline{\xi}\pa{h_K - E\fer} \vp}\xep \d y\\
& \qquad \qquad + \ab{\Omega}\int_{\Omega} \pa{\overline{f}((-i\na_k + A) g)}(y) \cdot\pa{\overline{\xi} (-i\na_K \vp)}\xep \d y  \\
& \qquad \qquad + \ab{\Omega}\int_{\Omega}\pa{\overline{f} \pa{\tfrac 12 \ep (-i\na_k + A)^2 + V} g}(y) \pa{\overline{\xi }\vp}\xep  \d y \\
& \qquad  \underset{\substack{\eqref{eq:oscillating}}}{=} \ep^{-1}\ps{f, g} \ps{ \xi, \pa{h_K - E\fer} \vp} + \ps{f, (-i\na_k + A) g} \cdot \ps{\xi, -i\na_K \vp} \\
& \qquad\qquad  +  \ps{f , \pa{\tfrac 12 \ep (-i\na_k + A)^2 + V} g} \ps{\xi, \vp} + O(\ep^\infty).
\end{align*}
Then we take $\xi = \psi_a$, $f = \alpha_a$ and $\vp = \psi_b$, $g = \beta_b$, which concludes the proof of Proposition~\ref{prop:weak_conv}.

\section{Computation of matrix coefficients\\using graphene symmetries}%
\label{sec:Computation of matrix coefficients using graphene symmetries}

In this section, our goal is to compute scalar products of the kind
\begin{multline*}
\ps{Y_1 w_a, Y_2 (-i\partial_{K,j}) w_b}, \qquad \ps{Y_1 (-i\partial_{K,n}) w_a, Y_2 (-i\partial_{K,j}) w_b}, \\
\ps{Y_1 (-i\partial_{K,q}) w_a, Y_2 (-i\partial_{K,j}) Y_3 (-i\partial_{K,m}) w_b},
\end{multline*}
with respectively one, two or three derivatives, for particular operators $Y_j$'s of interest involved in graphene. This will provide the tools to compute all the matrices involved in~\eqref{eq:def_matrices}, by using symmetries. See~\cite{WinZul10} for a systematic study of this kind of quantities.

\subsection{One derivative}%
\label{sub:One derivative}

The case of one derivative is well-known~\cite{FefWei12}, ~\cite[(S3)]{CanGarGon23b}. If $R_{\f{2\pi}3} m = y m$ with $m \in \C^2$, then
\begin{align}\label{eq:eigenvec_R} 
m = \left\{
\begin{array}{ll}
t \icol{1 \\ -\eta i}  & \mbox{if $y = \omega^\eta$ and $\eta \in \{-1,1\}$, for some $t \in \C$,} \\
0 & \mbox{otherwise},
\end{array}
\right.
\end{align}
and this leads to the computation of $\ps{\w_a, (-i\na_K) \w_b}$ for instance, see~\eqref{eq:vF}. Here we recall the computation when only one derivative is involved. As in~\cite{CanGarGon23b}, we define
\begin{align*}
\pa{\mathfrak{R}f}(x_1,x_2) := f(x_1,-x_2).
\end{align*}
For $s \ge 0$, we will use the spaces of pseudo-periodic functions
\begin{align*}
	H_k^s(\Omega) := \acs{f \in H^s_{\text{loc}}(\R^2) \;\Big\rvert\;
	f(x + ma) = e^{ik \cdot ma} f(x),\ \forall m \in \Z^2},
\end{align*}
and write $L_k^2(\Omega):=H_k^0(\Omega)$. For every operator $Y_j$ appearing below, we set $S_j := e^{iKx} Y_j e^{-iKx}$ on $L_K^2(\Omega)$. All commutation relations involving the $S_j$ and the spatial symmetries are understood on the spaces specified in the hypotheses. In particular, when a derivative is applied after $S_j$, we require $S_j$ to map $L_K^2(\Omega)$ continuously into $H_K^1(\Omega)$.

\begin{lemma}[One derivative]\label{lem:one_der} 
Take two bounded operators $Y_1$ and $Y_2$ of $L^2\per(\Omega)$. Assume that, on $L_K^2(\Omega)$, $S_1$ and $S_2$ commute with $\cR_{\f{2\pi}{3} }$ and with $\cC \cP$.
\begin{itemize}
	\item There exist $v \in \C$ such that for any $j,a \in \{1,2\}$, $\ps{Y_1 w_a, Y_2 (-i\partial_{K,j}) w_a} =0$ and
\begin{align*}
\ps{Y_1 w_1, Y_2 (-i\na_{K}) w_2} = v \icol{1 \\ -i}, \qquad\qquad  \ps{Y_1 w_2, Y_2 (-i\na_{K}) w_1} = \overline{v} \icol{1 \\ i}.
\end{align*}
\item Take $\iota_a \in \{-1,1\}$ for any $a \in \{1,2\}$ and assume that for any $a \in \{1,2\}$, $S_a \mathfrak{R} = \iota_a \mathfrak{R} S_a$. Then $v \in \R$ if $\iota_1 \iota_2 = 1$ and $v \in i\R$ if $\iota_1 \iota_2 = -1$.
\end{itemize}
\end{lemma}
\begin{proof}
We recall that
 \begin{align}\label{eq:perm_partial_exp} 
\partial_{K,n} e^{-iKx} = e^{-iKx} \partial_n.
 \end{align}
so
\begin{align}\label{eq:rel_Y_S} 
\ps{Y_1 w_a, Y_2 (-i\partial_{K,j}) w_b} = -i \ps{S_1 \phi_a, S_2 \partial_j \phi_b}.
\end{align}
We will prove the equivalent statements on those last quantities. We recall that
\begin{align}\label{eq:transform_cR} 
\cR_\theta \na = R_{-\theta} \na \cR_\theta.
\end{align}
We define the vectors $m^{(ab)}\in \C^2$ by $m^{(ab)}_{j} := \ps{S_1\phi_a, S_2 \partial_j \phi_b}$. We have
\begin{align*}
m^{(ab)} &= \ps{ S_1 \phi_a, S_2 \na \phi_b} = \ps{\cR_{\f{2\pi}3} S_1 \phi_a, \cR_{\f{2\pi}3} S_2 \na \phi_b} = \ps{ S_1 \cR_{\f{2\pi}3} \phi_a,  S_2 \cR_{\f{2\pi}3} \na \phi_b} \\
 &\underset{\substack{\eqref{eq:transform_cR}}}{=} \; R_{-\f{2\pi}{3}}\ps{ S_1 \cR_{\f{2\pi}3} \phi_a, S_2 \na \cR_{\f{2\pi}3} \phi_b}  \underset{\substack{\eqref{eq:syms_phi}}}{=} \; \omega^{b-a} R_{-\f{2\pi}{3}} m^{(ab)}.
\end{align*}
Thus $R_{\f{2\pi}{3}} m^{(ab)} = \omega^{b-a} m^{(ab)}$. The only vector invariant under $R_{\f{2\pi}{3}}$ is $0$ so $m^{(11)} = m^{(22)} = 0$. Then $R_{\f{2\pi}{3}} m^{(12)} = \omega m^{(12)}$ implies $m^{(12)} \in \Ker \bpa{R_{\f{2\pi}{3}} - \omega} = \C \icol{1 \\ -i}$ so $m^{(12)} = \nu \icol{1 \\ -i}$ for some $\nu \in \C$. Similarly, $m^{(21)} = g \icol{1 \\ i}$ for some $g \in \C$. We define $\underline{1} := 2$ and $\underline{2} := 1$ and will use
\begin{align*}
\cP \cC  \na = - \na \cP \cC, \qquad \qquad \cP \cC \phi_a = \phi_{\underline{a}}.
 \end{align*}
We have
\begin{align*}
	\nu &= m^{(12)}_1 = \ps{S_1\phi_1, S_2 \partial_1 \phi_2} = \ps{\cP \cC S_2 \partial_1 \phi_2 , \cP \cC S_1\phi_1} = \ps{S_2 \cP \cC \partial_1 \phi_2 , S_1\cP \cC \phi_1} \\
	&=-\ps{S_2  \partial_1 \cP \cC\phi_2 , S_1\cP \cC \phi_1} = -\ps{S_2 \partial_1 \phi_1, S_1 \phi_2} = - \overline{\ps{S_1 \phi_2 , S_2 \partial_1 \phi_1}} \\
	&= - \overline{m^{(21)}_1} = - \overline{g}.
\end{align*}
In~\cite[Supplementary Material, Appendix I]{CanGarGon23b} it is proved that
\begin{align}\label{eq:transf_phi_R} 
\mathfrak{R} \phi_1 = \eta_{\mathfrak R} \phi_2
\end{align}
with a sign $\eta_{\mathfrak R} \in \{\pm 1\}$. We use $S_a \mathfrak{R} = \iota_a \mathfrak{R} S_a$ and moreover we have
\begin{align}\label{eq:com_R_der} 
\mathfrak{R} \partial_n = (-1)^{n+1} \partial_n \mathfrak{R},
\end{align}
also written $\mathfrak{R} \na = \sigma_3 \na \mathfrak{R}$. We compute
\begin{align*}
\nu &=m^{(12)}_1 = \ps{S_1  \phi_1, S_2 \partial_1 \phi_2}  =  \ps{\mathfrak{R} S_1  \phi_1, \mathfrak{R} S_2 \partial_1 \phi_2} =  \iota_1 \iota_2 \ps{ S_1 \mathfrak{R} \phi_1,  S_2 \mathfrak{R}\partial_1 \phi_2} \\
  &\underset{\substack{\eqref{eq:com_R_der}}}{=} \;   \iota_1 \iota_2 \ps{ S_1  \mathfrak{R}\phi_1,  S_2 \partial_1 \mathfrak{R}\phi_2} \underset{\substack{\eqref{eq:transf_phi_R}}}{=} \;  \iota_1 \iota_2 \ps{ S_1  \phi_{2},  S_2 \partial_1 \phi_{1}}=   \iota_1 \iota_2 m^{(2 1)}_{1} = - \iota_1 \iota_2\overline{\nu}.
\end{align*}
Hence $\nu \in i\R$ if $\iota_1 \iota_2 = 1$ and $\nu \in \R$ if $\iota_1 \iota_2 = -1$. We have $\ps{Y_1 w_1, Y_2 \na_K w_2} = \nu \icol{1 \\ -i}$ and $\ps{Y_1 w_2, Y_2 \na_K w_1} = -\overline{\nu} \icol{1 \\ i}$. Equivalently,
\begin{align*}
\ps{Y_1 w_1, Y_2 (-i\na_K) w_2} = -i\nu \icol{1 \\ -i}, \qquad \ps{Y_1 w_2, Y_2 (-i\na_K) w_1} = i\overline{\nu} \icol{1 \\ i}.
\end{align*}
Therefore the conclusion follows with $v=-i\nu$.
\end{proof}

\subsection{Two derivatives}%
\label{sub:Two derivatives}

We recall the matrix $D$ defined before Proposition~\ref{prop:order1_kindep}.

We start by showing the following result.
\begin{lemma}\label{lem:rotM}
Take $M \in \C^{2\times 2}$ and $y \in \{1,\omega,\overline{\omega}\}$. Then $R_{\f{2\pi}3} M R_{-\f{2\pi}3} = y M$ if and only if
\begin{align*}
M = \left\{
\begin{array}{ll}
s \1_2 + r \sigma_2 & \mbox{if $y = 1$, for some } s,r \in \C \\
s D  & \mbox{if $y = \omega$, for some $s \in \C$} \\
s D^*  & \mbox{if $y = \overline{\omega}$, for some $s \in \C$}.
\end{array}
\right.
\end{align*}
\end{lemma}
\begin{proof}
We write $M = \mat{a & b \\ c & d}$, and compute
	\begin{multline}\label{eq:sys_vanish}
		0 = 4\pa{R_{\f{2\pi}3} M R_{-\f{2\pi}3} - y M} \\
		= \mat{a\pa{1-4y} + 3d + \sqrt{3} \pa{c+b} & b\pa{1-4y} - 3c + \sqrt{3} \pa{d-a} \\
		c\pa{1-4y} - 3b + \sqrt{3} (d-a) & d\pa{1-4y} + 3a - \sqrt{3} (b+c)}.
	\end{multline}
When $y = 1$, the system is equivalent to
	\begin{align*}
\left\{
\begin{array}{l}
	c + b = \sqrt{3} (a-d)  \\
	\sqrt{3} (c+b) = a - d
\end{array}
\right.	
	\end{align*}
itself equivalent to $d = a$ and $c = -b$. When $y = \omega = -\f 12 + i\f{\sqrt{3}}2$, \eqref{eq:sys_vanish} is equivalent to
	\begin{align*}
	\left\{
	\begin{array}{l}
		\sqrt{3} (a+d) + c + b - 2ia = 0 \\
		\sqrt{3} (b-c) + d - a - 2ib = 0 \\
		\sqrt{3} (c-b) + d - a - 2ic = 0 \\
		\sqrt{3} (d+a) - b - c - 2id = 0.
\end{array}
\right.
\end{align*}
Adding the first and last equations gives $a = -d$, subtracting the second one and the third one gives $b = c$, and finally the first equation gives $b = ia$.

When $y = \overline{\omega}$, we replace $i$ by $-i$ in the previous system, we obtain $a = -d$, $b = c$ and $b = -ia$.
\end{proof}

The last lemma enables to obtain the form of the matrix elements involving two derivatives, as shown by the next result.
\begin{lemma}[Two derivatives]\label{lem:D}
Take two bounded operators $Y_1$ and $Y_2$ of $L^2\per(\Omega)$. Assume that, on $L_K^2(\Omega)$, $S_1$ and $S_2$ commute with $\cR_{\f{2\pi}{3} }$ and with $\cC \cP$.
\begin{itemize}
	\item There exist $t, s ,r \in \C$ such that for any $n,j \in \{1,2\}$,
\begin{align}\label{eq:double_der_u}
\ps{Y_1 (-i\partial_{K,n}) w_1, Y_2 (-i\partial_{K,j}) w_2} &= r D_{nj} \nonumber\\
\ps{Y_1 (-i\partial_{K,n}) w_2, Y_2 (-i\partial_{K,j}) w_1} &= \overline{r} D^*_{nj} \nonumber\\
\ps{Y_1 (-i\partial_{K,n}) w_1, Y_2 (-i\partial_{K,j}) w_1} &=t \delta_{nj} + s \pa{\sigma_2}_{nj} \\
\ps{Y_1 (-i\partial_{K,n}) w_2, Y_2 (-i\partial_{K,j}) w_2} &= \overline{t} \delta_{nj} - \overline{s} \pa{\sigma_2}_{nj}.\nonumber
\end{align}
	\item Take $\iota_j \in \{-1,1\}$ for any $j \in \{1,2\}$ and assume that $S_j \mathfrak{R} = \iota_j \mathfrak{R} S_j$. If $\iota_1\iota_2 = 1$, then $t,s,r\in\R$ and if $\iota_1\iota_2 = -1$, then $t,s,r\in i\R$.
\end{itemize}
\end{lemma}
\begin{proof}
From~\eqref{eq:perm_partial_exp} we have
\begin{align*}
\ps{Y_1 (-i\partial_{K,n}) w_a, Y_2 (-i\partial_{K,j}) w_b} = \ps{S_1 \partial_n \phi_a, S_2 \partial_j \phi_b}.
\end{align*}
We define the $2 \times 2$ matrices $M^{(ab)}$ by
	\begin{align*}
		M^{(ab)}_{ij} := \ps{S_1\partial_i \phi_a, S_2 \partial_j \phi_b}.
	\end{align*}
We recall that for any $A \in \cM_{2 \times 2}(\C)$, $C_1 ,C_2 \in \cM_{2 \times 1}(\C)$ and $L_1, L_2 \in \cM_{1 \times 2}(\C)$, we have $A \mat{C_1 & C_2} = \mat{A C_1 & A C_2}$ and $\mat{L_1 \\ L_2} A = \mat{L_1 A \\ L_2 A}$. Moreover, we consider $\na = \mat{\partial_1 \\ \partial_2}$ as a column vector and $\na^T = \mat{\partial_1 & \partial_2}$. Hence
	\begin{align*}
		R_{\f{2\pi}3} M^{(ab)} &= R_{\f{2\pi}3} \mat{\ps{ S_1 \na \phi_a, S_2 \partial_1 \phi_b} & \ps{ S_1 \na \phi_a, S_2 \partial_2 \phi_b}} \\
				       &=\mat{\ps{R_{\f{2\pi}3} S_1 \na \phi_a, S_2 \partial_1 \phi_b} & \ps{R_{\f{2\pi}3} S_1 \na \phi_a, S_2 \partial_2 \phi_b}} \\
				       &=  \ps{R_{\f{2\pi}3} S_1 \na \phi_a, S_2 \na^T \phi_b}
				       \end{align*}
and
				       \begin{align*}
		M^{(ab)} R_{-\f{2\pi}3} &= \mat{\ps{S_1 \partial_1 \phi_a, S_2 \na \phi_b}^T \\ \ps{S_1 \partial_2 \phi_a, S_2 \na \phi_b}^T} R_{-\f{2\pi}3}= \mat{\ps{S_1 \partial_1 \phi_a, S_2 \na \phi_b}^T R_{\f{2\pi}3}^T \\ \ps{S_1 \partial_2 \phi_a, S_2 \na \phi_b}^T R_{\f{2\pi}3}^T} \\
					&= \mat{\ps{S_1 \partial_1 \phi_a,R_{\f{2\pi}3} S_2 \na \phi_b}^T \\ \ps{S_1 \partial_2 \phi_a,R_{\f{2\pi}3} S_2 \na \phi_b}^T}= \ps{ S_1 \na \phi_a, S_2 \pa{R_{\f{2\pi}3}\na}^T \phi_b}.
	\end{align*}
We deduce that
	\begin{align*}
		\pa{R_{\f{2\pi}3} M^{(ab)} R_{-\f{2\pi}3}}_{ij} &= \pa{\ps{ S_1 R_{\f{2\pi}3} \na \phi_a, S_2 \pa{R_{\f{2\pi}3}\na}^T \phi_b}}_{ij}\\
								&= \ps{ \pa{S_1 R_{\f{2\pi}3}\na \phi_a}_i, \pa{S_2 R_{\f{2\pi}3} \na \phi_b}_j} \\
		&= \ps{ \cR_{\f{2\pi}3} \pa{S_1 R_{\f{2\pi}3}\na \phi_a}_i, \cR_{\f{2\pi}3}\pa{S_2  R_{\f{2\pi}3} \na \phi_b}_j} \\
		&\underset{\substack{\eqref{eq:transform_cR}}}{=} \; \ps{ \pa{S_1\na \cR_{\f{2\pi}3}\phi_a}_i, \pa{ S_2\na\cR_{\f{2\pi}3} \phi_b}_j} \\
		&\underset{\substack{\eqref{eq:syms_phi}}}{=} \; \omega^{b-a}\ps{ \pa{S_1\na \phi_a}_i, \pa{ S_2\na \phi_b}_j} \\
		&= \omega^{b-a} \ps{S_1 \partial_i \phi_a, S_2  \partial_j \phi_b} =\omega^{b-a} M^{(ab)}_{ij},
	\end{align*}
so $R_{\f{2\pi}3} M^{(ab)} R_{-\f{2\pi}3} = \omega^{b-a} M^{(ab)}$. Using Lemma \ref{lem:rotM}, we obtain that there exist $r,r',t,s,t',s' \in \C$ such that
\begin{align*}
\ps{S_1 \partial_n \phi_1, S_2 \partial_j \phi_2} &= r D_{nj} \\
\ps{S_1 \partial_n \phi_2, S_2 \partial_j \phi_1} &= r' D_{nj}^* \\
\ps{S_1 \partial_n \phi_1, S_2 \partial_j \phi_1} &=t \delta_{nj} + s \pa{\sigma_2}_{nj} \\
\ps{S_1 \partial_n \phi_2, S_2 \partial_j \phi_2} &=t' \delta_{nj} + s' \pa{\sigma_2}_{nj}.
\end{align*}

Then
\begin{align}\label{eq:pc_action} 
	M^{(ab)}_{nm} &= \ps{S_1 \partial_n \phi_a, S_2  \partial_m \phi_b} = \overline{\ps{\cP \cC S_1 \partial_n \phi_a, \cP \cC S_2 \partial_m \phi_b}} \nonumber \\
		    &= \overline{\ps{ S_1 \partial_n \cP \cC\phi_a,  S_2 \partial_m \cP \cC\phi_b}} = \overline{\ps{S_1 \partial_n \phi_{\underline{a}}, S_2 \partial_m \phi_{\underline{b}}}} = \overline{M^{(\underline{a} \underline{b})}_{nm}},
\end{align}
hence
\begin{align}\label{eq:set1} 
	t &= M^{(11)}_{11} = \overline{M^{(22)}_{11}} = \overline{t'} \nonumber\\
	i s &= M^{(11)}_{21} = \overline{M^{(22)}_{21}} = \overline{i s'} = -i \; \overline{s'} \\
r &= M^{(12)}_{11} = \overline{M^{(21)}_{11}} = \overline{r'}.\nonumber
\end{align}
This concludes the proof of~\eqref{eq:double_der_u}.

Set $\rho:=\iota_1\iota_2$. Then
\begin{align}\label{eq:action_R} 
M^{(ab)}_{nm}  &=\ps{S_1 \partial_n \phi_a, S_2 \partial_m \phi_b}  =  \ps{\mathfrak{R} S_1 \partial_n \phi_a, \mathfrak{R} S_2 \partial_m \phi_b}   \nonumber \\
		       &= \rho \ps{ S_1 \mathfrak{R}\partial_n \phi_a,  S_2 \mathfrak{R}\partial_m \phi_b} \underset{\substack{\eqref{eq:com_R_der}}}{=} \;  \rho(-1)^{n+m} \ps{ S_1 \partial_n \mathfrak{R}\phi_a,  S_2 \partial_m \mathfrak{R}\phi_b} \nonumber \\
		       &\underset{\substack{\eqref{eq:transf_phi_R}}}{=} \; \rho (-1)^{n+m} \eta_{\mathfrak R}^2 \ps{ S_1 \partial_n \phi_{\underline{a}},  S_2 \partial_m \phi_{\underline{b}}}=  \rho(-1)^{n+m}M^{(\underline{a} \underline{b})}_{nm}. 
\end{align}
Finally,
\begin{align*}
		t &= M^{(11)}_{11} = \rho M^{(22)}_{11} = \rho t' = \rho\overline{t}, \nonumber\\
		i s &= M^{(11)}_{21} = -\rho M^{(22)}_{21} = -\rho i s' = \rho i\overline{s}, \\
r &= M^{(12)}_{11} = \rho M^{(21)}_{11} =\rho r'=\rho\overline{r}. \nonumber
\end{align*}
\end{proof}

As a consequence of Lemma~\ref{lem:D}, we can compute
\begin{align}\label{eq:double_der_u_generic}
\ps{Y_1 \der_k w_1, Y_2 \der_k w_1} =t, \; \ps{Y_1 \der_k w_2, Y_2 \der_k w_2} =\overline{t}, \; \ps{Y_1 \der_k w_1, Y_2 \der_k w_2} = r e^{i 2 \theta_k}.
\end{align}
Indeed, using that $k_1 = \ab{k} \cos \theta_k$, $k_2 = \ab{k} \sin \theta_k$, we have
\begin{align*}
	\ps{Y_1 \der_k w_1, Y_2 \der_k w_2} = \f{r}{\ab{k}^2 }  \pa{k_1^2 - k_2^2 + 2 i k_1 k_2} = \f{r}{\ab{k}^2 }  \pa{k_1 + i k_2}^2 = r e^{i 2 \theta_k}.
\end{align*}

\subsection{Three derivatives}%
\label{sub:Three derivatives}

To finish, we compute the elements involving three derivatives.

Let us define
\begin{align*}
F^{ab}_{rjm} := \ps{Y_1 (-i\partial_{K,r}) w_a, Y_2 (-i\partial_{K,j}) Y_3 (-i\partial_{K,m}) w_b},
\end{align*}
where $a,b\in\{1,2\}$ label the two Bloch states and $r,j,m\in\{1,2\}$ label Cartesian coordinates. Define
\begin{align*}
\chi_1 := F^{11}_{111}, \quad \;\;    \chi_2 := -i F^{11}_{222}, \quad \;\; \gamma_1 := F^{12}_{122}, \quad \;\;   \gamma_2 := F^{12}_{212}, \quad \;\;   \gamma_3 :=  F^{12}_{221}.
\end{align*}
For any proposition $\bbP$, we define $\delta_{\bbP} := 1$ if $\bbP$ is true and $\delta_{\bbP} := 0$ otherwise.

\begin{lemma}[Three derivatives]\label{lem:three_derivatives}
Assume that $S_2$ is bounded on $L_K^2(\Omega)$ and that $S_1,S_3$ are bounded from $L_K^2(\Omega)$ to $H_K^1(\Omega)$. Assume moreover that each $S_j$, $j\in\{1,2,3\}$, commutes with $\cR_{\f{2\pi}{3} }$, $\cC \cP$ and $\mathfrak{R}$ on $L_K^2(\Omega)$, with the commutation relations for $S_1$ and $S_3$ also holding as identities from $L_K^2(\Omega)$ to $H_K^1(\Omega)$. Then
\begin{align}\label{eq:formulas_F_aa} 
\begin{array}{l}
F^{11}_{122} \hspace{-0.05cm} = \hspace{-0.05cm}  F^{11}_{212} \hspace{-0.05cm} = \hspace{-0.05cm}  F^{11}_{221} \hspace{-0.05cm} = \hspace{-0.05cm}  F^{22}_{122} \hspace{-0.05cm} = \hspace{-0.05cm}  F^{22}_{212} \hspace{-0.05cm} = \hspace{-0.05cm}  F^{22}_{221} \hspace{-0.05cm} = \hspace{-0.05cm}  - F^{11}_{111} \hspace{-0.05cm} = \hspace{-0.05cm}  -F^{22}_{111} \hspace{-0.05cm} = \hspace{-0.05cm}  - \chi_1 \in \R \\
F^{11}_{112} \hspace{-0.05cm} = \hspace{-0.05cm}  F^{11}_{211} \hspace{-0.05cm} = \hspace{-0.05cm}  F^{11}_{121} \hspace{-0.05cm} = \hspace{-0.05cm}  -F^{22}_{112} \hspace{-0.05cm} = \hspace{-0.05cm}  - F^{22}_{211} \hspace{-0.05cm} = \hspace{-0.05cm}  - F^{22}_{121} \hspace{-0.05cm} = \hspace{-0.05cm}  -  F^{11}_{222}\hspace{-0.05cm} = \hspace{-0.05cm}  F^{22}_{222} \hspace{-0.05cm} = \hspace{-0.05cm}  -i\chi_2 \in i\R,
\end{array}
\end{align}
and
\begin{align}\label{eq:formulas_F} 
\begin{array}{l}
F^{12}_{122} = F^{21}_{122} = i F^{12}_{211} = -i F^{21}_{211} = \gamma_1 \in \R  \\
F^{12}_{212} = F^{21}_{212}= iF^{12}_{121} = -iF^{21}_{121} =  \gamma_2 \in \R \\
F^{12}_{221} = F^{21}_{221} = i F^{12}_{112} = -i F^{21}_{112} = \gamma_3 \in \R  \\
F^{12}_{111} = F^{21}_{111} = i F^{12}_{222} = -i F^{21}_{222} = \gamma_1 + \gamma_2 + \gamma_3 \in \R. 
\end{array}
\end{align}
If moreover $Y_1 = Y_3$ and $Y_2 = 1$, then $\chi_2 = 0$ and $\gamma_1 = \gamma_3$.
\end{lemma}
\begin{proof}
Using $\cP \cC$ and $\mathfrak{R}$, the same arguments as in~\eqref{eq:pc_action} and~\eqref{eq:action_R} show that for any $a,b,r,j,m \in \{1,2\}$,
\begin{align*}
F^{ab}_{rjm} = \overline{F^{\underline{a} \underline{b}}_{rjm}}, \qquad \qquad F^{ab}_{rjm} = (-1)^{r+j+m+1} F^{\underline{a} \underline{b}}_{rjm},
\end{align*}
where we used the previously defined indicator notation and the identity
\begin{align*}
(-1)^{\delta_{r=2}+\delta_{j=2}+\delta_{m=2}}=(-1)^{r+j+m+1},
\end{align*}
which follows from the commutation of $\mathfrak{R}$ with $\partial_s$ for $s\in\{r,j,m\}$. Using both equations we deduce that $F^{ab}_{rjm} = (-1)^{r+j+m+1} \overline{F^{ab}_{rjm}}$ hence
\begin{align}\label{eq:fab} 
\begin{array}{ll}
F^{ab}_{rjm} = F^{\underline{a} \underline{b}}_{rjm} \in \R & \mbox{if } r+j+m \in 2\N +1, \\
F^{ab}_{rjm} = -F^{\underline{a} \underline{b}}_{rjm} \in i\R & \mbox{if }r+j+m \in 2\N.
\end{array}
\end{align}

For $r \in \{1,2\}$, let us denote by $P_r : \C^2 \rightarrow \C$ the projection on the $r^{\text{th}}$ coordinate, and $Q := 2R_{-\f{2\pi}{3}} = \mat{-1 & \sqrt{3} \\ -\sqrt{3} & -1}$. As previously we set $S_j := e^{iKx} Y_j e^{-iKx}$ for $j \in \{1,2,3\}$. We have
\begin{align}\label{eq:relation_rot_M} 
	&8\omega^{a-b} F^{ab}_{rjm} \nonumber \\
	&=8 \omega^{a-b}\ps{S_1 (-i\partial_r) \phi_a, S_2 (-i\partial_j) S_3 (-i\partial_m) \phi_b} \nonumber \\
	&= 8 \omega^{a-b} \ps{S_1 (P_r (-i\na)) \phi_a, S_2 (P_j(-i\na)) S_3 (P_m(-i\na)) \phi_b} \nonumber \\
	&=  8 \omega^{a-b}\ps{\cR_{\f{2\pi}{3}} S_1 (P_r (-i\na)) \phi_a, \cR_{\f{2\pi}{3}}S_2 (P_j(-i\na)) S_3 (P_m(-i\na)) \phi_b} \nonumber \\
	&\underset{\substack{\eqref{eq:transform_cR}}}{=} \;  8 \omega^{a-b} \times \nonumber \\
	&\hspace{-0.5cm} \ps{ S_1 (P_r R_{-\f{2\pi}{3}} (-i\na)) \cR_{\f{2\pi}{3}} \phi_a, S_2 (P_j R_{-\f{2\pi}{3}}(-i\na))  S_3 (P_m R_{-\f{2\pi}{3}}(-i\na)) \cR_{\f{2\pi}{3}} \phi_b} \nonumber \\
	&\underset{\substack{\eqref{eq:syms_phi}}}{=} \;  \sum_{\substack{1 \le r',j',m' \le 2}} \ps{ S_1 Q_{rr'} (-i\partial_{r'})  \phi_a, S_2 Q_{jj'}(-i\partial_{j'})  S_3 Q_{mm'}(-i\partial_{m'}) \phi_b} \nonumber \\
	&=   \sum_{\substack{1 \le r',j',m' \le 2}} Q_{rr'} Q_{jj'} Q_{mm'}        F^{ab}_{r'j'm'}
\end{align}
We then define the coordinates
\begin{align*}
	g_1 &:= (1,1,1), \qquad g_2 := (1,1,2), \qquad g_3 := (1,2,1), \qquad g_4 := (1,2,2), \\
g_5 &:= (2,1,1), \qquad g_6 := (2,1,2), \qquad g_7 := (2,2,1), \qquad g_8 := (2,2,2),
\end{align*}
the vectors $F^{ab} \in \C^8$ by $F^{ab} := \pa{F^{ab}_{g_\mu}}_{1 \le \mu \le 8}$, and the $8 \times 8$ matrix $W$ by
\begin{align*}
W_{\alpha\beta} := Q_{(g_\alpha)_1 (g_\beta)_1} Q_{(g_\alpha)_2 (g_\beta)_2} Q_{(g_\alpha)_3 (g_\beta)_3},
\end{align*}
where $(g_\alpha)_s$ is the $s^{\text{th}}$ coordinate of $g_\alpha$, where $s \in \{1,2,3\}$ and $\alpha , \beta \in \{1,\dots,8\}$. We can compute
\begin{align*}
	\tiny{
W = 
\mat{
 -1 & \sqrt{3} &  \sqrt{3} & -3 & \sqrt{3} & -3 &-3 & 3\sqrt{3} \\
 -\sqrt{3} & -1 & 3 & \sqrt{3} &  3 & \sqrt{3} & -3\sqrt{3} & -3 \\
 -\sqrt{3} &  3 &-1 & \sqrt{3} &  3 &-3\sqrt{3} &  \sqrt{3} & -3 \\
 -3 &-\sqrt{3} & -\sqrt{3} & -1 & 3\sqrt{3} &  3 & 3 & \sqrt{3} \\
 -\sqrt{3} &  3 & 3 &-3\sqrt{3} & -1 & \sqrt{3} &  \sqrt{3} & -3 \\
 -3 &-\sqrt{3} &  3\sqrt{3} &  3 &-\sqrt{3} & -1 & 3 & \sqrt{3} \\
 -3 & 3\sqrt{3} & -\sqrt{3} &  3 &-\sqrt{3} &  3 &-1 & \sqrt{3} \\
 -3\sqrt{3} & -3 &-3 &-\sqrt{3} & -3 &-\sqrt{3} & -\sqrt{3} & -1 \\
}
},
\end{align*}
the relation~\eqref{eq:relation_rot_M} can be written, for any $a,b \in \{1,2\}$,
\begin{align}\label{eq:system_8} 
(W - 8 \omega^{a-b}) F^{ab} = 0.
\end{align}

Let us first determine $F^{11}$ and $F^{22}$. We define
\begin{align*}
	\tiny{
A := 
\mat{
   \sqrt{3} & -9 & -3 & - 3 \sqrt{3} &  3 & - \sqrt{3} &   \sqrt{3} &  3 \\
  3 &   \sqrt{3} &   3 \sqrt{3} & -9 & - \sqrt{3} &  3 & -3 & - \sqrt{3} \\
 -24 & 0 & 0 &  24 & 0 &  24 &  24 & 0 \\
 0 &  24 &  24 & 0 &  24 & 0 & 0 & -24 \\
 0 & -6 &  6 & 0 & 0 & - 2 \sqrt{3} &   2 \sqrt{3} & 0 \\
 -6 & 0 & 0 & -6 & - 2 \sqrt{3} & 0 & 0 & - 2 \sqrt{3} \\
   2 \sqrt{3} & 0 & 0 &   2 \sqrt{3} & -6 & 0 & 0 & -6 \\
 0 &   2 \sqrt{3} & - 2 \sqrt{3} & 0 & 0 & -6 &  6 & 0 \\
}
}
\end{align*}
respecting $\det A \neq 0$, and
\begin{align*}
	\tiny{
\widetilde{A} := \mat{
0 &  1 & 0 & 0 & -1 & 0 & 0 & 0 \\
0 & 0 & 0 &  1 & 0 & -1 & 0 & 0 \\
0 & 0 & 0 & 0 & 0 & 0 & 0 & 0 \\
0 & 0 & 0 & 0 & 0 & 0 & 0 & 0 \\
0 &  1 & -1 & 0 & 0 & 0 & 0 & 0 \\
 1 & 0 & 0 &  1 & 0 & 0 & 0 & 0 \\
0 & 0 & 0 & 0 &  1 & 0 & 0 &  1 \\
0 & 0 & 0 & 0 & 0 &  1 & -1 & 0 \\
}
}
\end{align*}
and we can check that $A (W - 8) = 96\widetilde{A}$. Thus the systems~\eqref{eq:system_8} with $a=b$ can be rewritten $\widetilde{A} F^{11} = 0$ and $\widetilde{A} F^{22} = 0$. The first system is equivalent to
\begin{align}\label{eq:eq_sys} 
F^{11}_{112} = F^{11}_{211} = F^{11}_{121} = -F^{11}_{222}, \quad \text{ and }  \quad F^{11}_{122} = F^{11}_{212} = F^{11}_{221} = -F^{11}_{111}.
\end{align}
Using~\eqref{eq:fab} we have $F^{11}_{222} \in i \R$ so $\chi_2 = -i F^{11}_{222}  \in \R$. Similarly, we have $\chi_1 \in \R$. Again using~\eqref{eq:fab}, we deduce~\eqref{eq:formulas_F_aa}.

We now determine $F^{12}$ and $F^{21}$. We define
\begin{align*}
&B := \\ 
&\tiny{
\mat{
  7 \sqrt{3} - 3i & -3 - 3 \sqrt{3}i & 0 & 0 & -3+ \sqrt{3}i &  \sqrt{3}+ 3i &  4 \sqrt{3} & 0 \\
 -3+ \sqrt{3}i & -3 \sqrt{3}+ 3i & 0 & 0 &  \sqrt{3}+ 3i &  3 - \sqrt{3}i &  4 \sqrt{3}i & 0 \\
 -\sqrt{3} - 3i & -3 - 3 \sqrt{3}i & 0 & 0 & -3+ \sqrt{3}i & -7 \sqrt{3}+ 3i &  4 \sqrt{3} & 0 \\
 -3+ \sqrt{3}i &  5 \sqrt{3}+ 3i & 0 & 0 & -7 \sqrt{3}+ 3i &  3 - \sqrt{3}i &  4 \sqrt{3}i & 0 \\
 -4 \sqrt{3} &  12 & 0 & 0 &  4 \sqrt{3}i & -4 \sqrt{3} & -4 \sqrt{3}+ 12i & 0 \\
 -8 \sqrt{3}i & -8 \sqrt{3} & 0 & 0 & 0 & 0 & -8 \sqrt{3}i & -8 \sqrt{3} \\
  24 \sqrt{3} & 0 & 0 &  24 \sqrt{3} & -24 \sqrt{3}i & 0 & 0 & -24 \sqrt{3}i \\
 0 &  24 \sqrt{3} & -24 \sqrt{3} & 0 & 0 & -24 \sqrt{3}i &  24 \sqrt{3}i & 0 \\
}
}
\end{align*}
respecting $\det B \neq 0$, and
\begin{align*}
	\tiny{
\widetilde{B} := \mat{
-i & 0 & 0 & 0 & 0 & 0 & 0 &  1 \\
0 &  i & 0 & 0 & 0 & 0 &  1 & 0 \\
0 & 0 & -1 & 0 & 0 &  i & 0 & 0 \\
0 & 0 & 0 &  1 &  i & 0 & 0 & 0 \\
0 & 0 & 0 & 0 &  1 &  i &  i & -1 \\
0 & 0 & 0 & 0 & 0 & 0 & 0 & 0 \\
0 & 0 & 0 & 0 & 0 & 0 & 0 & 0 \\
0 & 0 & 0 & 0 & 0 & 0 & 0 & 0 \\
}.
}
\end{align*}
We can check that $B (W - 8\omega) = 96\widetilde{B}$, so the system~\eqref{eq:system_8} with $a \neq b$ can be rewritten $\widetilde{B} F^{21} = 0$, and we deduce that
\begin{align*}
F^{21}_{222} &= i F^{21}_{111}, \qquad \qquad F^{21}_{221} = - i F^{21}_{112}, \qquad \qquad F^{21}_{212} = -i F^{21}_{121} \\
F^{21}_{211} &= i F^{21}_{122}  \qquad \qquad F^{21}_{211} + i F^{21}_{212} + i F^{21}_{221} - F^{21}_{222} = 0,
\end{align*}
and we can conclude by using~\eqref{eq:fab}.

If moreover, $Y_1 = Y_3$ and $Y_2 = 1$, then
\begin{align*}
	-i\chi_2 &= F^{11}_{121} = \ps{Y_1 (-i\partial_{K,1}) w_1, (-i\partial_{K,2}) Y_1 (-i\partial_{K,1}) w_1} \\
&= \ps{(-i\partial_{K,2}) Y_1 (-i\partial_{K,1}) w_1,  Y_1 (-i\partial_{K,1}) w_1} \\
&= \overline{\ps{Y_1 (-i\partial_{K,1}) w_1, (-i\partial_{K,2}) Y_1 (-i\partial_{K,1}) w_1}} = \overline{F^{11}_{121}} = i \chi_2,
\end{align*}
so $\chi_2 = 0$. Similarly,
\begin{align*}
\gamma_3 &= F^{12}_{221} = \ps{Y_1 (-i\partial_{K,2}) w_1, (-i\partial_{K,2}) Y_1 (-i\partial_{K,1}) w_2} \\
 &= \overline{\ps{ Y_1 (-i\partial_{K,1}) w_2,  (-i\partial_{K,2}) Y_1 (-i\partial_{K,2}) w_1}} = \overline{F^{21}_{122}} = \overline{\gamma_1} = \gamma_1.
\end{align*}
\end{proof}

\subsection{Computation of the matrices~\eqref{eq:def_matrices} in the case of $\cF_{1,k}$, proof of Proposition~\ref{prop:order1_kdep}}%
\label{sub:Matrices in the case of k}

Those computations involve one or two derivatives, so Lemma~\ref{lem:D} is going to provide the answers.

\subsubsection{Satisfying the assumptions of Lemma~\ref{lem:D} }%
\label{ssub:Satisfying the assumptions of Lemma}

We define $h := \f{1}{2} (-\Delta) + v$ as an operator of $L_K^2(\Omega)$. With the notation~\eqref{eq:def_phi}, we define the operator $\cP$ as the orthogonal projection of $L^2_K(\Omega)$ onto $\Span (\phi_1,\phi_2) \subset L^2_K(\Omega)$, and the pseudo-inverse
\begin{align*}
	Z := \left\{
	\begin{array}{ll}
	\pa{(E\fer - h)_{\mkern 1mu \vrule height 2ex\mkern2mu \cP_\perp L^2_K(\Omega) \rightarrow \cP_\perp L^2_K(\Omega)}}^{-1} & \mbox{on } \cP_\perp L^2_K(\Omega) \\
	0 & \mbox{on } \cP L^2_K(\Omega).
	\end{array}
	\right.
\end{align*}
We have formally $Z = e^{iKx} R e^{-iKx}$, and $Z$ has the role of ``$S_a$'' in Lemma~\ref{lem:D}. We have that
\begin{align*}
\sigma_3 (ma) = m_1 \sigma_3 a_1 + m_2 \sigma_3 a_2 = -m_1 a_2 - m_2 a_1 = - (\sigma_1 m) a
\end{align*}
so
\begin{align*}
K \cdot ( -(\sigma_1 m) a) = K \cdot ( \sigma_3 (ma)) = (\sigma_3 K) \cdot ma \underset{\substack{\sigma_3 K = K}}{=} \; K \cdot ma.
\end{align*}
If $\phi \in L^2_K(\Omega)$, then
\begin{align*}
	(\mathfrak{R} \phi)(x + ma) &= \phi(\sigma_3(x+ma)) = \phi(\sigma_3 x - (-(\sigma_1 m) a)) \\
&= e^{i K \cdot (-(\sigma_1 m) a)} \phi(\sigma_3 x) = e^{i K \cdot ma} \pa{\mathfrak{R} \phi} (x),
\end{align*}
and $\mathfrak{R} \phi \in L^2_K(\Omega)$. Hence $\mathfrak{R}$ sends $L^2_K(\Omega)$ into itself. Moreover, since $\mathfrak{R} \phi_1 = \eta_{\mathfrak R} \phi_2$, then $\mathfrak{R}$ sends $\cP L^2_K(\Omega)$ into itself and $\cP_\perp L^2_K(\Omega)$ as well. We have $\mathfrak{R} \Delta = \Delta \mathfrak{R}$ and by the symmetries~\eqref{eq:honeycomb}, $\mathfrak{R} v = v \mathfrak{R}$ so $\mathfrak{R}$ commutes with $h$. We can conclude that
\begin{align*}
	\mathfrak{R} Z = Z \mathfrak{R},
\end{align*}
and we are going to be able to satisfy the assumption ``$S_a \mathfrak{R} = \mathfrak{R} S_a$'' of Lemma~\ref{lem:D}.

Moreover, $R_{\f{2\pi}{3} }K -K$ belongs to the reciprocal lattice of $\bbL$ so for any $m \in \Z^2$, $K \cdot R_{-\f{2\pi}{3} } (ma) = R_{\f{2\pi}{3} } K \cdot  (ma) = K \cdot ma$ modulo $2\pi$, and we can see that $\cR_{\f{2\pi}{3} }$ sends $L^2_K(\Omega)$ onto itself and commutes with $Z$. Similarly $\cP \cC$ sends $L^2_K(\Omega)$ onto itself and commutes with $Z$.

We proved that all the assumptions of Lemma~\ref{lem:D} are satisfied, with $\iota_1 \iota_2 = 1$, so we can apply it.

\subsubsection{Matrix $S$}%
\label{ssub:Matrix S k}

We have $S = \pa{\ps{R \der_k w_a, R \der_k w_b}}_{1 \le a,b \le 2}$. Applying Lemma~\ref{lem:D} and~\eqref{eq:double_der_u_generic} yields $S$ as in~\eqref{eq:expr_SM_k}.

\subsubsection{Matrix $M$}%
\label{ssub:Matrix M k}

We have
\begin{align*}
M_{ab}=\ps{R \der_k w_a, \pa{h_K - E\fer} R \der_k w_b}  = - \ps{\der_k w_a, R \der_k w_b},
\end{align*}
and we proceed as for $S$.

\subsubsection{Matrix $T$}%
\label{ssub:Matrix T k}

Then
\begin{align*}
T_{ab} = \ps{w_a, (-i\na_K) R \der_k w_b} = \ab{k}^{-1} \sum_{1\le \alpha \le 2} k_\alpha \ps{w_a, (-i\na_K) R (-i\partial_{K,\alpha}) w_b}.
\end{align*}
Using Lemma~\ref{lem:D},
\begin{align*}
T_{11} &= \ab{k}^{-1}\pa{k_1  \ps{w_1, (-i\na_K) R (-i\partial_{K,1}) w_1} + k_2  \ps{w_1, (-i\na_K) R (-i\partial_{K,2}) w_1}} \\
       &= \ab{k}^{-1}\pa{k_1  \mat{\tur \\ i\sur} + k_2 \mat{-i\sur \\ \tur}} = \ab{k}^{-1} \mat{k_1 \tur - i k_2 \sur\\ ik_1 \sur + k_2 \tur} \\
       &= \pa{\tur \1 + \sur \sigma_2} \tfrac{k}{\ab{k} } .
\end{align*}
Similarly,
\begin{align*}
T_{12} = \rur e^{i\theta_k} \mat{1 \\ i}, \qquad T_{21} = \rur e^{-i\theta_k} \mat{1 \\ -i}, \qquad T_{22} &= \pa{\tur \1 - \sur \sigma_2} \tfrac{k}{\ab{k} } 
\end{align*}
so we have~\eqref{eq:expr_T_k}.

\subsubsection{Matrix $L$}%
\label{ssub:Matrix L k}

The matrix $L$ is the first one which involves three derivatives in the scalar products. We apply Lemma~\ref{lem:three_derivatives}, we are in the configuration where ``$Y_1 = Y_3$'' and ``$Y_2 = 1$'' so $\chi_2 = 0$ and $\gamma_1 = \gamma_3$.  We define $e_1 := \icol{1 \\ 0}$ and $e_2 := \icol{0 \\ 1}$. We have $F^{11}_{rjm} =  \chir (-1)^{1+ \delta_{r=j=m}} \delta_{r+j+m \in 2\N +1}$. We compute, for $a \in \{1,2\}$,
\begin{align*}
&\ps{\psi_{a+2}, (-i\na_K) \psi_{a+2}} = \ps{R \der_k w_a, (-i\na_K) R \der_k w_a} \\
&\qquad = \ab{k}^{-2} \sum_{1 \le r,j,m \le 2} e_j k_r k_m F^{aa}_{rjm} = \chir \ab{k}^{-2} \mat{k_1^2 - k_2^2 \\ -2 k_1 k_2} =
	\chir \mat{\cos(2 \theta_k) \\ -\sin (2\theta_k)}.
\end{align*}
and
\begin{align*}
&\ps{\psi_3, (-i\na_K) \psi_4} =\ps{R \der_k w_1, (-i\na_K) R \der_k w_2} \\
&\qquad = \ab{k}^{-2} \mat{ k_1^2 (2\gaur + \gadr) + k_2^2 \gadr - 2 i k_1 k_2 \gaur \\ -i k_1^2 \gadr - i k_2^2 (2\gaur + \gadr) + 2k_1 k_2 \gaur} \\
&\qquad =  \mat{\gadr + 2 \gaur e^{-i\theta_k} \cos \theta_k \\ -i\gadr + 2 \gaur e^{-i\theta_k} \sin \theta_k} = \overline{\ps{\psi_4, (-i\na_K) \psi_3}}
\end{align*}
and we can conclude~\eqref{eq:expr_L_k}.

\subsection{Case of excited states, proof of Proposition~\ref{prop:excited}}%
\label{sub:Case of excited states}

Defining $\phi_3(x) := e^{i K \cdot x} w_3(x)$, the assumptions of Proposition~\ref{prop:excited} give $\cR_{\f{2\pi}{3} } \phi_3 = \cP \cC \phi_3 = \phi_3$ and $\mathfrak{R}\phi_3=\eta_3\phi_3$, where $\eta_{\mathfrak R}\eta_3=-1$. Moreover,
\begin{align*}
&R_{\f{2\pi}{3}}	\ps{\phi_1, (-i\na) \phi_3} = R_{\f{2\pi}{3}}\ps{\cR_{\f{2\pi}{3}} \phi_1, \cR_{\f{2\pi}{3}} (-i\na) \phi_3} \underset{\substack{\eqref{eq:transform_cR}}}{=} \; \overline{\omega}\ps{\phi_1, (-i\na) \phi_3}.
\end{align*}
Thus using also~\eqref{eq:eigenvec_R}, $\ps{\phi_1, (-i\na) \phi_3} = i \vpp \icol{1 \\ i}$ for some $\vpp \in \C$. Similarly, $\ps{\phi_2, (-i\na) \phi_3} = i c \icol{1 \\ -i}$ for some $c \in \C$, and $\ps{\phi_3, (-i\na) \phi_3} = 0$. Then
\begin{align*}
	\ps{\phi_1, (-i\na) \phi_3} = \overline{\ps{\cP \cC \phi_1, \cP \cC (-i\na) \phi_3}} =  \overline{\ps{\phi_2,  (-i\na) \cP \cC \phi_3}} = \overline{\ps{\phi_2, (-i\na) \phi_3}}
\end{align*}
and we deduce that $\vpp = -\overline{c}$. Moreover,
\begin{align*}
	\ps{\phi_1, (-i\na) \phi_3} &= \ps{\fR \phi_1, \fR (-i\na) \phi_3} = \sigma_3 \eta_{\mathfrak R} \ps{\phi_2,  (-i\na) \fR \phi_3} \\
& = \eta_{\mathfrak R}\eta_3\sigma_3 \ps{\phi_2,  (-i\na) \phi_3} = -\sigma_3 \ps{\phi_2,  (-i\na) \phi_3}
\end{align*}
implying that $\vpp = -c$ and $\vpp \in \R$.

\section{Proof of Theorem~\ref{thm:main_derivative_enriched_validity_time_scale} }%
\label{sec:Proof of Theorem }

\subsection{Definition}%
\label{sub:Definition a}

Throughout this section, $A=0$ and $V=0$, so $h^\ep=h$.

The matrix $\cS$ is positive definite. By~\eqref{eq:derivative_family_first_vectors}, the orthonormality of $\w_1,\w_2$, and $(h_K-E\fer)\w_j=0$, we have, for any $z_1,z_2 \in \C$,
\begin{align}\label{eq:prepared_coordinate_identities}
J\iota(z_1,z_2)=z_1\w_1+z_2\w_2,
\qquad \iota^*\cS\iota=1_{\C^2},
\qquad \cM\iota=0.
\end{align}
Equip the coefficient space with the scalar product
\begin{align}\label{eq:gram_scalar_product}
    \ps{Z,Z'}_{\cS}
    :=\ps{JZ,JZ'}_{L^2(\Omega)}
    =Z^*\cS Z',
    \qquad |Z|_{\cS}:= \sqrt{\ps{Z,Z}_{\cS}}.
\end{align}
Adjoints involving this coefficient space will be denoted by $\dagger$. In particular,
\begin{align}\label{eq:weighted_reconstruction_adjoint}
    J^\dagger=\cS^{-1}J^*,
    \qquad J^\dagger J=1_{\C^{N_{\cF}}},
    \qquad \Pi:=JJ^\dagger=J\cS^{-1}J^*.
\end{align}
Thus $J$ is unitary from $(\C^{N_{\cF}},\ps{\cdot,\cdot}_{\cS})$ onto $\Span\cF$, and $\Pi$ is the $L^2$-orthogonal projection onto that space.

Set $\cA(k):=h_{K+k}-E\fer$ and define
\begin{align}\label{eq:ritz_symbol}
B(k) :=  J^\dagger\cA(k)J =\cS^{-1}\left( \cM+\cL\cdot k +\tfrac12|k|^2\cS\right), 
\end{align}
which is self-adjoint for \eqref{eq:gram_scalar_product}, i.e. ${B(k)^\dagger=B(k)}$, and $\bbH^\ep= \ep^{-1} \cS B(-i\ep\na)$.

For the proof, we will measure the local quality of the microscopic family by
\begin{align}\label{eq:family_approximation_profile}
    d_{\cF}(k) :=\mysup{u\in\ran P(k)\\ \nor{u}{L^2_{\rm per}(\Omega)}=1} \dist\pa{u , \Span \cF}_{H^2_{\mathrm{per}}(\Omega)}.
\end{align}

For $|k|$ small, let $G(k):=\bpa{\ps{v_a(k),v_b(k)}_{L^2\per(\Omega)}}_{1\le a,b\le 2}$ and define $u_j(k):=\sum_{a=1}^2v_a(k)\bpa{G(k)^{-1/2}}_{aj}$ for $j\in\{1,2\}$. Since $G(0)=1_{\C^2}$, the matrix $G(k)$ is positive definite for $|k|$ small, and $(u_1(k),u_2(k))$ is an $L^2\per(\Omega)$-orthonormal basis of $\ran P(k)$. Moreover, $G(k)^{-1/2}$ is real-analytic, so~\eqref{eq:derivative_enriched_family_assumption} and Leibniz's rule imply that $\partial_k^\beta u_j(0)\in\Span\cF$ for every $|\beta|\le\ell$. Any $u(k)\in\ran P(k)$ satisfying $\nor{u}{L^2\per(\Omega)}=1$ can therefore be written as $u(k)=c_1u_1(k)+c_2u_2(k)$, where $|c_1|^2+|c_2|^2=1$. Applying Taylor's formula at order $\ell$ in $H^2\per(\Omega)$ to each $u_j$ gives
\begin{align}\label{eq:derivative_family_approximation}
    d_{\cF}(k)\le C\ab{k}^{\ell+1}
\end{align}
for $\ab{k}$ small. Throughout, $C>0$ denotes a constant which is independent of $\ep$, of $k$ and of $t$.

For any orthonormal projection $\cP$, we will use the notation
\begin{align*}
\cP^\perp := 1 - \cP.
\end{align*}

\subsection{A fibrewise Rayleigh--Ritz estimate}

Choose $r_{\mathrm R}>0$ so that
\[
    \sigma(B(0))\cap\{z\in\C \;|\; 0<|z|\le 2r_{\mathrm R}\}
    =\varnothing.
\]
This is possible by \eqref{eq:no_spurious_zero_mode}, which says that $0$ is an eigenvalue of $B(0)$ of multiplicity exactly two.  Since $B$ is entire, after decreasing $k_0>0$ if necessary, the circle $\Gamma_{\mathrm R}:=\acs{z\in\C \;|\; \ab{z}=r_{\mathrm R}}$ lies in the resolvent set of $B(k)$ and encloses exactly the two Ritz eigenvalues issuing from zero, for any $\ab{k}<k_0$.  Define the corresponding Riesz projection by
\begin{align}\label{eq:low_ritz_projection}
    \cQ(k)
    :=\frac{1}{2\pi i}\oint_{\Gamma_{\mathrm R}}
      (z-B(k))^{-1}\,\d z,
    \qquad \ab{k}<k_0,
\end{align}
where the contour is positively oriented.  Thus $\ran\cQ(k)$ has dimension two and $\ran\cQ(0)=\ran\iota$.  We now verify that this projection is $\cS$-orthogonal.  Set
\begin{align*}
    B_\cS(k):=\cS^{1/2}B(k)\cS^{-1/2},
    \qquad
    \cQ_\cS(k):=\cS^{1/2}\cQ(k)\cS^{-1/2}.
\end{align*}
Since $B(k)^\dagger=\cS^{-1}B(k)^*\cS=B(k)$, one has $B_\cS(k)^*=B_\cS(k)$.  Moreover,
\begin{align*}
    \cQ_\cS(k)
    =\frac{1}{2\pi i}\oint_{\Gamma_{\mathrm R}}
      \big(z-B_\cS(k)\big)^{-1}\,\d z,
\end{align*}
and it respects $\cQ_\cS(k)^*=\cQ_\cS(k)$.  It follows that
\begin{align*}
    \cQ(k)^\dagger
    =\cS^{-1}\cQ(k)^*\cS
    =\cS^{-1/2}\cQ_\cS(k)^*\cS^{1/2}
    =\cQ(k).
\end{align*}
The remaining identities follow from the spectral calculus, so
\begin{align}\label{eq:ritz_projection_properties}
    \cQ(k)^\dagger=\cQ(k),
    \qquad \cQ(k)^2=\cQ(k),
    \qquad B(k)\cQ(k)=\cQ(k)B(k),
\end{align}
and $e^{-itB(k)}$ is unitary for $\ab{\cdot}_{\cS}$ and commutes with $\cQ(k)$. We define
    \begin{align*}
    \mathfrak E(k,t):=e^{-it\cA(k)}J-Je^{-itB(k)}.
    \end{align*}

\begin{lemma}[Propagation in one Bloch fibre]\label{lem:fibre_propagator}
There is $C>0$ such that, for any $\ab{k}<k_0$, any $t\in\R$, and any $Z\in \ran\iota=\Span (e_1,e_2)\subset\C^{N_{\cF}}$,
    \begin{align}\label{eq:fibre_propagator_bound}
        \nor{\mathfrak E(k,t)Z}{H^2\per(\Omega)}
        \le C\pa{
        \ab{k}+\ab{k}^{2\ell+2}\ab{t}}\ab{Z}_{\cS}.
    \end{align}
\end{lemma}

\begin{proof}
\emph{Step 1: the trial space almost contains the exact low subspace.} Recall from \eqref{eq:weighted_reconstruction_adjoint} that $\Pi=JJ^\dagger$ is the $L^2$-orthogonal projection onto $\Span \cF$, bounded on $H^2\per(\Omega)$.  For any $u\in\ran P(k)$ and $f\in\Span \cF$, we have $\Pi^\perp u=\Pi^\perp (u-f)$. Taking the infimum in $f$ and using \eqref{eq:family_approximation_profile} therefore yields
    \begin{align}\label{eq:orthogonal_projection_approximation}
        \nor{\Pi^\perp P(k)}{L^2\per
        \to H^2\per}
        \le C d_{\cF}(k).
    \end{align}
From now on we decrease $k_0$ so that $C d_{\cF}(k)\le\tfrac12$ for $\ab{k}<k_0$, which is possible by \eqref{eq:derivative_family_approximation}. Define $\cT(k):=\left.\Pi\right|_{\ran P(k)}$ as a map from $\ran P(k)$ into $\Span\cF$.

\emph{Step 2: the low Ritz subspace is close to the exact one.} Define
\begin{align*}
\cA_{\cF}(k)&:=\left.\Pi\cA(k)\Pi\right|_{\Span\cF}, \quad \cA_P(k)&:=\left.\cA(k)\right|_{\ran P(k)}, \quad \widetilde\cQ(k)&:=J\cQ(k)J^\dagger.
\end{align*}
By \eqref{eq:ritz_symbol}, $\cA_{\cF}(k)=JB(k)J^\dagger$ on $\Span \cF$, so $\cA_{\cF}(k)$ is unitarily equivalent to $B(k)$.  Together with \eqref{eq:ritz_projection_properties}, this shows that $\widetilde \cQ(k)$ is the $L^2$-orthogonal projection onto the low Ritz subspace $J\ran \cQ(k)$, and that it commutes with $\cA_{\cF}(k)$. Because $P(k)$ reduces $\cA(k)$, the restriction $\cA_P(k)$ maps $\ran P(k)$ into itself. Hence, for every $u\in\ran P(k)$, the definitions of $\cA_{\cF}(k)$ and $\cT(k)$ give
\begin{align}\label{eq:trial_exact_intertwining}
\bpa{\cA_{\cF}(k)\cT(k)-\cT(k)\cA_P(k)}u =\Pi\cA(k)\Pi u-\Pi\cA(k)u =-\Pi\cA(k)\Pi^\perp P(k)u,
\end{align}
which says that $\cT(k)$ intertwines $\cA_P(k)$ and $\cA_{\cF}(k)$ up to a defect, which is at most $C d_{\cF}(k)$ by~\eqref{eq:orthogonal_projection_approximation} and the uniform boundedness of $\cA(k):H^2\per(\Omega)\to L^2\per(\Omega)$.

After decreasing $k_0$, the contour $\Gamma_{\mathrm R}$ also encloses both eigenvalues of $\cA_P(k)$, which tend to zero as $k\to0$. Both resolvents below are uniformly bounded on this contour in the $L^2$ norms. Write $F(k):=\cA_{\cF}(k)\cT(k)-\cT(k)\cA_P(k)$ and, for $z\in\Gamma_{\mathrm R}$, set $R_{\cF}:=(z-\cA_{\cF}(k))^{-1}$ and $R_P:=(z-\cA_P(k))^{-1}$. The resolvent relations $\pa{z - \cA_{\cF}} R_{\cF} = \1_{\Span\cF}$ and $ \pa{z - \cA_P} R_P = \1_{\ran P(k)}$ give
    \begin{multline*}
R_{\cF}F(k)R_P =R_{\cF}\bpa{\cA_{\cF}(k)\cT(k)-\cT(k)\cA_P(k)}R_P\\
= \bpa{zR_{\cF}-\1_{\Span\cF}}\cT(k)R_P-R_{\cF}\cT(k)\bpa{zR_P-\1_{\ran P(k)}}  = R_{\cF}\cT(k)-\cT(k)R_P.
    \end{multline*}
Moreover, the Riesz formulas read $\widetilde\cQ(k)=\frac{1}{2\pi i}\oint_{\Gamma_{\mathrm R}}R_{\cF}\,\d z$, and $\1_{\ran P(k)}=\frac{1}{2\pi i}\oint_{\Gamma_{\mathrm R}}R_P\,\d z$. Integrating over $z$ yields
    \begin{align*}
        \widetilde\cQ(k)\cT(k)-\cT(k)
        =\frac{1}{2\pi i}\oint_{\Gamma_{\mathrm R}}R_{\cF}F(k)R_P\,\d z.
    \end{align*}
The first resolvent has range in the fixed space $\Span\cF$, where the $L^2$ and $H^2$ norms are equivalent.  Thus the right-hand side is $O(d_{\cF}(k))$ from $\ran P(k)$ with its $L^2$ norm to $H^2\per$. Since $\widetilde\cQ(k)\Pi=\widetilde\cQ(k)$, applying \eqref{eq:orthogonal_projection_approximation} yields
    \begin{align*}
        \nor{(1-\widetilde\cQ(k))P(k)}{L^2\per\to H^2\per}
        \le C d_{\cF}(k).
    \end{align*}
In particular, $\left.\widetilde\cQ(k)\right|_{\ran P(k)}$ is a bijection onto $J\ran\cQ(k)$ for small $k$, since both spaces have dimension two.  For $f=\widetilde\cQ(k)u$ with $u\in\ran P(k)$, one has $\nor{u}{L^2\per}\le2\nor{f}{L^2\per}$ and $P(k)^\perp f=P(k)^\perp (f-u)$.  The frame $u_1(k),u_2(k)$ is uniformly bounded in $H^2\per$, so $P(k)$ is uniformly bounded from $L^2\per$ to $H^2\per$.  Consequently,
    \begin{align}\label{eq:ritz_projection_estimate}
        \nor{P(k)^\perp \widetilde\cQ(k)}{L^2\per\to H^2\per}
        \le C d_{\cF}(k).
    \end{align}

\emph{Step 3: Rayleigh--Ritz stationarity.}  Write $\cV(k):=\left.J\right|_{\ran \cQ(k)}$ and $\cD(k):=P(k)^\perp \cV(k)$. The map $\cV(k)$ is unitary from $\ran \cQ(k)$, equipped with the $\cS$ scalar product, onto the low Ritz subspace $J\ran\cQ(k)$.
For $Z\in\ran\cQ(k)$, the identities $J^\dagger J=1_{\C^{N_{\cF}}}$ and $\cQ(k)Z=Z$ imply $\widetilde\cQ(k)JZ =JZ$. Since $\cV(k)Z=JZ$, ~\eqref{eq:ritz_projection_estimate} and $\nor{JZ}{L^2\per}=\ab{Z}_{\cS}$ give
    \begin{align}\label{eq:normal_component_bound}
\nor{\cD(k)Z}{H^2\per} =\nor{P(k)^\perp JZ}{H^2\per} =\nor{P(k)^\perp\widetilde\cQ(k)JZ}{H^2\per} \le C d_{\cF}(k)\ab{Z}_{\cS},
    \end{align}
Using $\cV(k)Z=JZ$, for all $Z,Z'\in\ran\cQ(k)$,
    \begin{align*}
        \ps{Z,\cV(k)^\dagger\cV(k)Z'}_{\cS}
        =\ps{\cV(k)Z,\cV(k)Z'}_{L^2\per}
        =\ps{Z,Z'}_{\cS},
    \end{align*}
 that is, $\cV(k)^\dagger\cV(k)=\1_{\ran\cQ(k)}$. Consequently,
    \begin{align}\label{eq:polar_square}
        \cV(k)^\dagger P(k)\cV(k)=\1_{\ran\cQ(k)}-\cD(k)^\dagger \cD(k),
    \end{align}
which is positive definite for small $k$.  On $\ran\cQ(k)$, set
    \begin{align*}
        \Theta(k):=\bpa{1-\cD(k)^\dagger\cD(k)}^{1/2}.
    \end{align*}
The polar isometry of $P(k)\cV(k)$, extended by zero on $\ran(1-\cQ(k))$, is
    \begin{align}\label{eq:exact_ritz_polar_isometry}
        \cW(k):=P(k)\cV(k)\Theta(k)^{-1}\cQ(k).
    \end{align}
It is unitary from $(\ran\cQ(k),\ab{\cdot}_{\cS})$ onto $\ran P(k)$ with its $L^2$ scalar product, and $\cW(k)\cW(k)^\dagger=P(k)$. In particular, it is uniformly bounded from $(\ran\cQ(k),\ab{\cdot}_{\cS})$ to $H^2\per$. Since $\Theta(k)^2=1_{\ran\cQ(k)}-\cD(k)^\dagger\cD(k)$, we have $(1-\Theta(k))(1+\Theta(k))=\cD(k)^\dagger\cD(k)$. The operator $\Theta(k)$ is positive for the $\cS$ scalar product, and therefore $\nor{(1+\Theta(k))^{-1}}{\mathcal L_{\cS}(\ran\cQ(k))}\le 1$. Estimate~\eqref{eq:normal_component_bound}, together with the continuous embedding $H^2\per(\Omega)\subset L^2\per(\Omega)$, thus gives 
\begin{align}\label{eq:b1mt} 
\nor{1-\Theta(k)}{\mathcal L_{\cS}(\ran\cQ(k))}\le C d_{\cF}(k)^2.
\end{align}
Moreover, the identity $\cV(k)=\cW(k)\Theta(k)+\cD(k)$ implies that, for every $Z\in\ran\cQ(k)$,
\begin{align}\label{eq:polar_factor_bound}
        \nor{\bpa{\cW(k)-\cV(k)}Z}{H^2\per}
        &\le\nor{\cW(k)\bpa{1-\Theta(k)}Z}{H^2\per}+\nor{\cD(k)Z}{H^2\per} \nonumber\\
        &\le C\bpa{d_{\cF}(k)^2+d_{\cF}(k)}\ab{Z}_{\cS} \le C d_{\cF}(k)\ab{Z}_{\cS},
\end{align}
Define the exact and Ritz low Hamiltonians on $\ran\cQ(k)$ by
    \begin{align*}
        \fh(k):=\cW(k)^\dagger \cA(k)\cW(k),
        \qquad
        b(k):=\cV(k)^\dagger \cA(k)\cV(k)
        =\left.B(k)\right|_{\ran \cQ(k)}.
    \end{align*}
Since $P(k)$ commutes with $\cA(k)$, the relations $P(k)\cD(k)=0$ and $P(k)\cW(k)=\cW(k)$ make the mixed terms vanish. Substituting $\cV(k)=\cW(k)\Theta(k)+\cD(k)$ therefore gives
    \begin{align}\label{eq:ritz_stationarity_identity}
        b(k)
        =\Theta(k)\fh(k)\Theta(k)
         +\cD(k)^\dagger \cA(k)\cD(k).
    \end{align}
The uniform bounds on $\cW(k)$ into $H^2\per$ and on $\cA(k):H^2\per\to L^2\per$ imply $\nor{\fh(k)}{\mathcal L_{\cS}(\ran\cQ(k))}\le C$. Also, $0\le\Theta(k)\le1$ for the $\cS$ scalar product. Equations~\eqref{eq:normal_component_bound}--\eqref{eq:ritz_stationarity_identity} then give
    \begin{align*}
        \nor{\cD(k)^\dagger \cA(k)\cD(k)}{\mathcal L_{\cS}(\ran\cQ(k))}
	+ \nor{b(k)-\fh(k)}{\mathcal L_{\cS}(\ran\cQ(k))}
\le C d_{\cF}(k)^2.
    \end{align*}
Indeed, each factor $\cD(k)$ contributes $d_{\cF}(k)$ to the first bound. For the second, expand $\Theta\fh\Theta-\fh=(\Theta-1)\fh\Theta+\fh(\Theta-1)$ and use~\eqref{eq:b1mt}.

\emph{Step 4: Duhamel in the two-dimensional space.} Both $\fh(k)$ and $b(k)$ are self-adjoint for the $\cS$ scalar product. Their unitary propagators satisfy
    \begin{align*}
        e^{-it\fh(k)}-e^{-itb(k)}
        =-i\int_0^t e^{-i(t-s)\fh(k)}
        \bpa{\fh(k)-b(k)}e^{-isb(k)}\,\d s.
    \end{align*}
Thus
    \begin{align*}
        \nor{e^{-it\fh(k)}-e^{-itb(k)}}{\mathcal L_{\cS}(\ran\cQ(k))}
        \le C d_{\cF}(k)^2\ab{t}.
    \end{align*}
Since $\cA(k)\cW(k)=\cW(k)\fh(k)$ on $\ran\cQ(k)$,
    \begin{align*}
        e^{-it\cA(k)}\cW(k)=\cW(k)e^{-it\fh(k)}\cQ(k).
    \end{align*}
The uniform $H^2$ bound on $\cW(k)$ therefore yields
    \begin{align}\label{eq:low_propagator_bound}
        \nor{e^{-it\cA(k)}\cW(k)
          -\cW(k)e^{-itB(k)}\cQ(k)}
          {(\C^{N_{\cF}},\ab{\cdot}_{\cS})\to H^2\per}
        \le C d_{\cF}(k)^2\ab{t},
    \end{align}
    where
\begin{align*}
        \nor{R}{(\ran\cQ(k),\ab{\cdot}_{\cS})\to H^2\per}
        :=\sup_{\substack{Z\in\ran\cQ(k)\\ Z\neq0}}
        \frac{\nor{RZ}{H^2\per}}{\ab{Z}_{\cS}}.
\end{align*}
Also, \eqref{eq:polar_factor_bound} and $J\cQ(k)=\cV(k)\cQ(k)$ give
    \begin{align}\label{eq:ritz_intertwining}
        \nor{\cW(k)-J\cQ(k)}
          {(\C^{N_{\cF}},\ab{\cdot}_{\cS})\to H^2\per}
        \le C d_{\cF}(k).
    \end{align}
The exact propagator is uniformly bounded on $H^2\per$ for $\ab{k}<k_0$ and $t\in\R$, because it preserves the graph norm of $h_{K+k}$.

\emph{Step 5: conclusion.}  At $k=0$ we have $\ran \cQ(0)=\ran \iota$, and $k\mapsto\cQ(k)$ is analytic because $B$ is.  Hence, for $Z\in \ran \iota$,
    \begin{align}\label{eq:high_leakage}
       \ab{(1-\cQ(k))Z}_{\cS}=\ab{\bpa{\cQ(0)-\cQ(k)}Z}_{\cS}
       \le C\ab{k}\,\ab{Z}_{\cS}.
    \end{align}
Split $Z=\cQ(k)Z+(1-\cQ(k))Z$ in $\mathfrak E(k,t)Z$. The part carried by $(1-\cQ(k))Z$ is bounded by $C\ab{k}\ab{Z}_{\cS}$, using \eqref{eq:high_leakage}, the uniform boundedness of $e^{-it\cA(k)}$ on $H^2\per(\Omega)$, the unitarity of $e^{-itB(k)}$ for $\ab{\cdot}_{\cS}$, and the boundedness of $J$ with values in $H^2\per(\Omega)$.  For the part carried by $\cQ(k)Z$, and using $\cW(k)=\cW(k)\cQ(k)$,
    \begin{align*}
    & \mathfrak E(k,t)\cQ(k)Z =e^{-it\cA(k)}\bpa{J\cQ(k)-\cW(k)}Z\\
    &\quad +\bpa{e^{-it\cA(k)}\cW(k)-\cW(k)e^{-itB(k)}\cQ(k)}Z +\bpa{\cW(k)-J\cQ(k)}e^{-itB(k)}\cQ(k)Z,
    \end{align*}
The first and third lines are bounded using \eqref{eq:ritz_intertwining}, whereas the second one is bounded using \eqref{eq:low_propagator_bound}. After reducing $k_0$ so that $k_0<1$, inserting \eqref{eq:derivative_family_approximation} and using $\ab{k}^{\ell+1}\le \ab{k}$ for $\ab{k}<1$ proves \eqref{eq:fibre_propagator_bound}.
\end{proof}

\subsection{Bloch integration and proof of the theorem}

Let $\cU$ denote the Bloch--Floquet transform in periodic gauge,
\begin{align}\label{eq:bloch_floquet_transform}
    (\cU f)(q,x)
    :=\ab{\cB}^{- \f 12}\sum_{\gamma\in\bbL}
      e^{-iq\cdot(x+\gamma)}f(x+\gamma),
    \qquad (q,x)\in\cB\times\Omega.
\end{align}
It is unitary from $L^2(\R^2)$ onto $L^2(\cB,L^2\per(\Omega))$ and $\cU h\cU^{-1}=\int_\cB^\oplus h_q\,\d q$. With the unitary Fourier transform on $\R^2$,
\[
    \widehat a(\xi):=(2\pi)^{-1}
      \int_{\R^2}e^{-i\xi\cdot X}a(X)\,\d X,
\]
Poisson summation gives
\begin{align}\label{eq:bloch_reconstruction_formula}
    (\cU\cG a)(q,x)
    =|\Omega|^{-1/2}\ep^{-1}\sum_{G\in\bbL^*}e^{iG\cdot x}
      J\widehat a\left(\tfrac{q-K+G}{\ep}\right).
\end{align}

\begin{proof}[Proof of Theorem~\ref{thm:main_derivative_enriched_validity_time_scale}]
We have $\ep\cS^{-1}\bbH^\ep=B(-i\ep\na_X)$. Fourier transformation of~\eqref{eq:effective_evolution}, with $T=\ep t$, yields $\widehat\alpha(\xi,\ep t)=e^{-itB(\ep\xi)}\widehat{\alpha^0}(\xi)$. Apply~\eqref{eq:bloch_reconstruction_formula} to the two terms of~\eqref{eq:def_remainder}. Then $h_qe^{iG\cdot x}=e^{iG\cdot x}h_{q+G}$ shows that the summand indexed by $G$ carries $\mathfrak E(k,t)$ with $k=q-K+G$, so that
    \begin{align}\label{eq:bloch_remainder_series}
    (\cU r^\ep)(q,\cdot,t)=|\Omega|^{-\f 12}\ep^{-1}\sum_{G\in\bbL^*}e^{iG\cdot x}\,
    \mathfrak E(q-K+G,t)\,\widehat{\alpha^0}\pa{\tfrac{q-K+G}{\ep}}.
    \end{align}
The Bloch transform identifies the $H^2(\R^2)$ norm with a uniformly equivalent fibre $H^2\per$ norm, because $q\in\cB$ is bounded. For $G\in\bbL^*$, set $F_G(q):=\mathfrak E(q-K+G,t)\widehat{\alpha^0}((q-K+G)/\ep)$. Multiplication by $e^{iG\cdot x}$ has $H^2\per$ operator norm at most $C(1+\ab{G})^2$, so~\eqref{eq:bloch_remainder_series} gives
    \begin{multline*}
        \nor{(\cU r^\ep)(q,\cdot,t)}{H^2\per}
        \le C\ep^{-1}\sum_{G\in\bbL^*}(1+\ab{G})^2\nor{F_G(q)}{H^2\per}\\
        \le C\ep^{-1}\bpa{\sum_{G\in\bbL^*}(1+\ab{G})^{-4}}^{1/2}
        \bpa{\sum_{G\in\bbL^*}(1+\ab{G})^8\nor{F_G(q)}{H^2\per}^2}^{1/2}.
    \end{multline*}
The first sum in the last line is finite. Squaring this estimate, integrating over $q\in\cB$, and using the norm equivalence furnished by the Bloch transform give an integral over $\cB$. Finally, making the change of variables $k=q-K+G$, for which $\d k=\d q$ and $q\in\cB$ corresponds to $k\in\cB-K+G$, separately in each summand yields
    \begin{align*}
        \nor{r^\ep(\cdot,t)}{H^2(\R^2)}^2
        \le C\ep^{-2}\sum_{G\in\bbL^*}\int_{\cB-K+G}(1+\ab{G})^8
        \nor{\mathfrak E(k,t)\widehat{\alpha^0}(k/\ep)}{H^2\per(\Omega)}^2\,\d k.
    \end{align*}
The sets $\cB-K+G$ tile $\R^2$ up to sets of vanishing measure, and $1+\ab{G}\le C(1+\ab{k})$ for $k\in\cB-K+G$. Unfolding the sum therefore gives
    \begin{align}
        \nor{r^\ep(\cdot,t)}{H^2(\R^2)}^2
	&\le C\ep^{-2}\int_{\R^2}(1+\ab{k})^8
        \nor{\mathfrak E(k,t)\widehat{\alpha^0}(k/\ep)}{H^2\per(\Omega)}^2\,\d k \nonumber \\
	&\le C\int_{\R^2}(1+\ab{\ep\xi})^8 \nor{\mathfrak E(\ep\xi,t)\widehat{\alpha^0}(\xi)}{H^2\per(\Omega)}^2\,\d\xi.
    \label{eq:integrated_fibre_error}
    \end{align}

To estimate all momenta, use the elliptic bound
    \begin{align*}
        \nor{u}{H^2\per}
        \le C\pa{\nor{h_{K+k}u}{L^2\per}
             +(1+\ab{k})^2\nor{u}{L^2\per}},
        \qquad u\in H^2\per(\Omega),
    \end{align*}
whose constant is independent of $k$. Applying it to $e^{-it\cA(k)}JZ$, and using conservation of the $L^2$ norm and of $\nor{h_{K+k}u}{L^2\per}$, gives a bound $C(1+\ab{k})^2\ab{Z}_{\cS}$. Indeed, $\nor{h_{K+k}JZ}{L^2\per}\le C(1+\ab{k})^2\ab{Z}_{\cS}$. The effective propagator is $\cS$-unitary and $J$ is bounded into $H^2\per$, so uniformly in $Z \in \C^{N_\cF}$, $t \in \R$ and $k \in \R^2$,
\begin{align*}
\nor{\mathfrak E(k,t) Z}{H^2\per} \le C(1+\ab{k})^2 \ab{Z}_\cS.
\end{align*}
For $\ab{k}\ge k_0$ this is bounded by $C(\ab{k}+\ab{k}^2)$. Combining this with Lemma~\ref{lem:fibre_propagator} gives the global estimate
    \begin{align}\label{eq:global_fibre_error}
        \nor{\mathfrak E(k,t)Z}{H^2\per}
        \le C\pa{\ab{k}+\ab{k}^2+\ab{k}^{2\ell+2}\ab{t}}\ab{Z}_{\cS},
        \qquad Z\in\ran\iota.
    \end{align}

Finally, $\widehat{\alpha^0}=\iota\widehat{\mathfrak a}$ and $|\widehat{\alpha^0}|_{\cS}=\ab{\widehat{\mathfrak a}}$ by \eqref{eq:prepared_coordinate_identities}. Insert~\eqref{eq:global_fibre_error} into~\eqref{eq:integrated_fibre_error}. Since $0<\ep\le1$, the weight is bounded by $(1+\ab{\xi})^8$, and all polynomial moments of $\widehat{\mathfrak a}$ are square-integrable. Taking the square root therefore gives, uniformly for $t\in\R$,
    \begin{align}\label{eq:constant_potential_pointwise_error}
        \nor{r^\ep(\cdot,t)}{H^2(\R^2)}
        \le C\pa{\ep+\ep^2+\ep^{2\ell+2}\ab{t}}
        \le C\pa{\ep+\ep^{2\ell+2}\ab{t}}.
    \end{align}
Taking the supremum over $0\le t\le t_\ep$ proves \eqref{eq:constant_potential_sharp_error}.
\end{proof}

\section{Formal derivation of the perturbation series for the macroscopic functions}%
\label{sec:Formal derivation of the perturbation series for the macroscopic functions}

In this section, we provide a formal justification of the choice of adding derivatives with respect to the momentum, in the families $\cF$ of Section~\ref{sub:Dirac modes and their derivatives}.

We recall that when $A=0$, from~\eqref{eq:def_h_ep_q} we have
\begin{align*}
h^\ep_q = \tfrac{1}{2} (-i\na_q)^2 + v(x) + \ep V(\ep x).
\end{align*}
We formally define $X = \ep x$ as the macroscopic spatial variable, and hence $\na = \na_x + \ep \na_X$, $\na_q = \na_{x,q} + \ep \na_X$ where $\na_{x,q} := \na_x +iq$, and
\begin{align*}
	h^\ep_q &= \tfrac{1}{2} \pa{(-i\na_{x,q})^2 + 2 \ep (-i\na_{x,q})(-i\na_X) + \ep^2 (-i\na_X)^2 } + v(x) + \ep V(X) \\
&= h_q +  \ep (-i\na_{x,q})(-i\na_X) + \tfrac 12 \ep^2 (-i\na_X)^2 + \ep V(X).
\end{align*}
We are interested in a branch $\ep \mapsto (E^\ep_q, \xi^\ep_q)$ of eigenmodes of $h^\ep_q$. Thus we expand the eigenmodes in $\ep$ as
\begin{align*}
\xi^\ep_q(x,X) = \xi^{(0)}_q(x,X) + \ep \xi^{(1)}_q(x,X) + O(\ep^2), \qquad E^\ep_q = E^{(0)}_q + \ep E^{(1)}_q + O(\ep^2).
\end{align*}
Since $\xi^\ep_q$ is $\ep^{-1}\Omega$-periodic, then the $\xi^{(j)}_q$'s are as well, for $j \in \{0,1\}$. The zeroth order in $\ep$ of $h^\ep_q \xi^\ep_q = E^\ep_q \xi^\ep_q$ is
\begin{align}\label{eq:first_order} 
\pa{\tfrac{1}{2} (-i\na_{x,q})^2 + v(x) - E^{(0)}_q} \xi^{(0)}_q(x,X) = 0
\end{align}
and the following order is
\begin{multline}\label{eq:second_order} 
\pa{\tfrac{1}{2} (-i\na_{x,q})^2 + v(x) - E^{(0)}_q} \xi^{(1)}_q(x,X) \\
= -\pa{(-i\nabla_{x,q}) \cdot (-i\nabla_X) + V(X) - E^{(1)}_q} \xi^{(0)}_q(x,X).
\end{multline}
From~\eqref{eq:first_order} we obtain that for $q$ in a neighborhood of $K$, there exist $\Omega$-periodic functions $\alpha_q^1,\alpha_q^2$ such that
\begin{align*}
\xi^{(0)}_q(x,X) = \sum_{a=1}^{2} \alpha_q^a(X) g^a_q(x),
\end{align*}
where $(g^a_q)_{1 \le a \le 2}$ are the Bloch eigenfunctions of $\tfrac{1}{2} (-i\na_q)^2 + v$ at the Fermi level. Explicitly, $g^1_q=u^q_{m\fer}$ and $g^2_q=u^q_{m\fer+1}$. Then~\eqref{eq:second_order} becomes
\begin{multline*}
\pa{\tfrac{1}{2} (-i\na_{x,q})^2 + v(x) - E^{(0)}_q} \xi^{(1)}_q(x,X) \\
  = -\sum_{j=1}^{2} \pa{(-i\na_X \alpha_q^j)(X) \cdot (-i\na_{x,q} g^j_q)(x) + (V(X) - E^{(1)}_q)\alpha_q^j(X) g^j_q(x)}.
\end{multline*}
This equation has a solution if and only if the right-hand side of this last equation is orthogonal to $\Span (g^a_q)_{a \in \{1,2\}}$, that is, for all $X \in \Omega$,
\begin{align*}
 0 
   &=  \sum_{j=1}^{2} (-i\na_X \alpha_q^j)(X) \cdot\ps{g_q^a,  (-i\na_{x,q}) g_q^j}_{L^2\per(\Omega)} + (V(X) - E^{(1)}_q) \alpha_q^j(X)\delta_{aj}.
\end{align*}
Consequently, this system is equivalent to
\begin{align*}
	\pa{\slashed{D}_q - E^{(1)}_q} \alpha_q = 0,
\end{align*}
where $F_q := \pa{\ps{g_q^a, (-i\na_{x,q}) g_q^b}}_{1 \le a,b \le 2}$ and
\begin{align*}
	\slashed{D}_q := F_q \cdot (-i\na_X) + \1_{2 \times 2} \otimes V(X).
\end{align*}
Finally, we are interested in $q = K + \ep k$ so we can expand
\begin{align*}
g^b_{K + \ep k} = g^b_K + \ep \pa{\delta^+_K g^b_\cdot}(k) + O(\ep^2)
\end{align*}
where $\pa{\delta^+_x f}y$ is the directional derivative of a function $f$ at $x$ in direction $y$, and the choice of $g^b_K \in \Span \pa{w_1,w_2}$ depends on $k$, it is such that we can do the expansion. 
We now expand
\begin{align*}
	\xi^{(0)}_q(x,X) = \sum_{a=1}^{2} \alpha_q^a(X) \pa{g^a_K + \ep \pa{\delta_K^+ g_\cdot^a}(k)} + O(\ep^2)
\end{align*}
Let us denote by $P_\cF$ the orthogonal projection onto the subspace~\eqref{eq:proj_space}. The error between the exact and the approximate eigenfunctions is bounded (up to a constant independent of $\ep$) by $\nor{P^\perp_\cF \xi_q^\ep}{} $, see~\cite{GarSta24}. We focus on the first order error in $\ep$. In the case where $\cF = \cF_0$, we have
\begin{align*}
	\nor{P^\perp_{\cF_0} \xi_q^\ep}{L^2\per(\Omega)} \le \ep \nor{P^\perp_{\cF_0} \pa{\xi_q^{(1)} + \sum_{a=1}^{2} \alpha_q^a \pa{\delta_K^+ g_\cdot^a}(k)}}{L^2\per(\Omega)} + O(\ep^2)
\end{align*}
while in the case where $\cF = \cF_1$, since $\pa{\delta_K^+ g_\cdot^a}(k) \in \Span \cF_1$ we have that $P^\perp_{\cF_1} \pa{\delta_K^+ g_\cdot^a}(k) = 0$ and
\begin{align*}
	\nor{P^\perp_{\cF_1} \xi_q^\ep}{L^2\per(\Omega)} \le \ep \nor{P^\perp_{\cF_1} \xi_q^{(1)} }{L^2\per(\Omega)} + O(\ep^2).
\end{align*}
Moreover, since $\cF_0 \subset \cF_1$, we have that for any function $f$, $\nor{P^\perp_{\cF_1} f}{L^2\per(\Omega)} \le \nor{P^\perp_{\cF_0} f}{L^2\per(\Omega)}$. We see that taking a family containing the first derivatives removes the term $\pa{\delta_K^+ g_\cdot^a}(k)$ and leads one to expect that the effective operator will be more accurate.


\section*{Research data management}
The Julia~\cite{BezEdeKar17} code enabling to produce the figures of this document can be found on the Github repository
\begin{center}
\url{https://github.com/lgarrigue/superlattice_graphene}
\end{center}

\section*{Acknowledgement}
I thank Eric Cancès for a useful discussion.

\appendix

\section{Formal degenerate perturbation theory}%
\label{sec:Degenerate perturbation theory}

Here we review first order degenerate perturbation theory when the degeneracy is lifted at first order and when the degeneracy is exactly $2$. We are only interested in the formal computations and will not be precise about the rigorous assumptions that are needed, see, for instance,~\cite{Kato,ReeSim4} for this purpose.

\subsection{The problem}%
\label{sub:The problem}

Consider some self-adjoint operators $A^0,A^1,A^2$ of a Hilbert space $\cH$. Take, for any $s \in \R$, a Hamiltonian
\begin{align*}
A(s) = A^{0} + s A^{1} + s^2 A^{2}
\end{align*}
and assume that there exists an eigenvalue $E^0 \in \R$ of $A^0$ such that $\dim \Ker (A^0 - E^0) = 2$. Then there exists two branch eigenstates $s \mapsto U_n(s) \in \cH$ of $A(s)$ for $n \in \{-,+\}$, analytic in $s \in [0,s_0[$ for some $s_0 >0$, such that $\Ker (A^0 - E^0) = \Span \pa{U_-(0),U_+(0)}$. The corresponding eigenvalues are denoted by $E_n(s)$. We expand in the so-called Rayleigh-Schrödinger series
\begin{align*}
U_n(s) = U_n^0 + s U_n^1 + s^2 U_n^2 + O(s^3), \qquad E_n(s) = E_n^0 + s E_n^1 + s^2 E_n^2 + O(s^3).
 \end{align*}
The goal is to compute $U_n^0$ and $U_n^1$ when the degeneracy is lifted at first order. This problem is well-known in the physics literature, see for instance~\cite{Kato,Hir69}. We provide the computations here for completeness.

\subsection{The reasoning}
We consider intermediate normalization, hence $U_n^k \perp U_n^{0}$ for any $k \in \N \backslash \{0\}$, but we do not necessarily have $\norm{U_n(s)} = 1$, see~\cite[Appendix A]{GarSta24} for more details on this procedure. Nevertheless, as presented in~\cite[Lemma C.1]{GarSta24}, perturbation vectors from intermediate normalization are equal to perturbation vectors from unit normalization at zeroth and first order, so the conclusion will not change.

We develop at the first 3 orders the eigenvalue equations
\begin{align*}
\bpa{A(s) - E_n(s)} U_n(s) = 0
\end{align*}
and obtain
\begin{align}
\bpa{A^{0} - E^0} U_n^{0} &= 0, \label{eq:first_schro} \\
\bpa{A^{1} - E_n^{1}} U_n^{0} + \bpa{A^{0} - E^0} U_n^{1} &= 0,\label{eq:second_schro} \\
\bpa{A^{2} - E_n^{2}} U_n^{0} +\bpa{A^{1} - E_n^{1}} U_n^{1} + \bpa{A^{0} - E^0} U_n^{2} &= 0.\label{eq:third_schro}
\end{align}
Let us denote by $P$ the orthogonal projection onto $\Ker \bpa{A^{0}-E^{0}}$, and $P_\perp := 1 - P$. Let us define the pseudoinverse
\begin{align*}
	G := \left\{
	\begin{array}{ll}
	\pa{(E^0 - A^0)_{\mkern 1mu \vrule height 2ex\mkern2mu P_\perp \cH \rightarrow P_\perp \cH}}^{-1} & \mbox{on } P_\perp \cH \\
	0 & \mbox{on } P \cH,
	\end{array}
	\right.
\end{align*}
extended by linearity on $\cH$. Applying $G$ to~\eqref{eq:second_schro} yields
 \begin{align}\label{eq:Pperp_psi1}
 P_\perp U_n^{1} = G \bpa{A^{1}- E_n^{1}} U_n^{0} = G A^1 U_n^0,
\end{align}
while applying $P$ gives
\begin{align*}
P A^{1} P U_n^{0} = E_n^{1} U_n^{0},
\end{align*}
that is $\bpa{E^{1}_n,U_n^0}$ is an eigenmode of the matrix $P A^{1} P$, in a basis we are free to choose, this is how we get $E^{1}_n$. Then we define $P_n^{1} := P \bpa{1-\widetilde{P}^{1}_n}$ where $\widetilde{P}^{1}_n$ is the orthogonal projection onto $\Ker \bpa{P \pa{A^{1}  - E^{1}_n} P_{\mkern 1mu \vrule height 2ex\mkern2mu P \cH}} \subset P \cH$, hence $P_n^1$ is the orthogonal projection onto $P \cH \cap \pa{\Ker \bpa{P \pa{A^{1} - E^{1}_n} P}}^\perp$. For any vector $u \in \cH$, let us denote by $P_u$ the orthogonal projection onto $\Span u$. We assume that degeneracies are lifted at first order, i.e. that $(P A^1 P)_{\mkern 1mu \vrule height 2ex\mkern2mu P \cH \rightarrow P \cH}$ has no degenerate eigenvalue, then we have $P^1_n = P_{U_{-n}^{0}}$. We apply $P_n^{1}$ to~\eqref{eq:third_schro} and obtain
\begin{align*}
0 &= P_n^{1} \bpa{A^{1} - E^{1}_n} U_n^{1} + P_n^{1} A^{2} U_n^{0} \\
& \underset{\substack{P + P_\perp = 1}}{=} \; P_n^{1} \bpa{A^{1} - E^{1}_n}  P U_n^{1} +P_n^{1} \bpa{A^{1} - E^{1}_n}  P_\perp U_n^{1} + P_n^{1} A^{2} U_n^{0} \\
& \underset{\substack{P U_n^{1} = P_n^{1}U_n^{1}}}{=} \; P_n^{1} \bpa{A^{1} - E^{1}_n}  P_n^{1} U_n^{1} +P_n^{1} A^{1} P_\perp U_n^{1} + P_n^{1} A^{2} U_n^{0} \\
& \underset{\substack{~\eqref{eq:Pperp_psi1} }}{=} \;P_n^{1} \bpa{A^{1} - E^{1}_n}  P_n^{1} U_n^{1} +P_n^{1} A^{1} G A^{1} U_n^{0} + P_n^{1} A^{2} U_n^{0} \\
& = - P_n^{1} \bpa{E^{1}_n - A^{1} }  P_n^{1} U_n^{1} +P_n^{1} \pa{A^{2}  + A^{1} G A^{1}} U_n^{0},
\end{align*}
where we used that $U_n^{1} \perp U_n^{0}$ hence $P U_n^{1} = P_n^{1} U_n^{1}$. Finally, degeneracy is lifted at first order so the restriction $\pa{P_n^{1} \bpa{E^{1}_n - A^{1}}  P_n^{1}}_{\mkern 1mu \vrule height 2ex\mkern2mu \Span U^0_{-n} \rightarrow\Span U^0_{-n}}$ is invertible and its inverse is $\pa{E^1_n - E^1_{-n}}^{-1} P_{U^0_{-n}}$. We apply this operator in the last equation, yielding
\begin{align}\label{eq:Pu} 
P U_n^{1} 
&= \pa{E^1_n - E^1_{-n}}^{-1} \hspace{-0.1cm} P_{U^0_{-n}} \bpa{A^2 + A^{1} G A^{1}} U_n^{0} \nonumber \\
&=\hspace{-0.05cm} \f{\ps{U^0_{-n}, \bpa{A^2 + A^{1} G A^{1}} U_n^{0}}}{E^1_n - E^1_{-n}} U_{-n}^{0}.
\end{align}

\subsection{Conclusion}%
\label{sub:Conclusion general deg}

To conclude, we saw that $U_n^0$ are vectors which diagonalize the matrix $P A^1 P$, i.e. such that
\begin{align*}
P A^{1} P U_n^{0} = E_n^{1} U_n^{0}.
\end{align*}
Then assembling~\eqref{eq:Pperp_psi1} and~\eqref{eq:Pu}, we can deduce the first order
\begin{align}\label{eq:exc_psi1} 
U_n^{1} = G A^{1} U_n^{0}+  \f{\ps{U^0_{-n}, \bpa{A^2 + A^{1} G A^{1}} U_n^{0}}}{E^1_n - E^1_{-n}} U_{-n}^{0}.
\end{align}
Using the scalar products of~\eqref{eq:second_schro} and~\eqref{eq:third_schro} against $U_n^0$, we also obtain
\begin{align*}
E_n^1 &= \ps{U_n^0, A^1 U_n^0}, \\
E_n^2 &= \ps{U_n^0, A^2 U_n^0} + \ps{U^0_n , A^1 U^1_{n}} = \ps{U_n^0, \pa{A^2 + A^1 G A^1} U_n^0}
\end{align*}

\section{Degenerate perturbation theory for our problem}%
\label{sec:Degenerate perturbation theory our problem}

\subsection{Directional derivatives}%
\label{sub:Directional derivatives}

We recall the definition of directional derivatives. Let us take two open normed spaces $B$ and $C$, and a map $f : B \to C$. We consider $x ,y \in B$ such that there exists $s_0 > 0$ with $x+sy \in B$ uniformly in $s \in [0,s_0[$. If the limit
\begin{align*}
\pa{\delta^+_x f}(y) := \mylim{s \to 0^+} \f{f(x+sy)-f(x)}{s}
\end{align*}
exists, then we say that this is the directional derivative of $f$ in the direction $y$. Similarly, we can define the directional derivatives of order $\ell$ of $f$ at $x$ in the direction $y$ as $\bpa{\delta^{+,\ell}_x f}(y) := \pa{\tfrac{\d^\ell}{\d s^\ell} \vp_y}_{\mkern 1mu \vrule height 2ex\mkern2mu s = 0}$, where $\vp_y(s) := f(x+sy)$. We then have the Taylor expansion, for instance up to order 2,
\begin{align*}
f(x + sy) = f(x) + s \bpa{\delta^{+,1}_x f}(y) + \tfrac 12  s^2 \bpa{\delta^{+,2}_x f}(y) + O(s^3).
\end{align*}

\subsection{$k\cdot p$ perturbation method}%
\label{sub:kp perturbation}

Perturbation theory in the Bloch momentum variable $k$ is called the $k\cdot p$ method in physics literature~\cite{LutKoh55,Kittel,PetCar10,Harrison12}. We recall it in this section.

Take $k \in \R^2 \backslash \{(0,0)\}$, the goal of this section is to differentiate the eigenvectors of $h_{K + k}$ with respect to $k$, that is in the direction given by $k$. We define the operators
\begin{align*}
A(s) := h_{K + s k} = \tfrac 12 \bpa{-i\na_{K + s k}}^2 + v = A^0 + s A^1 + s^2 A^2,
\end{align*}
where
\begin{align}\label{eq:def_A} 
	A^0 := \tfrac 12 \pa{-i\na_K }^2 + v, \qquad A^1 := \ab{k}  \der_k = k \cdot (-i\na_K ), \qquad A^2 := \tfrac{\ab{k}^2}{2} .
\end{align}
Similarly as in Section~\ref{sec:Degenerate perturbation theory}, let us assume that we start the perturbation from a level $E^0$ in which $\Ker (A^0 - E^0)$ is of dimension 2, as in the graphene case, and that the degeneracy is lifted at first order, in any direction. We denote by $\bbU_{\eta}(q)$ the two eigenvectors of $h_q$, $\eta \in \{-,+\}$, defined for $q$ locally around $q= K$, such that their eigenvalues at $q = K$ are $E^0$, and such that the maps $U_{\eta,k}(s) := \bbU_{\eta}(K+sk)$ are smooth for $s$ in a neighborhood of $0^+$. The corresponding eigenvalues of $\bbU_\eta(q)$ are denoted by $e_\eta(q)$, and $E_{\eta,k}(s) := e_\eta(K + sk)$.

For any $\ell \in \N \cup \{0\}$, we define the directional derivatives of $\bbU$ and $e$ at $K$ in the direction $k$,
\begin{align*}
U_{\eta,k}^\ell := \bpa{\delta^{+,\ell}_K \bbU_\eta}(k) = \pa{\tfrac{\d^\ell}{\d s^\ell} U_{\eta,k}}_{\mkern 1mu \vrule height 2ex\mkern2mu s = 0}, \quad 
E_{\eta,k}^\ell := \bpa{\delta^{+,\ell}_K e_\eta}(k) = \pa{\tfrac{\d^\ell}{\d s^\ell} E_{\eta,k}}_{\mkern 1mu \vrule height 2ex\mkern2mu s = 0}
\end{align*}
for $\eta \in \{-,+\}$ and $\ell \in \N\cup\{0\}$. In particular,
\begin{align*}
U_{\eta,k}(s) = U_{\eta,k}^{0} + s U_{\eta,k}^{1}+O(s^2), \qquad E_{\eta,k}(s) = E^0 + s E_{\eta,k}^1 + O(s^2).
\end{align*}
We define $k_\C := k_1 + i k_2 \in \C$. In the basis $\w_1, \w_2$ (defined in~\eqref{eq:w1w2}), $(P A^{1} P)_{\mkern 1mu \vrule height 2ex\mkern2mu P\cH \to P \cH}$ can be rewritten
 \begin{align*}
 \pa{\ps{\w_i, A^{1} \w_j}}_{1\le i,j \le 2} = k \cdot \pa{\ps{\w_i, \pa{- i\na_K } \w_j}}_{1\le i,j \le 2} \underset{\substack{\eqref{eq:vF}}}{=} \;  v\fer \sigma \cdot k
 \end{align*}
and its two eigenvalues are $E^{1}_{\eta,k} = \eta v\fer \ab{k} $.
 
We recall that for $\kappa \in \C$, the matrix $M := \mat{0 & \kappa \\ \overline{\kappa} & 0}$ has eigenvectors $g_\eta = \mat{\eta \f{\kappa}{\ab{\kappa}} & 1}^T$, with $M g_\eta = \eta \ab{\kappa} g_\eta, \eta \in \{\pm\}$. Hence from Section~\ref{sub:Conclusion general deg} we deduce that
\begin{align*}
\uz{\eta} =   \f{ 1}{\sqrt 2} \pa{\eta\f{\overline{k_\C}}{\ab{k}}\w_1 + \w_2}
\end{align*}
and $P A^{1} P \uz{\eta} = E^1_{\eta,k} \uz{\eta}$. The family $\uz{\eta}$ is not continuous in $k$ around $k=0$. Moreover, from~\eqref{eq:exc_psi1} we have
\begin{align*}
\uo{\eta} 
= \ab{k} R \der_k \uz{\eta}  + \f{\eta \ab{k} }{2v\fer}  \ps{\uz{-\eta}, \der_k R \der_k \uz{\eta}}\uz{-\eta}
\end{align*}
and
\begin{align*}
E_{\eta,k}^1 = \ab{k} \ps{U_{\eta,k}^0, \der_k U_{\eta,k}^0}, \qquad E_{\eta,k}^2 &= \f{\ab{k}^2 }{2}  + \ab{k}^2 \ps{U^0_{\eta,k},  \der_k R \der_k U_{\eta,k}^0}.
\end{align*}
Finally
\begin{align*}
\bbU_{\eta}(K+k) &= \bbU_{\eta} \pa{K+ \ab{k} \tfrac{k}{\ab{k}}} = U^0_{\eta,\f{k}{\ab{k}}} + \ab{k} U^1_{\eta,\f{k}{\ab{k} } } + O(\ab{k}^2) \\
 &= U^0_{\eta,k} + U^1_{\eta,k } + O(\ab{k}^2).
\end{align*}

\subsection{Building reduced spaces from perturbation theory}%
\label{sub:Building reduced spaces from perturbation theory}

Avoiding the use of degenerate perturbation theory with eigenstates, we can use perturbation theory with density matrices, which is much simpler when the goal is only to generate the space spanned by perturbative eigenvectors. See~\cite{McWeeny62} for the computation of coefficients in density matrix pertubation theory, and~\cite[Section 9]{GarSta24} for a mathematical review.

Let us denote by $\Gamma_k(s)$ the orthogonal projection onto the space
\begin{align*}
\Span \pa{U_{\eta,k}(s)}_{\eta \in \{-,+\}}.
\end{align*}
We expand it as $\Gamma_k(s) = \sum_{n=0}^{+\infty} s^n \Gamma_k^n$. We define
\begin{align*}
\cB_{\ell,k} &:= \Span \pa{\pa{U^m_{\eta,k}}_{\eta \in \{-,+\}}^{0 \le m \le \ell}} 
\end{align*}
which is the space spanned by the derivatives of the eigenvectors in the direction $k$. We are searching for two kinds of families, with the minimal number of vectors in each, $\cF_{\ell,k}$ and $\cF_\ell$ such that
\begin{align}\label{eq:spanF} 
	\Span \cF_{\ell,k} = \cB_{\ell,\f{k}{\ab{k} } },\qquad\qquad \Span \cF_\ell = \underset{\substack{k \in \cB \backslash \{0\}}}{\bigoplus} \;\cB_{\ell,\f{k}{\ab{k} } }.
\end{align}
Remark that $\Span \cF_{\ell,k} \subset \Span \cF_\ell$. Moreover, we have
\begin{align*}
\bbU_{\eta}(K+k) = U + O(\ab{k}^{\ell + 1}) \qquad \text{where } U \in \Span\cF_{\ell,k} \subset \Span \cF_\ell.
\end{align*}
The first kind of family will span the ``perturbative space'' in the direction $k/\ab{k} $ up to order $\ell$, and the second one will span it uniformly in $k$. From~\cite[Lemma 4.2]{GarSta24} we have
\begin{align}\label{eq:gam} 
\cB_{\ell,k} =\Span\pa{\pa{\Gamma^n_k \w_a}_{a \in \{1,2\}}^{0 \le n \le \ell}}.
\end{align}
Now, since the derivatives $\Gamma^n$ are easy to compute, we can obtain the space written on the left of the last equation. Let us do this up to second order. We use the notations~\eqref{eq:def_A}. From~\cite{McWeeny62} or~\cite[Section 9]{GarSta24} we have
\begin{align*}
	\Gamma^1_k&= R A^1 P + P A^1 R, \\
	\Gamma^2_k &= R A^1 P A^1 R - P A^1 R^2 A^1 P + R A^2 P + P A^2 R \\
		 & \qquad \qquad + R A^1 R A^1 P + P A^1 R A^1 R - R^2 A^1 P A^1 P - P A^1 P A^1 R^2,
\end{align*}
Remark that $\Gamma^0_k = P$, which we denote by $\Gamma^0$ since it does not depend on $k$. Hence,
\begin{align*}
\Gamma^1_k \w_a &= \ab{k} R \der_k \w_a, \\
\Gamma^2_k \w_a &= \ab{k}^2 \pa{  - P \der_k R^2 \der_k \w_a + (R \der_k)^2 \w_a - R^2 \der_k P \der_k \w_a}.
\end{align*}
We precise that $(R \der_k)^2  = \sum_{a,b \in \{1,2\}} \ab{k}^{-2} k_a k_b R(-i\partial_a)R(-i\partial_b)  $.

For the second order,
\begin{align*}
P \der_k R^2 \der_k \w_a \in \Span \cF_{1,k}, \qquad R^2 \der_k P \der_k \w_a &\in \Span \pa{R^2 \der_k \w_a}_{1 \le a \le 2}.
 \end{align*}
Thus for the families defined in~\eqref{eq:cFs}, we see that~\eqref{eq:spanF} and~\eqref{eq:gam} are satisfied. Following this method, one can obtain reduced spaces families $\cF_{\ell,k}$ and $\cF_\ell$ at all perturbative orders $\ell$. From~\cite[Lemma 4.2]{GarSta24}, except in exceptional cases where the corresponding family of vectors is not linearly independent (in which case we need less vectors), for any $\ell \in \N \cup \{0\}$ we need exactly $2(\ell+1)$ elements to build the family $\cF_{\ell,k}$ at order $\ell$. Hence $\ab{\cF_{\ell,k}} = N_{\cF} = 2(\ell +1)$.

\section{Schur reduction}%
\label{sec:Schur complement}

In this section, we quickly recall the principle of the Schur reduction~\cite{Schur18,Schur18b}, at a formal level. Consider two self-adjoint operators $H$ and $Y$ on a Hilbert space $\cH$, take an orthogonal projection $\PP$ onto a closed subspace of $\cH$, define $\PP^\perp := 1 - \PP$,  and assume that $\PP Y \PP^\perp = 0$. Let us denote by $(E,\phi)$ a non-degenerate solution of the generalized eigenvalue problem associated to $H$ and $Y$, so $H \phi = E Y \phi$, and $Y$ is called the mass operator. 
Then
\begin{align}\label{eq:init_} 
(H-EY) \PP \phi + (H-EY) \PP^\perp \phi = 0.
\end{align}
Applying $\PP^\perp$ gives $\PP^\perp (EY-H) \PP^\perp \phi =  \PP^\perp H \PP \phi$. Assume that there exists $\ep > 0$ such that
\begin{align*}
 \sigma \pa{\restr{\pa{\PP^\perp (EY-H) \PP^\perp}}{\PP^\perp \cH \rightarrow \PP^\perp \cH}} \cap \; ]-\ep,\ep[ = \varnothing,
\end{align*}
and define the partial inverse
\begin{align*}
	r(E) &:= 
\left\{
\begin{array}{ll}
0 & \mbox{on } \PP\cH , \\
\pa{ \restr{\pa{\PP^\perp (EY-H) \PP^\perp}}{\PP^\perp \cH \rightarrow \PP^\perp \cH}}^{-1} & \mbox{on } \PP^\perp \cH,
\end{array}
\right.
\end{align*}
extended by linearity on all of $\cH$. We have $r(E) \PP^\perp (EY-H) \PP^\perp = \PP^\perp$ so $\PP^\perp \phi = r(E) H \PP \phi$. Applying $\PP$ to~\eqref{eq:init_} and using $1 = \PP + \PP^\perp$ gives
\begin{align}\label{eq:schur_eq} 
\PP \pa{H + H r(E) H - EY } \PP \phi = 0, \qquad \qquad \phi = (1 + r(E) H) \PP \phi.
\end{align}
Because $r(E)$ depends on the unknown eigenvalue, the reduced equation is nonlinear in $E$. Once $(E,\PP\phi)$ has been obtained from $\PP\pa{H+Hr(E)H-EY}\PP\phi=0$, the full vector is reconstructed by applying $1+r(E)H$ to $\PP\phi$.

\subsection{Schur reduction applied to the effective operator}%
\label{sub:Schur complement applied here}

Then, with the same notations as in Section~\ref{sub:Reduction to a matrix}, we apply the Schur reduction with $H \leftarrow \bbH\elk$ (from~\ref{eq:effective_op_general}), $Y \leftarrow \cS$ and $\PP \leftarrow P_2$. From~\eqref{eq:M_L_S}, we have that
\begin{align*}
&P_2 \cM = \cM P_2 = 0,\quad  P_2^\perp \cM P_2^\perp = M,\quad  P_2 \cS P_2^\perp = P_2^\perp \cS P_2 = 0, \\
&\qquad \qquad P_2 \cL P_2^\perp = T,\quad  P_2^\perp \cL P_2 = T^*.
\end{align*}
Hence we have formally
\begin{align*}
r(E) &= P_2^\perp \pa{ \pa{E \cS - \bbH\elk}_{\mkern 1mu \vrule height 2ex\mkern2mu P_2^\perp \cH \rightarrow P_2^\perp \cH}}^{-1} P_2^\perp 
= -P_2^\perp \pa{\ep^{-1} \pa{M + O(\ep)}}^{-1} P_2^\perp \\
  & = -\ep P_2^\perp M^{-1} P_2^\perp + O(\ep^2).
\end{align*}
In our application, $M$ is indeed invertible. Hence
\begin{align*}
P_2 \bbH\elk r(E) \bbH\elk P_2 = -\ep T \cdot(-i\na_k) (M^{-1} T^*) \cdot(-i\na_k) + O(\ep^2).
\end{align*}
Combining this expansion with the $P_2$-block of $\bbH\elk$ yields the effective operator~\eqref{eq:schur_operator}. To reconstruct an approximation of the eigenvector of the original operator, one needs to form the quantity
\begin{align*}
	(1+r(E) \bbH\elk) P_2 &= P_2 - \ep \mat{0 & 0 \\ 0 & M^{-1}}\bbH\elk P_2 + O(\ep^2) \\
			   &= (1 - \ep P_2^\perp M^{-1} T^* ) P_2 + O(\ep^2).
\end{align*}
Again here the reader should not take this $O(\ep^2)$ as a spectral estimate, it is only a formal computation.

\section{Literature}%
\label{sec:Literature}

We give here a detailed account of the mathematical results most directly connected with our approximation. Throughout this subsection, $x$ denotes the microscopic variable and $y=\ep x$ the macroscopic variable. When a cited paper uses a different coordinate convention, we give its unitarily equivalent rescaled Hamiltonian in the variables $(x,y)$.

Buslaev~\cite{Buslaev87}. In the Schrödinger specialization of the general periodic-coefficient equations treated there, the operator is $H_\ep=-\frac{d^2}{dx^2}+p(x)+q(\ep x)$, where $p(x+a)=p(x)$, and the equation is $(H_\ep-E)\psi=0$. The paper formulates the construction for pseudodifferential equations with periodic coefficients. The author constructs WKB-type asymptotic solutions transported along the corresponding Hamiltonian flow and discusses turning points and transitions between bands. It concerns asymptotic solutions rather than a uniform spectral approximation by a finite matrix operator.

Guillot--Ralston--Trubowitz~\cite{GuiRalTru88}. With $y=\ep x$ and in the symmetric gauge for a constant magnetic field $\omega$, the rescaled operator is $H_\ep=\left(-i\na_x+\frac12\omega\mathbin{\times}y\right)^2 + v(x)$. The authors construct semiclassical quasimodes for a slowly perturbed periodic Schrödinger operator near a non-degenerate stable periodic orbit of the effective classical Hamiltonian. The quasimodes are concentrated in a tube of radius of order $\ep^{1/2}$ around the projected orbit, their energies have asymptotic expansions, and the residual is bounded by a prescribed power of $\ep$ allowed by the stability order.

Gérard--Martinez--Sjöstrand~\cite{GerMarSjo91}. They consider $P_\ep=P(y,x,\zeta)$, where $y=\ep x$, $\zeta=D_x+A(y)$, $D_x=-i\na_x$, and the coefficients of the elliptic differential operator $P(y,x,\eta)$ are periodic in $x$. The main Schrödinger example is~\eqref{eq:GMS_ham}. This is the closest general operator-theoretic framework to the present problem. By a Grushin-problem reduction, the spectral problem for a periodic Schrödinger operator with slowly varying electric and magnetic perturbations, in a fixed energy window containing finitely many Bloch bands, is reduced to a finite system of semiclassical pseudodifferential operators on the dual torus. The effective Hamiltonian has an operator-valued symbol with a complete asymptotic expansion. The construction is spectral and valid to arbitrary asymptotic order, but the effective symbol remains abstract. For our weaker perturbation $\ep V(\ep x)$, one must additionally Taylor-expand that symbol in the amplitude $\ep$.

Dimassi~\cite{Dimassi93}. The operator is~\eqref{eq:GMS_ham} where $v$ is lattice-periodic. The author applies the effective-Hamiltonian method to slow perturbations of a periodic Schrödinger operator and proves complete asymptotic expansions for spectral trace quantities localized in energy. The coefficients are obtained from the semiclassical effective symbol.

Markowich--Mauser--Poupaud~\cite{MarMauPou94}. The microscopic Hamiltonian is $H_\hbar=-\frac{\hbar^2}{2m}\Delta_x+v(x)$, where $v$ is periodic, and the evolution is $i\hbar\partial_t\psi=H_\hbar\psi$. No macroscopic spatial variable occurs in this model. The limit studied is the semiclassical limit of its band-Wigner representation. Using Wigner functions and Bloch decomposition, these authors derive semiclassical transport equations for electrons in a periodic potential. The theorem concerns weak limits of quadratic observables of time-dependent solutions.

Poupaud--Ringhofer~\cite{PouRin96}. With lattice scale $\ep$, and scaled Planck constant $h$, the rescaled electron Hamiltonian is $H_{\ep,h}=-\frac{h^2}{2m\ep^2}\Delta_x+v(x)+V(\ep x)$ and the evolution is $ih\partial_t\psi_{\ep,h}=H_{\ep,h}\psi_{\ep,h}$. Two small parameters are considered, one for the lattice scale and one for Planck's constant. When both the semiclassical and scale-separation limits are taken, the authors obtain a Vlasov equation. When the semiclassical parameter is fixed and the lattice period tends to zero, they obtain an effective Schrödinger equation. In both limits the free mass is replaced by the Hessian of the relevant Bloch energy, i.e. the effective-mass tensor. The result is therefore a single-band, non-degenerate effective-mass theorem, whereas our Dirac point requires a coupled two-band or multi-band model.


Dimassi--Guillot--Ralston~\cite{DimGuiRal02}. Their time-dependent magnetic Hamiltonian is $H_\ep(t)=\left(-i\na_x+\frac12\omega\mathbin{\times}x+A(\ep x,\ep t)\right)^2+v(x)+V(\ep x,\ep t)$, and the evolution is $-i\partial_tu=H_\ep(t)u$, where $v$ is periodic and $\omega$ is the constant background magnetic field. For a periodic electric potential and a constant magnetic field perturbed by slowly varying electric and magnetic fields, the authors explicitly construct wave packets localized near semiclassical trajectories in an isolated magnetic Bloch band. This extends band dynamics to magnetic perturbations but still relies on isolation of the selected band.

Panati--Spohn--Teufel~\cite{PanSpoTeu03}. The exact Hamiltonian in their units is~\eqref{eq:GMS_ham} and the evolution is $i\ep\partial_t\psi=H_\ep\psi$, with $v$ lattice-periodic. For an isolated family of Bloch bands subject to slowly varying electric and magnetic potentials, the authors construct an almost-invariant subspace and a unitary intertwiner to a reference space. They prove that an effective Hamiltonian governs the quantum dynamics in this subspace to every order in $\ep$. Its principal symbol is the Peierls substitution, and the first correction, containing geometric quantities such as Berry terms and magnetic moments, is computed explicitly.

Carles--Sparber~\cite{CarSpa12}. Their rescaled problem is $i\ep\partial_t\psi^\ep=H_\ep\psi^\ep$ with $H_\ep$ given by~\eqref{eq:GMS_ham} where $A = 0$. They construct semiclassical wave packets concentrated along the Peierls-substituted Hamiltonian flow. The envelope satisfies a Schrödinger equation with a time-dependent effective mass, and the approximation is justified up to Ehrenfest times. The argument assumes a smooth isolated Bloch band and hence does not cover the loss of smooth eigenvectors at a Dirac point.

Fefferman--Weinstein~\cite{FefWei14}. The microscopic operator is $H=-\Delta_x+v(x)$ and the evolution is $i\partial_t\psi=H\psi$, where $V$ is a honeycomb lattice potential. The slowly modulated initial data are $\ep\alpha_j(\ep x)\Phi_j(x)$. For initial wave packets spectrally concentrated near a honeycomb Dirac point, the authors prove that the large but finite time evolution of the continuum Schrödinger equation is governed, to leading order with a controlled remainder, by a two-dimensional massless Dirac system for the macroscopic envelope.


Cancès--Garrigue--Gontier~\cite{CanGarGon23}. The rescaled highly oscillatory periodic operator is $-\Delta_x+\ep v(\ep x,x)+\ep^2 V(\ep x)$, with $\ep\in\mathbb N^{-1}$. The function $v$ is periodic in both variables and has zero mean in the microscopic variable. The authors rigorously construct first- and second-order correctors. They prove estimates for the resolvent equation, including degenerate eigenvalues, obtain eigenvalue and eigenfunction corrections, and control integrated-density-of-states observables. This shows what a genuine second-order spectral homogenization theorem entails, but its high-amplitude scaling is different from the weak, slow perturbation in~\eqref{eq:exact_ep_inv_omega}.

Cancès--Garrigue--Gontier~\cite{CanGarGon23b}. The starting point is $-\frac12\Delta_{x,z}+V_{d,\theta}^{(2)}(x,z)$, with $V_{d,\theta}^{(2)}=U_{d,\theta}V+U_{d,\theta}^{-1}V+V_{\rm int,d}(z)$, where $x$ is the microscopic in-plane variable, $z$ is the transverse variable, and the two unitary pull-backs translate the layers by $\pm d/2$ and rotate them by $\mp\theta/2$. The paper derives a moiré-scale continuum model for twisted bilayer graphene directly from a Kohn--Sham Hamiltonian by a variational ansatz and asymptotic expansion, without passing through a tight-binding model. The derivation produces a Bistritzer--MacDonald-type operator.

Bal--Cazeaux--Massatt--Quinn~\cite{BalCazMas26}. Writing $y=\ep x$, their exact general family is $H_\ep=\operatorname{Op}^{\rm W}(a_\ep)$ with $a_\ep(x,\xi)=a(y,\xi;\ep)$, where $a$ is a Hermitian matrix-valued tight-binding symbol. A finite-range model has $a$ equal to a finite sum of hopping symbols $\alpha_j(y)e^{iv_j\cdot\xi}$ together with their adjoints, it is this symbol that is Taylor-expanded at a degenerate point $(E,K)$. For tight-binding Hamiltonians with slowly varying coefficients and wave packets localized near a spectral degeneracy, the authors systematically Taylor-expand the matrix-valued symbol and prove long-time approximation by macroscopic continuum systems. Higher-order symbols give correspondingly higher-order dynamical errors, and the framework covers strained, twisted, single-layer, and multilayer graphene models, including position-dependent degeneracy points. 

Meng~\cite{Meng26}. His two-scale family is $H_\ep=T\bigl(-i\na_x+ A(x,y)\bigr)+V(x,y)$, where either $T(-i\na+A)=(-i\na+A)^2$ or, in two dimensions, $T(-i\na+A)=\sigma\cdot(-i\na+A)+\sigma_3M^2$. Meng constructs localized normalized quasimodes from the Taylor expansion of the cell Hamiltonian near a distinguished point $(k_0,X_0)$. He uses correctors for building the eigenfunctions, which are built from Taylor expansions in momentum of the eigenfunctions. 

To summarize, the general spectral reduction in~\cite{GerMarSjo91}, the dynamical Dirac reduction in~\cite{FefWei14}, the reduced-space estimates in~\cite{GarSta24}, and the recent higher-order constructions in~\cite{BalCazMas26,Meng26} are the closest rigorous results to different parts of our method. None of them, in its stated setting, proves an error bound for the individual supercell eigenvalues and eigenvectors produced by the finite $N_{\cF}\times N_{\cF}$ operators of the present paper.

\bibliographystyle{siam}
\bibliography{superlattice_graphene}
\end{document}